\documentclass[acmsmall,screen]{acmart}
\setcopyright{rightsretained}
\acmDOI{10.1145/3704865}
\acmYear{2025}
\acmJournal{PACMPL}
\acmVolume{9}
\acmNumber{POPL}
\acmArticle{29}
\acmMonth{1}
\title{TensorRight: Automated Verification of Tensor Graph Rewrites}         %% [Short Title] is optional;
\author{Jai Arora}
\orcid{0009-0004-7759-481X}             %% \orcid is optional
\affiliation{
  \institution{University of Illinois Urbana-Champaign}            %% \institution is required
  \country{USA}                    %% \country is recommended
}
\email{jaia3@illinois.edu}          %% \email is recommended

\author{Sirui Lu}
\orcid{0000-0003-2757-7603}             %% \orcid is optional
\affiliation{
  \institution{University of Washington}            %% \institution is required
  \country{USA}                    %% \country is recommended
}
\email{siruilu@cs.washington.edu}          %% \email is recommended

\author{Devansh Jain}
\orcid{0009-0006-1442-1502}             %% \orcid is optional
\affiliation{
  \institution{University of Illinois Urbana-Champaign}            %% \institution is required
  \country{USA}                    %% \country is recommended
}
\email{devansh9@illinois.edu}          %% \email is recommended

\author{Tianfan Xu}
\orcid{0000-0003-2528-0507}             %% \orcid is optional
\affiliation{
  \institution{University of Illinois Urbana-Champaign}            %% \institution is required
  \country{USA}                    %% \country is recommended
}
\email{tianfan3@illinois.edu}          %% \email is recommended

\author{Farzin Houshmand}
\orcid{0000-0001-8516-2401}             %% \orcid is optional
\affiliation{
  \institution{Google}            %% \institution is required
  \country{USA}                    %% \country is recommended
}
\email{farzinh@google.com}          %% \email is recommended

\author{Phitchaya Mangpo Phothilimthana}
\orcid{0000-0003-3492-3690}             %% \orcid is optional
\affiliation{
  \institution{Google DeepMind}            %% \institution is required
  \country{USA}                    %% \country is recommended
}
\email{mangpo@google.com}          %% \email is recommended

\author{Mohsen Lesani}
\orcid{0000-0002-3165-2322}             %% \orcid is optional
\affiliation{
  \institution{University of California, Santa Cruz}            %% \institution is required
  \country{USA}                    %% \country is recommended
}
\email{mlesani@ucsc.edu}          %% \email is recommended

\author{Praveen Narayanan}
\orcid{0000-0003-3331-9627}             %% \orcid is optional
\affiliation{
  \institution{Google}            %% \institution is required
  \country{USA}                    %% \country is recommended
}
\email{pravnar@google.com}          %% \email is recommended

\author{Karthik Srinivasa Murthy}
\orcid{0000-0002-6063-9653}             %% \orcid is optional
\affiliation{
  \institution{Google}            %% \institution is required
  \country{USA}                    %% \country is recommended
}
\email{ksmurthy@google.com}          %% \email is recommended

\author{Rastislav Bodik}
\orcid{0000-0001-6639-1647}             %% \orcid is optional
\affiliation{
  \institution{Google DeepMind}            %% \institution is required
  \country{USA}                    %% \country is recommended
}
\email{rastislavb@google.com}          %% \email is recommended

\author{Amit Sabne}
\orcid{0000-0002-2179-0078}             %% \orcid is optional
\affiliation{
  \institution{Google}            %% \institution is required
  \country{USA}                    %% \country is recommended
}
\email{asabne@google.com}          %% \email is recommended

\author{Charith Mendis}
\orcid{0000-0002-8140-2321}             %% \orcid is optional
\affiliation{
  \institution{University of Illinois Urbana-Champaign}            %% \institution is required
  \country{USA}                    %% \country is recommended
}
\email{charithm@illinois.edu}          %% \email is recommended

\renewcommand{\shortauthors}{Arora et al.}



\usepackage{booktabs}   %% For formal tables:
\usepackage{subcaption} %% For complex figures with subfigures/subcaptions
\usepackage[capitalize]{cleveref}
\usepackage{enumitem}
\usepackage{xcolor}
\usepackage{listings}
\usepackage{enumerate}
\usepackage{siunitx}
\usepackage{tabularx,multirow}

\usepackage[utf8]{inputenc}
\usepackage[T1]{fontenc}
\usepackage{microtype}

\newcommand*{\funname}[1]{\ifmmode\mathit{#1}\else\textit{#1}\fi}

\theoremstyle{definition}
\newtheorem{definition}{Definition}
\newtheorem{theorem}{Theorem}
\newtheorem{lemma}[theorem]{Lemma}

\newcommand{\project}[0]{\textsc{TensorRight}}

\usepackage{xparse}
\usepackage{wrapfig}
\usepackage{relsize}
\usepackage{tikz}
\usetikzlibrary{positioning}
\usepackage{mathbbol}
\usepackage{tcolorbox}
\usepackage{bm}
\usepackage{semantic}
\usepackage{amsmath}
\usepackage{float}
\usepackage[ruled, linesnumbered]{algorithm2e}
\usepackage{subcaption}
\usepackage{pifont}% http://ctan.org/pkg/pifont
\usepackage{soul}
\newcommand{\icon}[1]{\raisebox{.5pt}{\textcircled{\raisebox{-.7pt}{#1}}}}

\usepackage{mathpartir}

\makeatletter
\newcommand{\rulelabelp}[2]{%
   \protected@write \@auxout {}{\string \newlabel {#1}{{\textsc{#2}}{\thepage}{#2}{#1}{}} }%
   \hypertarget{#1}{}
}
\makeatother

\newcommand{\mysetminusD}{\hbox{\tikz{\draw[line width=0.6pt,line cap=round] (3pt,0) -- (0,6pt);}}}
\newcommand{\mysetminusT}{\mysetminusD}
\newcommand{\mysetminusS}{\hbox{\tikz{\draw[line width=0.45pt,line cap=round] (2pt,0) -- (0,4pt);}}}
\newcommand{\mysetminusSS}{\hbox{\tikz{\draw[line width=0.4pt,line cap=round] (1.5pt,0) -- (0,3pt);}}}

\newcommand{\mysetminus}{\mathbin{\mathchoice{\mysetminusD}{\mysetminusT}{\mysetminusS}{\mysetminusSS}}}

\def\<{\langle}
\def\>{\rangle}
\def\llet{\mathsf{let \ }}
\def\lin{\mathsf{\ in \ }}
\def\lif{\mathsf{if}}
\def\lthen{\mathsf{then}}
\def\lelse{\mathsf{else}}
\def\fmap{\mathsf{fmap}}
\def\fold{\mathsf{fold}}

\newcommand{\sem}[1]{\left\llbracket #1 \right\rrbracket}

\def\lland{\ \wedge \ }

\newcommand{\removetext}[1]{}
\newcommand{\dsl}[0]{\project{} DSL}
\newcommand{\frontend}[0]{\project{} Frontend}

\def\reshape{\mathsf{reshape}}
\def\slice{\mathsf{slice}}
\def\expand{\mathsf{expand}}
\def\tiota{\mathsf{iota}}
\def\transpose{\mathsf{transpose}}
\def\relabel{\mathsf{relabel}}
\def\dyslice{\mathsf{dy{\hbox{-}}slice}}
\def\dyupdateslice{\mathsf{dyup{\hbox{-}}slice}}
\def\reduce{\mathsf{reduce}}
\def\convbase{\mathsf{conv{\hbox{-}}base}}
\def\conv{\mathsf{conv}}
\def\tdot{\mathsf{dot}}
\def\pad{\mathsf{pad}}
\def\padlow{\mathsf{pad{\hbox{-}}low}}

\def\const{\mathsf{const}} 
\def\concat{\mathsf{concat}}
\def\binary{\mathsf{binary}}
\def\reverse{\mathsf{reverse}}
\def\bitcast{\mathsf{bitcast}}

\def\clamp{\mathsf{clamp}}
\def\select{\mathsf{select}}
\newcommand{\xla}[0]{\texttt{XLA}}
\newcommand{\xlahlo}[0]{\texttt{XLA-HLO}}

\newcommand{\df}{\stackrel{\text{df}}{=}}
\newcommand{\scalarf}{\mathsf{scalarf}}
\newcommand{\rclass}{\textsf{RClass}}
\newcommand{\rclasses}{\textsf{RClasses}}
\newcommand{\axes}{\textsf{Axes}}
\newcommand{\access}{\textsf{Access}}
\newcommand{\dtype}{\textsf{DType}}
\newcommand{\shape}{\textsf{Shape}}
\newcommand{\ind}{named}
\newcommand{\agg}{aggregated}
\newcommand{\indtrans}{index transformer}
\newcommand{\cond}{condition}
\newcommand{\valexp}{\mathsf{valid\mbox{-}expr}}
\newcommand{\rcrank}{\rclass\mbox{-}rank}
\newcommand{\vars}{vars}
\newcommand{\cproject}[3]{#1\vert^{#2}_{#3}}

\newcommand{\true}{\textsf{true}}
\newcommand{\false}{\textsf{false}}

\newcommand{\subst}[3]{#1\{#3/#2\}}

\newcommand{\IGNORE}[1]{}

\newcommand{\lstref}[1]{\hyperref[lst:#1]{Listing~\ref*{lst:#1}}}
\newcommand{\figref}[1]{\hyperref[fig:#1]{Fig.~\ref*{fig:#1}}}
\newcommand{\figlabel}[1]{\label{fig:#1}}
\newcommand{\secref}[1]{\hyperref[sec:#1]{\S\ref*{sec:#1}}}
\newcommand{\seclabel}[1]{\label{sec:#1}}
\newcommand{\tabref}[1]{\hyperref[tab:#1]{Table~\ref*{tab:#1}}}
\newcommand{\tablabel}[1]{\label{tab:#1}}
\newcommand{\ruleref}[1]{\hyperref[rule:#1]{Rule~\ref*{rule:#1}}}

\newcommand{\defref}[1]{\hyperref[def:#1]{Definition~\ref*{def:#1}}}
\newcommand{\deflabel}[1]{\label{def:#1}}
\newcommand{\lemmaref}[1]{\hyperref[lemma:#1]{Lemma~\ref*{lemma:#1}}}
\newcommand{\lemmalabel}[1]{\label{lemma:#1}}
\newcommand{\theoremref}[1]{\hyperref[theorem:#1]{Theorem~\ref*{theorem:#1}}}
\newcommand{\theoremlabel}[1]{\label{theorem:#1}}
\newcommand{\eqnref}[1]{\hyperref[eq:#1]{Equation~\ref*{eq:#1}}}
\newcommand{\eqnlabel}[1]{\label{eq:#1}}
\newcommand{\algref}[1]{\hyperref[alg:#1]{Algorithm~\ref*{alg:#1}}}
\newcommand{\alglabel}[1]{\label{alg:#1}}

\newcommand{\semref}[1]{\hyperref[rule:#1]{\textsc{#1}}}
\newcommand{\semlabel}[1]{\label{rule:#1}}
\newcommand{\dollar}{\mbox{\textdollar}}

\crefname{lstlisting}{listing}{listings}
\Crefname{lstlisting}{Listing}{Listings}

\definecolor{forestgreen}{rgb}{0.0, 0.27, 0.13}
\definecolor{cssgreen}{rgb}{0.0, 0.5, 0.0}
\definecolor{hookergreen}{rgb}{0.0, 0.44, 0.0}
\definecolor{navyblue}{rgb}{0.0, 0.0, 0.7}

\definecolor{mygreen}{HTML}{BEDB39}
\definecolor{myorange}{HTML}{FD7400}
\definecolor{mydarkgreen}{HTML}{29B895}
\definecolor{mydarkgreen2}{HTML}{1F8A70}
\definecolor{mydarkgreen3}{HTML}{104A3C}
\definecolor{myyellow}{HTML}{FFEA59}
\definecolor{myblue}{HTML}{004358}
\definecolor{myblue2}{HTML}{0088B2}
\definecolor{mygray}{HTML}{DFDFDF}

\lstdefinelanguage{haskell}{
 basicstyle=\ttfamily\footnotesize,
 keywords={do, newRClass, newRClass, newMap, newMaps, newTensor, Rewrite, intBinOp, boolBinOp, concatTensor, concatTensorList, compare, slice, pad, iota, dynamicSlice, dynamicUpdateSlice, reduce, broadcast, relabel, reshape, dot, conv, pointWiseCondition, elementWiseCondition, addPrecondition, addSIRelation, verifyDSL, precondition, siRelation, elementWiseArith},
 keywordstyle=\bfseries\color{mydarkgreen3},
 numbers=left,
 numbersep=1pt,
 stepnumber=1,
 ndkeywords={boolean, throw, import},
 ndkeywords={verify, assert, equal},
 ndkeywordstyle=\bfseries\color{magenta},
 identifierstyle=\color{black},
 sensitive=false,
 comment=[l]{\#},
 morecomment=[s]{/*}{*/},
 commentstyle=\color{brown}\ttfamily,
 stringstyle=\color{red}\ttfamily,
}

\lstdefinestyle{c}{
 basicstyle=\ttfamily\footnotesize,
 language=C,
 keywords={NULL, switch, int, str, float, void, const},
 keywordstyle=\bfseries\color{mydarkgreen3},
 ndkeywords={boolean, throw, import},
 ndkeywords={def, return, class, if ,else, while, do, else, true, false , catch, for, inline, __global__, __device__, __constant__, __register__, __shfl_sync, warp_size, acc, parallel},
 ndkeywordstyle=\bfseries\color{mydarkgreen3},
 identifierstyle=\color{black},
 sensitive=false,
 comment=[l]{\#},
 morecomment=[s]{/*}{*/},
 commentstyle=\color{mydarkgreen2}\ttfamily,
 stringstyle=\color{red}\ttfamily,
 moredelim=**[is][\color{myblue2}\bfseries]{`}{`},
 captionpos=b,
 flexiblecolumns,
 numbersep=1pt,
 keepspaces,
 escapechar=@,
 basicstyle=\ttfamily\scriptsize,
 belowskip=0.05cm,
 aboveskip=0.05cm,
 belowcaptionskip=0cm,
 abovecaptionskip=0cm,
}

\begin{document}

%% Abstract
%% Note: \begin{abstract}...\end{abstract} environment must come
%% before \maketitle command
\begin{abstract}
    Tensor compilers, essential for generating efficient code for deep learning models across various applications, employ tensor graph rewrites as one of the key optimizations.
    These rewrites optimize tensor computational graphs with the expectation of preserving semantics for tensors of arbitrary rank and size.
    Despite this expectation, to the best of our knowledge, there does not exist a fully automated verification system to prove the soundness of these rewrites for tensors of arbitrary rank and size.
    Previous works, while successful in verifying rewrites with tensors of concrete rank, do not provide guarantees in the unbounded setting.
    
    To fill this gap, we introduce \project{}, the first automatic verification system that can verify tensor graph rewrites for input tensors of arbitrary rank and size.
    We introduce a core language, \dsl{}, to represent rewrite rules using a novel axis definition, called \emph{\agg{}-axis}, which allows us to reason about an unbounded number of axes.
    We achieve unbounded verification by proving that there exists a bound on tensor ranks,
    under which bounded verification of all instances implies the correctness of the rewrite rule in the unbounded setting.
    We derive an algorithm to compute this rank using the denotational semantics of \dsl{}. \project{} employs this algorithm to generate a finite number of bounded-verification proof obligations, which are then dispatched to an SMT solver using symbolic execution to automatically verify the correctness of the rewrite rules.
    We evaluate \project{}'s verification capabilities by implementing rewrite rules present in \xla{}'s algebraic simplifier.
    The results demonstrate that \project{} can prove the correctness of 115 out of 175 rules in their full generality, while the closest automatic, \emph{bounded}-verification system can express only 18 of these rules.

\end{abstract}

\maketitle
\section{Introduction} \seclabel{intro}

Deep learning frameworks, such as TensorFlow~\cite{tensorflow}, PyTorch~\cite{pytorch1}, and JAX~\cite{jax}, along with their backend optimizing tensor compilers, such as \xla{}~\cite{xla} and TorchInductor~\cite{pytorch2}, have been instrumental in enabling machine learning (ML) practitioners to experiment, train, and deploy various neural network architectures.
These tensor compilers manipulate computations with tensors as first-class objects, utilizing tensor computational graphs as their intermediate representation (IR).
The nodes in these graphs represent tensor operators, while the edges denote input/output tensors.
Examples include \xla{}'s High Level Operators (\xlahlo{})~\cite{xla-hlo}, PyTorch's \textsf{torch.fx} operators~\cite{torch-fx}, and ONNX's tensor operators~\cite{onnx}.
Middle-end tensor compiler optimizations often transform these tensor graphs to produce more efficient variants.
A key optimization which has attracted significant research~\cite{tensat, taso, pet} is tensor graph rewrites.
This optimization is a common pass in industrial tensor compilers such as \xla{}\footnote{\url{https://github.com/openxla/xla/blob/main/xla/hlo/transforms/simplifiers/algebraic_simplifier.cc}}.

% properties of these rewrites and the forms
Tensor graph rewrites transform a subgraph of the original tensor graph to an equivalent version that is more efficient.
For example, consider the $\tdot$ (einsum) operator that takes two tensors and a set of contraction and batch axes as input and performs sum-of-products over the specified contraction axes.
If the batch and contraction axes are empty (precondition), an expensive $\tdot$ operation may be decomposed into a simpler composition of element-wise multiplication and expand operations (represented as $\tdot\textsf{(A,B)} \Rightarrow_C  \binary(*,\expand\textsf{(A)},\expand\textsf{(B)})$ in our notation, with precondition $C$).
In general, these rewrite rules are expected to be correct for tensors of arbitrary rank (number of axes) and size (individual axis sizes). %, unless explicitly guarded by tensor rank (number of dimensions) and size (dimension size) preconditions. 
We term this property as the rewrite rules being \emph{rank-} and \emph{size-polymorphic}. We found that most tensor graph rewrites in \xla{}'s algebraic simplifier have this property. Hence, it is important that compiler developers ascertain that the rules are indeed correct for input tensors of arbitrary rank and size.

There have been multiple efforts at formally proving the correctness of these rewrites. However, \emph{automatically} verifying tensor graph rewrites for tensors of arbitrary rank and size has remained challenging. Previous automatic verification techniques instantiate fixed-ranked, concrete-sized tensors with symbolic values as a part of their verification process. As a result, their proofs do not generalize to the \emph{unbounded setting}, where input tensors can be of arbitrary rank and size. Further, existing verification systems do not support preconditions on rules, which we find abundant in compilers such as \xla{}.
For example, TASO~\cite{taso} proposes an axiomatic approach to verify tensor graph rewrites. The rewrite rules they synthesize from their axiom pool are rank- and size-polymorphic. However, the axioms themselves are only verified on small, concrete-sized input tensors.
TENSAT~\cite{tensat} improves the search efficiency of TASO, but relies on TASO's rewrite rule synthesis and verification process.
PET~\cite{pet} uses statistical testing to give rigorous guarantees on more expressive rewrites for tensor computational graphs of concrete-rank and concrete-sized tensors. 
Successful works that have verified tensor graph rewrites in the unbounded setting have been manual verification efforts with proof assistants like Coq~\cite{atl-pldi, atl-popl}. %\jai{expand} %ATL~\cite{atl-popl} system proposes a tensor language and proves that the rewrites written in that language are correct in both rank and size polymorphic manner with the aid of the Coq theorem prover. ~\cite{atl-pldi} extends it to include tensor graph lowering process. These require manual effort that should be repeated for each new rewrite and can be burdensome during the rewrite rule development process. This process involves possibly many incorrect efforts at formulating a new rewrite rule and manually finding counterexamples is difficult for each such iteration. %Additionally, it is difficult to manually find counterexamples for each incorrect rule. 

In this paper, we introduce the first automatic, push-button verification system, \project{}, that allows users to succinctly express and verify tensor graph rewrite rules for input tensors of arbitrary rank and size. Further, in order to aid tensor compiler developers, we develop \project{} to be able to handle the complexities of rewrite rules found in the \xla{} compiler. We have to overcome several key challenges in realizing these goals. %We had to overcome several key challenges in building the \project{} system. %can reason about the correctness of tensor graph rewrite rules for input tensors of arbitrary rank and size. Further, \project{} supports verification of rules guarded by preconditions covering rules used by the industrial strength \xla{} compiler. %We had to overcome multiple technical challenges in this process. 
%A key challenge of building \project{} is to construct an algorithm that reasons about the correctness of rewrite rules in the unbounded setting. On one hand, directly instantiating unbounded tensors is not feasible. On the other hand, using fixed, \ATTN{concrete-shaped} symbolic tensors during automatic verification can result in proofs that are only valid for input tensors of that particular rank and size.

\paragraph{\textbf{Representation}}  First, we need to succinctly represent tensor graph rewrite rules in a way that allows reasoning about their correctness in the unbounded setting.
Second, we need to model the highly parameterized operators in \xlahlo{}.
For example, \xlahlo{}'s $\conv$ operator works with arbitrary batch, contraction, and spatial axes specifications and has rich padding and dilation attributes.
Further, \xla{} rewrite rules can be guarded by complicated preconditions. %We give such an example in \secref{}.

We overcome these challenges by designing a rewrite rule specification language, called \dsl{}, with tensor operators closely resembling those in \xlahlo{}.
\dsl{} introduces a novel axis definition called \emph{\agg{}-axis} that represents a possibly unbounded set of axes, rather than capturing one axis at a time.
A tensor in \dsl{} consists of a \emph{finite} number of \agg{}-axes, where they can potentially be instantiated to any number of axes.
This allows us to reason about how input tensors are mutated by tensor operators, treating similar axes collectively.
All \dsl{} operators are defined to work with \agg{}-axes, making the rewrite specifications rank- and size-polymorphic.
Additionally, the representation with \emph{\agg{}-axes} is general enough to support a sizable subset of tensor operators and their parameterizations, as defined in \xlahlo{} (e.g. $\conv$).
However, the representation cannot support \emph{layout-sensitive} operators, such as $\reshape$ and $\bitcast$.
We implemented the most common operators appearing in \xla{}'s rewrite rules to demonstrate \dsl{}'s expressivity.
Rewrite specifications in \dsl{} accept preconditions, which can also be rank- and size-polymorphic.
Finally, we provide denotational semantics of these operators, which we use to verify these rules. %To the best of our knowledge, this is the first time semantics for an industrial strength compiler's tensor operators are presented. 

\paragraph{\textbf{Verification}} The next major challenge that we need to overcome is: given specifications in the \dsl{}, how can we automatically verify that the rewrites are correct? This requires proving them correct in the unbounded setting. Instantiating unbounded tensors is not feasible, while using symbolic tensors of concrete-rank and size during automatic verification can result in proofs that only hold for input tensors of that particular rank and size, as shown in \secref{unboundedmotivation}.

To handle unbounded sizes, we leverage the capabilities of SMT solvers to perform unbounded reasoning using uninterpreted functions and unbounded integers.
For unbounded ranks, we overcome the challenge by proving that there exists a bound on the ranks, such that if we prove the rule correct for all possible ranks within the bound, then the rule is also correct for arbitrary ranks and sizes.
This allows us to reduce the unbounded verification problem to a set of bounded verification cases, which can be dispatched to an automatic verification engine. %We formally prove this fact constructively by providing 
We derive a bound inference algorithm using the denotational semantics of \dsl{}.
Given these theoretical foundations, \project{} automatically verifies a given rewrite rule written in \dsl{} in two steps. 
First, it uses the bound inference algorithm to find a \emph{sufficient} rank for each \agg{}-axis. Next, for all ranks equal to or below this bound, \project{} instantiates concrete-ranked input tensors for each \agg{}-axis. It then uses big-step operational semantics, derived from the denotational semantics of \dsl{}, to create proof obligations as SMT queries using symbolic execution. If all of these bounded cases are proven by an SMT solver, \project{} then concludes that the rewrite rule is correct.

We implement the \project{} system in Haskell and use Grisette~\cite{grisette} as the symbolic evaluation engine.
We further provide a \frontend{} that abstracts away some core \dsl{} constructs to make developing rewrite rules easier. \project{} dispatches all SMT queries to the Z3~\cite{z3} SMT solver to ascertain the soundness of a given rewrite rule.
To evaluate \project{}'s capabilities in representation and verification, we assessed it using a comprehensive set of rules incorporated within \xla{}'s algebraic simplifier.
We successfully represented 121 of these rules and verified 115 of them in \project{}. Almost all rules were verified in the unbounded setting within a second. Comparatively, other bounded-verification systems, such as TASO~\cite{taso} and PET~\cite{pet}, could only express 14 and 18 rules and verify 6 and 16 rules, respectively, exemplifying \project{}'s representation and verification capabilities.
Further, we show case studies where \project{} helps generalize rewrites with complicated preconditions, showcasing its usefulness during compiler development.

In summary, this paper makes the following contributions. 
\begin{itemize}
    \item We present a language, \dsl{} (\secref{dsl}) to specify tensor graph rewrites with
    (1) a novel axis construct called \agg{}-axes, allowing representation of operators and rewrite rules that are rank- and size-polymorphic 
    (2) operator specifications that closely resemble \xlahlo{}
    (3) precondition specifications on rewrite rules.
    \item We provide denotational semantics for \dsl{} (\secref{denotational-sem}).
    To the best of our knowledge, this is the first formalization of a sizable subset of operators in a production-quality tensor IR.
    \item We provide the first \emph{automatic} verification strategy (\secref{verification}) that can reason about the correctness of tensor graph rewrites that are rank- and size-polymorphic.
    \item We develop \project{} that implements this verification strategy and evaluate it (\secref{eval}) by representing and verifying tensor graph rewrites present in \xla{}'s algebraic simplifier.
\end{itemize}

\noindent\project{} is open-source, publicly available at \url{https://github.com/ADAPT-uiuc/TensorRight}

\section{Background and Motivation} \seclabel{motivation}

We first provide background on tensors and related concepts before motivating the need for \emph{automatic} and \emph{unbounded} verification of tensor graph rewrites. We then describe a key insight of our unbounded-verification methodology.

\subsection{Preliminaries} \seclabel{prelim}

\emph{Tensors} are a generalization of scalars (0-dimensional tensors), vectors (1-dimensional tensors), and matrices (2-dimensional tensors) to $n$-dimensional objects.
%\charith{A popular implementation of a tensor...}
A popular implementation of a tensor is multi-dimensional arrays.
An \emph{axis} of a tensor (also commonly known as a dimension) represents a direction across which the tensor's data can be traversed.
The \emph{rank} of a tensor (also commonly known as its dimensionality) refers to the number of axes the tensor has.
The \emph{shape} of a tensor describes the \emph{size} of each axis, i.e., the number of elements that exist along each axis.
The \emph{size} of a tensor refers to the total number of elements in the tensor, calculated as the product of individual axis-sizes.
The axes of an $n$-dimensional tensor are numbered from $0$ up to $n-1$.
Each element of a tensor is uniquely identified by a list of \emph{positional} indices, with one index for each axis.

For example, a 2-dimensional tensor $m$, containing 3 groups of elements along axis 0 and 4 groups of elements along axis 1, has a rank of 2, a shape of $3 \times 4$, and a size of 12.
Such a tensor can be accessed by a pair of positional indices: $m[i, j]$ denotes the value at the $i^{\text{th}}$ position along axis 0 and $j^{\text{th}}$ position along axis 1.
Another way to implement tensors, called \emph{named tensors}, assigns explicit names to the axes of a tensor, referred to as \emph{named-axes}.
\project{} adopts the latter approach, which we describe in detail in \secref{named-axes}.

\paragraph{Tensor Graph Rewrites}

A tensor operator refers to any operation that takes tensors as input and returns tensors as output.
A tensor computational graph is a directed acyclic graph that represents a sequence of tensor operations.
The nodes in the graph represent tensor operators, while the edges indicate the flow of data (tensors) between these operators.
Tensor graph rewriting is a key optimization employed by tensor compilers, which replaces a subgraph of the input graph
with another, equivalent subgraph, subject to certain preconditions.
This optimization is governed by a set of \emph{rules}, called tensor graph rewrite rules.

The algebraic simplifier of the \xla{} compiler %\cite{xla-compiler}
contains hundreds of tensor graph rewrites, executed during program compilation.
Given that \xla{} consistently executes this simplifier, it is crucial to ensure the correctness of these rewrite rules.
However, currently the rewrite rules are not verified.
Therefore, the developers rely on unit tests and limit the generality of the rules to alleviate concerns of introducing compiler bugs.

We aim to automatically verify tensor graph rewrites deployed in \xla{}, which work with tensors of arbitrary ranks and sizes.
Existing verified tensor graph rewrite systems are either not automatic, lack support for complex \xlahlo{} operators and preconditions, or cannot verify rewrite rules in the unbounded setting .
We now demonstrate the importance of automatic and unbounded verification of tensor graph rewrites
and discuss a key insight that enables us to achieve these goals.

\subsection{Need for Automatic Verification}

\newcommand{\complexrule}[0]{\textsc{FoldConvInputPad}}

The algebraic simplifier in \xla{} contains complex rewrite rules whose correctness is not intuitive.
Developers often limit the generality of these rewrite rules by imposing preconditions.
For instance, consider the \complexrule{} rule shown in \figref{motivating-example-automatic}.

\begin{wrapfigure}{R}{0.35\textwidth}
\vspace{-0.5em}
    \centering
    $
    \small
    \begin{array}{l}
        \boldsymbol{\mathsf{FoldConvInputPad:}}
        \\
        \ \ \ \ \llet S_{ol} = S_l + S_{lp} \lin \\ 
        \ \ \ \ \llet S_{oh} = S_h + S_{hp}  \lin \\ 
        \ \ \ \ \ \ \ \ \conv(\pad(t, 0, S_{lp}, S_{hp}, S_{ip}), t',
        \\
        \phantom{\ \ \ \ \ \ \ \ \conv(}B, F, O, 
        \\
        \phantom{\ \ \ \ \ \ \ \ \conv(}
        S_l, S_h, S_i, S'_i)
        \\
        \phantom{\mathsf{concatenate}(}
        \Longrightarrow_{S_{ip} = 0 \lland S_i = 1}
        %_{\substack{m_{ol} = m_l + (m_i + \overline{1}) \times m_{lp}\\ 
                                %   \ \ m_{oh} = m_h + (m_i + \overline{1}) \times m_{hp}\\ 
                                %   \ \ \ \ \ \ m_{oi} = (m_i + m_{pi}) + (m_i \times m_{pi}) }} 
                                %   \\[1cm]
        \\
        \ \ \ \ \ \ \ \ \conv(t, t', B, F, O, 
        \\
        \phantom{\ \ \ \ \ \ \ \ \conv(}
        S_{ol}, S_{oh}, S_{i}, S'_{i})
    \end{array}
    $
    \caption{\complexrule{} rule taken from \xla{}'s Algebraic Simplifier.}
    \figlabel{motivating-example-automatic}
    %\vspace{-0.7em}
\end{wrapfigure}

The idea behind the \complexrule{} rule is simple: fold the padding operator into the convolution operator.
It folds the edge padding, i.e., the lower and higher padding ($S_{lp}$ and $S_{hp}$) into the convolution padding ($S_l$ and $S_h$), but does not fold the interior padding ($S_{ip}$) into the base dilation ($S_i$).
The precondition of this rule requires zero interior padding ($S_{ip} = 0$) and a base dilation of one ($S_i = 1$).
The \xla{} repository contains the following comment\footnote{\url{https://github.com/openxla/xla/blob/ac380bb187abdb3efbbac776141e3a2300209232/xla/service/algebraic_simplifier.cc\#L8764-L8767}} on the preconditions:

\begin{quote}
    \emph{Edge padding composes with itself in the straightforward way, but composing interior padding is nontrivial, and we cowardly refuse to think about it. If we see interior padding in either the kPad or conv, bail if there's any sort of padding in the other.}
\end{quote}

Developers restrict the rule because the general case is non-trivial.
Existence of an automatic verification system would allow incremental refinement of the rule and provide counterexamples during development.
As a consequence, it would enable developers to build more general rewrite rules and be confident that the rewrite rules are valid for arbitrary ranks and sizes.
For example, \project{} can prove a more general version of the \complexrule{} rule with interior padding, as discussed in \secref{case-studies}.
The generalization involves removing the precondition and computing the folded padding and dilation attributes as: $S_{ol} = S_l + S_i \times S_{lp}$, $S_{oh} = S_h + S_i \times S_{hp}$, and $S_{oi} = S_i + S_i \times S_{ip}$.
This folding of attributes is non-trivial to come up with, but \project{} can prove the rewrite rule to be valid.
With the aid of \project{}, we believe that compiler engineers will be able to quickly iterate through complex rewrite rules and get feedback on their correctness through counterexamples, thereby increasing productivity. % and more complicated, bug-free rewrite rules inside the compiler.

\subsection{Need for Unbounded Verification} \seclabel{unboundedmotivation}

\definecolor{exampleyellow}{HTML}{ac973e}
\definecolor{rulegreen}{HTML}{00CC00}
\definecolor{rulepurple}{HTML}{cf8bff}
\definecolor{ruleblue}{HTML}{007FFF}
\definecolor{ruleorange}{HTML}{ff7247}
\newcommand{\zero}[0]{\emph{zero}}
\newcommand{\one}[0]{\textcolor{exampleyellow}{\mathsf{one}}}
\newcommand{\slicedyup}[0]{\textsc{SliceDyUpSlice}}

We demonstrate with an example that verifying a tensor graph rewrite for a certain rank may not be sufficient to guarantee correctness in the unbounded setting, where input tensors can be of arbitrary ranks and sizes.
%We demonstrate with an example, that verifying a tensor graph rewrite for a certain rank may not sufficient to verify the rewrite in the unbounded setting, where input tensors can be of arbitrary ranks and sizes.
%It may be tempting to believe that verifying a rule on a small, fixed-rank is sufficient.
Consider the \slicedyup{} rule shown in \eqnref{unbounded-invalid}, where $S$ denotes the shape of the input tensor \textsf{Y}, $\zero$ denotes a tensor with all values as 0, and $\overline{v}$ denotes a vector with all values as $v$.
Other operator inputs like \textsf{start}, \textsf{end}, and \textsf{stride} are called \emph{operator attributes}.

\vspace{5pt}
\begin{equation} \eqnlabel{unbounded-invalid}
    \dyupdateslice(\slice(\mathsf{Y},
    \hspace{-0.5cm}\tikz[baseline=(n1.base)]{\node[inner sep=0pt] (n1) {$\overline{0}$};
    \node[overlay, above left=of n1, yshift=-0.5cm, xshift=1cm] (t1) {\textsf{start}};
    \draw [->, shorten <=2pt] (n1.north) to (t1.south);
    },
    \tikz[baseline=(n2.base)]{\node[inner sep=0pt] (n2) {$\displaystyle\left\lfloor \frac{S+1}{2} \right\rfloor$};
    \node[overlay, above=of n2, yshift=-0.7cm] (t2) {\textsf{end}};
    \draw [->, shorten <=2pt] (n2.north) to (t2.south);
    },
    \tikz[baseline=(n3.base)]{\node[inner sep=0pt] (n3) {$\overline{1}$};
    \node[overlay, above right=of n3, yshift=-0.5cm, xshift=-1cm] (t3) {\textsf{stride}};
    \draw [->, shorten <=2pt] (n3.north) to (t3.south);
    }\hspace{-0.5cm} \hspace{-0.05cm}),
    \hspace{-0.25cm}\tikz[baseline=(n4.base)]{\node[inner sep=0pt] (n4) {$\zero$};
    \node[overlay, below left=of n4, yshift=0.65cm, xshift=1.4cm] (t4) {\textsf{update}};
    \draw [->, shorten <=2pt] (n4.south) to (t4.north);
    },
    \tikz[baseline=(n5.base)]{\node[inner sep=0pt] (n5) {$\overline{1}$};
    \node[overlay, below right=of n5, yshift=0.65cm, xshift=-1.25cm] (t5) {\textsf{offset}};
    \draw [->, shorten <=2pt] (n5.south) to (t5.north);
    }
    \hspace{-0.3cm}) \Longrightarrow \dyupdateslice(\slice(\mathsf{Y},
    \hspace{-0.5cm}\tikz[baseline=(n6.base)]{\node[inner sep=0pt] (n6) {$\overline{0}$};
    \node[overlay, above left=of n6, yshift=-0.5cm, xshift=1cm] (t6) {\textsf{start}};
    \draw [->, shorten <=2pt] (n6.north) to (t6.south);
    },
    \tikz[baseline=(n7.base)]{\node[inner sep=0pt] (n7) {S};
    \node[overlay, above=of n7, yshift=-0.42cm] (t7) {end};
    \draw [->, shorten <=2pt] (n7.north) to (t7.south);
    },
    \tikz[baseline=(n8.base)]{\node[inner sep=0pt] (n8) {$\overline{2}$};
    \node[overlay, above right=of n8, yshift=-0.5cm, xshift=-1cm] (t8) {\textsf{stride}};
    \draw [->, shorten <=2pt] (n8.north) to (t8.south);
    }\hspace{-0.55cm}),
    \hspace{-0.25cm}\tikz[baseline=(n9.base)]{\node[inner sep=0pt] (n9) {$\zero$};
    \node[overlay, below left=of n9, yshift=0.65cm, xshift=1.4cm] (t9) {\textsf{update}};
    \draw [->, shorten <=2pt] (n9.south) to (t9.north);
    },
    \tikz[baseline=(n10.base)]{\node[inner sep=0pt] (n10) {$\overline{1}$};
    \node[overlay, below right=of n10, yshift=0.65cm, xshift=-1.25cm] (t10) {\textsf{offset}};
    \draw [->, shorten <=2pt] (n10.south) to (t10.north);
    }\hspace{-0.3cm})
\end{equation}
\vspace{10pt}
The left-hand side (\textsf{LHS}) expression first applies the $\slice$ operator, which extracts a sub-tensor from \textsf{Y} by picking elements from the $0^{\text{th}}$ index (\textsf{start}) up to the ${\lfloor \frac{S+1}{2}\rfloor}^{\text{th}}$ index (\textsf{end}) along each axis.
%\charith{middle and end at the same time is confusing. Can we use start indices and end indices directly? 0 and S+1/2}.
It is then followed by the $\dyupdateslice$ operator, which \emph{zeroes} (\textsf{update}) out all the points whose axes indices are greater than or equal to 1 (\textsf{offset}).
The right-hand side (\textsf{RHS}) expression first applies the $\slice$ operator, which extracts a sub-tensor from \textsf{Y} by picking every $2^{nd}$ element (\textsf{stride}) along each axis.
It is then followed by the $\dyupdateslice$ operator, which zeroes out all the points whose axes indices are greater than or equal to 1.

\begin{figure}[h]
\centering
%\captionsetup{justification=centering, font=small}
\begin{tabular}{|c|c|c|}
    \hline
    & {\small Rank-1 input} & {\small Rank-2 input} \\ \hline
    \textsf{LHS} & 
    \begin{minipage}{0.4\textwidth}
        \centering
        \includegraphics[width=0.5\textwidth]{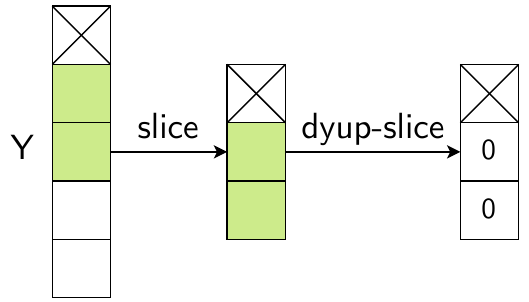}
    \end{minipage}
    & 
    \begin{minipage}{0.4\textwidth}
        \centering
        \includegraphics[width=0.9\textwidth]{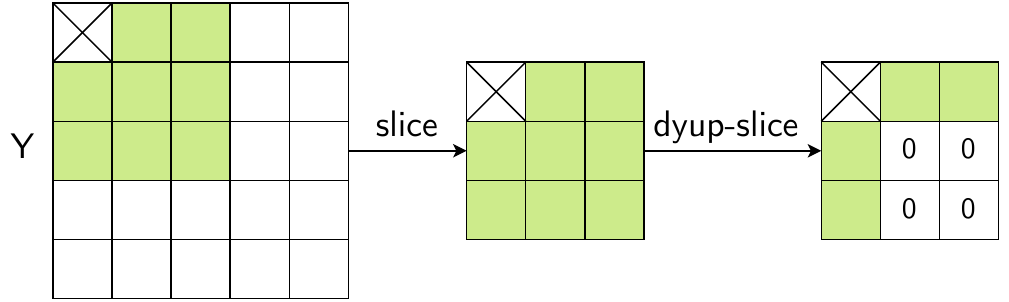}
    \end{minipage} \\ \hline
    \textsf{RHS} & 
    \begin{minipage}{0.4\textwidth}
        \centering
        \includegraphics[width=0.5\textwidth]{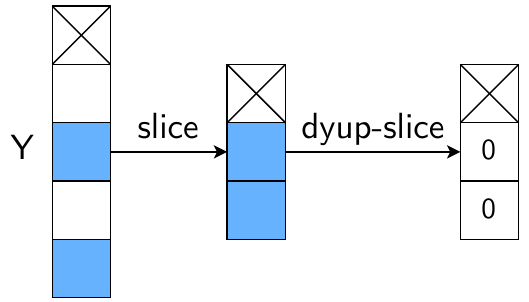}
    \end{minipage}
    & 
    \begin{minipage}{0.4\textwidth}
        \centering
        \includegraphics[width=0.9\textwidth]{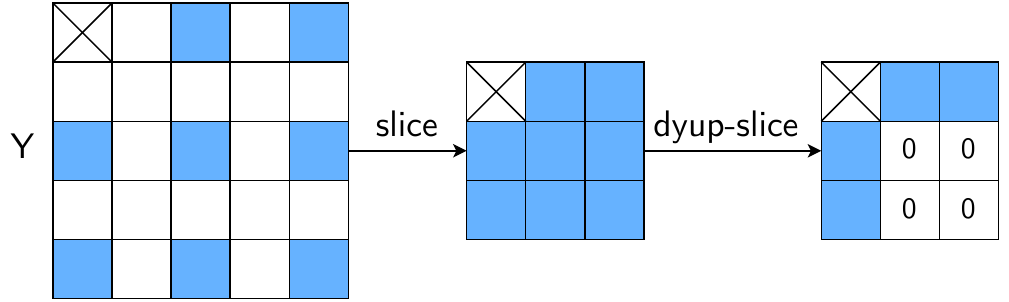}
    \end{minipage} \\ \hline
\end{tabular}
\caption{Illustration for \slicedyup{} rule depicting various regions in the input tensor for ranks 1 and 2. The leftmost element is shown as crossed out. The \textcolor{rulegreen}{\textsf{green}} and \textcolor{ruleblue}{\textsf{blue}} regions indicate the elements extracted by $\slice$ in \textsf{LHS} and \textsf{RHS}, respectively. The zeroed out region after the $\dyupdateslice$ is indicated by 0-elements.}
\figlabel{unbounded-motivation}
\end{figure}

\figref{unbounded-motivation} illustrates the rule applied to input tensors of rank 1 and 2.
The crossed-out element corresponds to the point with all indices as 0.
We refer to this as the \emph{leftmost} element.
The leftmost element is left untouched throughout the computation in both \textsf{LHS} and \textsf{RHS}.
The \textcolor{rulegreen}{\textsf{green}} region, along with the leftmost element, corresponds to the sub-tensor obtained after the $\slice$ in \textsf{LHS}.
The \textcolor{ruleblue}{\textsf{blue}} region, along with the leftmost element, corresponds to the sub-tensor obtained after the $\slice$ in \textsf{RHS}.
The zeroed out region after the $\dyupdateslice$ is represented using 0-elements.

As we can observe for the 1-dimensional case, the final \textsf{LHS} and \textsf{RHS} expressions are equal since the \textcolor{rulegreen}{\textsf{green}} and \textcolor{ruleblue}{\textsf{blue}} regions get zeroed out completely.
Meanwhile, for the 2-dimensional case, the \textcolor{rulegreen}{\textsf{green}} and \textcolor{ruleblue}{\textsf{blue}} regions are not completely zeroed out, so the \textsf{LHS} and \textsf{RHS} expressions have regions that do not match.
Therefore, the rule is valid for rank 1 but is invalid for rank 2.
In fact, the rule is invalid for any rank higher than 2.
This example demonstrates that verifying the rule for a certain rank, in this case 1, does not guarantee correctness at other ranks, making it important to verify the rule for all possible ranks.
%i.e., in the unbounded setting \charith{definition of unbounded setting covers both rank and size; can we omit it here?}.

\subsection{Key Observation} \seclabel{key-observation}

A rewrite rule is valid if the \textsf{LHS} and \textsf{RHS} expressions are equal for input tensors of any rank.
Otherwise, the rule is invalid and would exhibit a \emph{counterexample}.
A counterexample contains a valuation of all the variables in the rule (including tensors and operator attributes) and an access $A$ (list of positional indices), such that $\mathsf{LHS}[A]$ and $\mathsf{RHS}[A]$ do not match.

Verifying a rule for each rank separately is infeasible since there are a denumerable number of such ranks.
However, we make an observation that there exists a \emph{sufficient} rank $k$, such that if the rule is valid for rank $k$, then it can be proven valid for any rank greater than $k$.  
This insight allows us to avoid verifying the rewrite rule for ranks greater than $k$.
A more intuitive way to understand this is through its contraposition, i.e., if a counterexample exists at a rank greater than $k$, then a counterexample exists at rank $k$.
We demonstrate with the same example rule from \secref{unboundedmotivation} that verifying the rule for rank 2 is sufficient to ensure correctness for all higher ranks.

\begin{figure}[h]
    \centering
    %\captionsetup{justification=centering}
    \begin{subfigure}{0.48\textwidth}
        \centering
        \includegraphics[width=\textwidth]{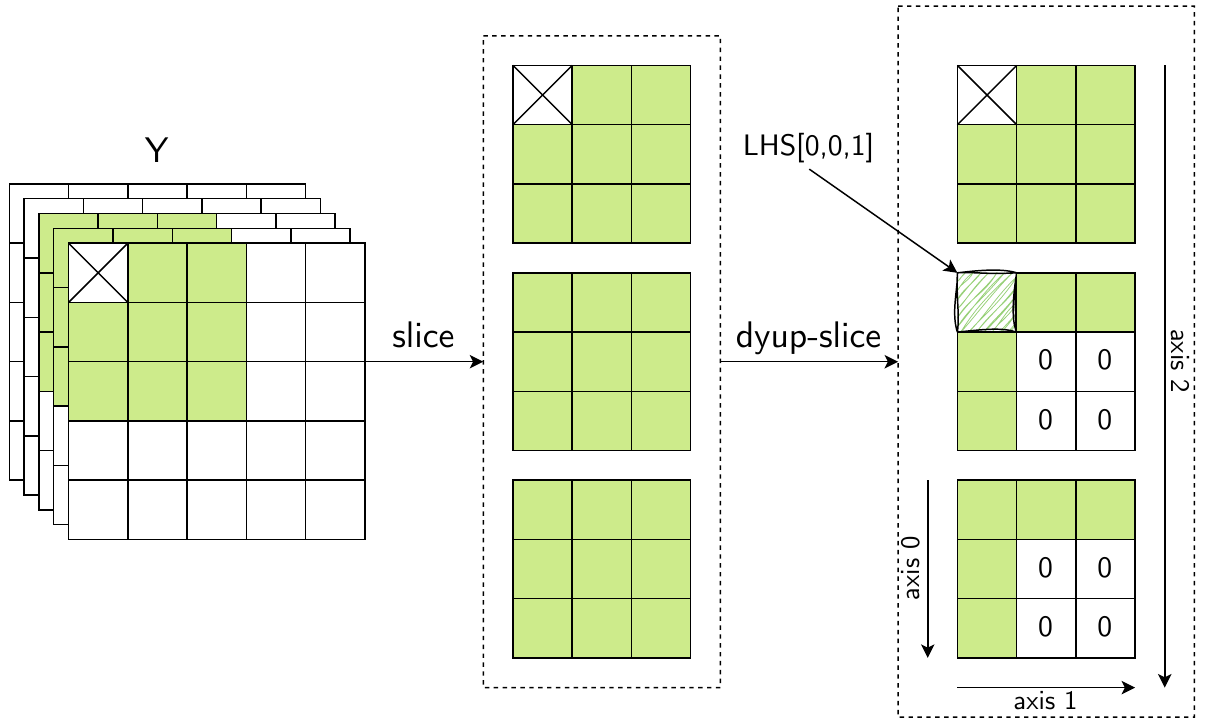}
        \caption{\textsf{LHS} for rank-3 input}
        \figlabel{um-3d-lhs}
    \end{subfigure}
    \begin{subfigure}{0.48\textwidth}
        \centering
        \includegraphics[width=\textwidth]{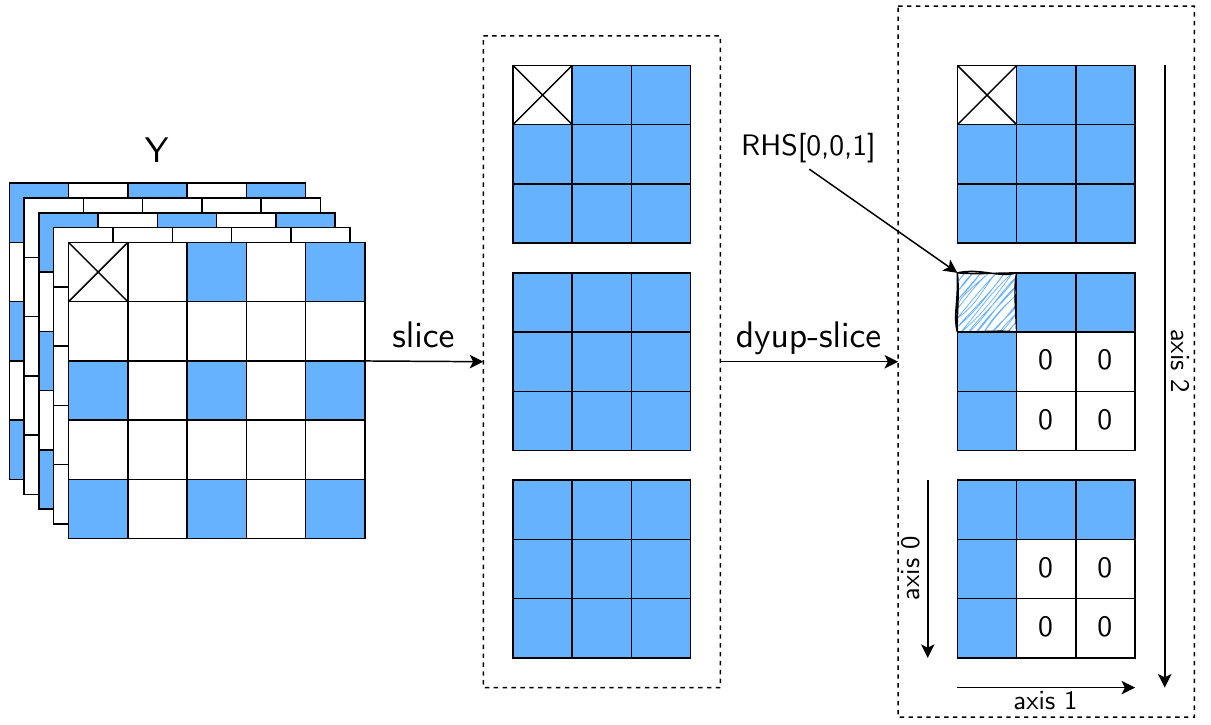}
        \caption{\textsf{RHS} for rank-3 input}
        \figlabel{um-3d-rhs}
    \end{subfigure}
    \vspace{-1em}
    \caption{The \slicedyup{} specialized for rank-3 inputs. The \textsf{LHS} and \textsf{RHS} expressions are presented using 2-dimensional cross-sections along axis 2. The access $A^3 = [0,0,1]$ is highlighted in \textsf{LHS} and \textsf{RHS}.}
    \figlabel{um-3d}
    %\vspace{-0.5em}
\end{figure}

Consider the \slicedyup{} rule applied to input tensors of rank 3, as shown in \figref{um-3d}.
Clearly, the \textsf{LHS} (\figref{um-3d-lhs}) and \textsf{RHS} (\figref{um-3d-rhs}) have regions (\textcolor{rulegreen}{\textsf{green}} and \textcolor{ruleblue}{\textsf{blue}}) that do not match.
Therefore the rule is invalid for rank 3 and exhibits a counterexample.
The counterexample would contain an access $A^3$ at which \textsf{LHS} and \textsf{RHS} do not match.
This access can correspond to any location in the \textcolor{rulegreen}{\textsf{green}} and \textcolor{ruleblue}{\textsf{blue}} regions.
Without loss of generality, we consider the case when $A^3 = [0,0,1]$, highlighted as a sketched-out location in \figref{um-3d-lhs} and \figref{um-3d-rhs}.

Given this counterexample at rank 3, we try to construct a counterexample at rank 2, which would contain an access $A^2$.
An obvious counterexample construction involves \emph{projecting} out one of the axes.
%$A^2$ is obtained by projecting $A^3$ on the remaining two axes.
There are 3 choices for the axis to project out: axis 0, axis 1, and axis 2, as shown in \figref{um-3d}.
% If we project out axis 2, then the resulting counterexample access $A^2$ would be $[0, 0]$, but $\mathsf{LHS}[0,0]$ and $\mathsf{RHS}[0,0]$ have to always match since it is the \textcolor{rulegreen}{\textsf{leftmost}} element.
If we project out axis 2, then the resulting counterexample access $A^2$ would be $[0, 0]$, but $\mathsf{LHS}[0,0]$ and $\mathsf{RHS}[0,0]$ have to always match since it is the leftmost element.
Therefore, this projection does not lead to a counterexample.
We instead observe that projecting out any of axis 0 or axis 1 results in a counterexample at rank 2.
In fact, any counterexample at rank 3 for the \slicedyup{} rule can be \emph{lowered} to a counterexample at rank 2.
Moreover, it can be shown that any counterexample at a higher rank can be lowered to a counterexample at rank 2.
Therefore, if the \slicedyup{} rule is valid for rank 2, then it is valid for any higher-rank.
Note that the same does not hold for rank 1: given a counterexample at rank 2, we cannot construct a counterexample at rank 1.
Based on these observations, we conclude that 2 is a sufficient rank for this rule and verifying the rule for ranks 1 and 2  ensures correctness in the unbounded setting.

In \project{}, we extend these observations to any arbitrary rule by first partitioning the axes of a tensor into ``groups'', where all axes in a group share the same ``role'' and are treated uniformly by the operators.
We then present an algorithm to compute a sufficient rank for each group, allowing us to avoid verifying the rule for ranks beyond these sufficient ranks.

\section{Overview} \seclabel{overview}

Our goal is to automatically verify rewrite rules for arbitrary tensors and operator attributes.
Handling arbitrary tensors requires reasoning about tensor values, axis sizes, and ranks, all of which could be arbitrary.
We illustrate the challenges in representing and verifying rewrite rules with the help of an example and present \project{}, that helps us overcome these challenges.

\subsection{\project{} Rewrite Rules}

Similar to many other tensor graph rewrite systems, \project{} rewrites are modeled as rewriting an \textsf{LHS} tensor expression to an \textsf{RHS} tensor expression, subject to certain preconditions.
The users use the constructs provided by \dsl{} to write tensor expressions and preconditions.
We use the notation $\mathsf{LHS} \Rightarrow_C \mathsf{RHS}$ to represent a generic tensor graph rewrite, where \textsf{LHS} and \textsf{RHS} are tensor expressions and $C$ is the precondition under which the rewrite rule is supposedly correct, which is verified by our system.

\paragraph{Example}

Consider the \textsc{DysliceToSlice} rule shown in \eqnref{overview-slice}, extracted from \xla{}'s algebraic simplifier, which desugars the $\dyslice$ operator to the more efficient $\slice$ operator.
\begin{equation} \eqnlabel{overview-slice}
    \dyslice(\mathsf{Y}, B, L) \Longrightarrow_{E - B' = L \lland P = 1 \lland B' = B} \slice(\mathsf{Y}, B', E, P) 
\end{equation}
\definecolor{rulepurple}{HTML}{E1D5E7}
\begin{wrapfigure}{R}{0.48\textwidth}
\vspace{-1em}
    \centering
    \includegraphics[width=0.48\textwidth]{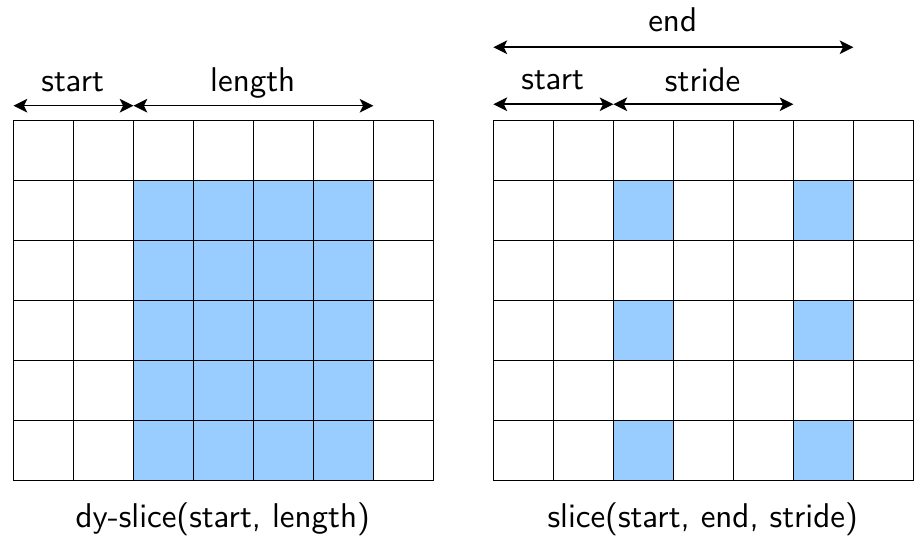}
    \caption{Illustration of $\dyslice$ and $\slice$ operators. The shaded regions denote the operator outputs.}
    \figlabel{dsl-example-vis}
\vspace{-1em}
\end{wrapfigure}
\figref{dsl-example-vis} depicts the \textsc{DysliceToSlice} rule visually.
The $\dyslice$ operator extracts a sub-tensor from the input tensor \textsf{Y}, where the start-index for each axis is specified in $B$ and the length of the slice along each axis is passed in $L$. 
Meanwhile, the $\slice$ operator also extracts a sub-tensor from within a bounding box in the input tensor \textsf{Y}.
The start-indices for the bounding box are specified in $B'$, while the end-indices (exclusive) are specified in $E$.
$P$ specifies the stride for each axis, which determines the step size between elements in the bounding box.
%which tells how many elements to skip in the bounding box \jai{rephrase}.

The \textsc{DysliceToSlice} rule is generally not correct, unless $E - B'$ (the size of the bounding box in $\slice$) is equal to $L$ (the length in $\dyslice$).
The other requirements are that $\slice$ should skip no elements, i.e., $P=1$, and the start indices in $\slice$ and $\dyslice$ must be the same, i.e., $B' = B$.
Since these are specified in the precondition, the \textsf{RHS} expression is equivalent to the \textsf{LHS} expression.
Our goal is to \emph{represent} and \emph{verify} this rule for arbitrary tensors and operator attributes.
We now discuss each challenge individually and explain how our system addresses them.
%We now discuss some challenges and how our system overcomes them.

\subsection{Representation in \dsl{}}
\seclabel{dsloverview}

\paragraph{Challenges in representation}
First, since the \textsc{DysliceToSlice} rule should be correct for all instantiations of the tensor \textsf{Y}, our system should allow representing a tensor of arbitrary rank and size. 
Second, it should allow specifying arbitrary operator attributes like \textsf{start}, \textsf{end}, \textsf{stride}, and \textsf{length}, to ensure that the rule is correct for all possible operator attributes.
%This is because the rule should be valid for all operations attributes as well.
It should also allow performing operations on attributes, like doing arithmetic on \textsf{end} and \textsf{start} while specifying the precondition.
Third, the operators provided by the system should model those that are found in \xlahlo{}.
Last, it should allow defining preconditions on a rule (e.g., \textsf{stride} values being 1).

\paragraph{\frontend{}} We present a frontend language in which users can express \emph{abstract} rewrite rules along with preconditions.
\project{} will build internal representations (\secref{dsl}) of rewrite rules from the specification, which can be instantiated to arbitrary ranks.

\begin{figure}[h]
    \centering
    \begin{lstlisting}[xleftmargin=1.5em,language=haskell,label=lst:dsl-example,mathescape=true,caption={The \textsc{DysliceToSlice} rule represented in \dsl{}.},frame=lines]
  rule = do
      rcls <- newRClass "rcls"
      [size, start, start', length, end, stride] <-
          newMaps ["size", "start", "start'", "length", "end", "stride"] rcls
      Y <- newTensor $\texttt{@}$TensorInt "Y" [rcls --> size]
      lhs <- dynamicSlice Y [rcls --> start] [rcls --> length]
      rhs <- slice Y [rcls --> start'] [rcls --> end] [rcls --> stride]
      precondition [end, start', length] $\dollar$ \[end, start', length] -> end - start' .== length
      precondition [stride] $\dollar$ \[stride] -> stride .== 1
      precondition [start, start'] $\dollar$ \[start, start'] -> start' .== start
      rewrite "DynamicSlice(Y) => Slice(Y)" lhs rhs

  verifyDSL rule
\end{lstlisting}
\vspace{-1em}
\end{figure}

\lstref{dsl-example} illustrates the \textsc{DysliceToSlice} rule implemented in the \frontend{}.
Instead of using fixed-rank tensors, the tensors in \dsl{} are represented with \emph{\agg{}-axes}, that can be instantiated to any number of axes.
All axes in an \agg{}-axis share the same ``role'' and are treated uniformly by all the operators, allowing us to reason about an unbounded number of axes compactly.
There might be multiple \agg{}-axes in a rule to capture different roles of axes.
Some of them must be instantiated to the same rank in a correct rule and this constraint is represented by a \emph{rank class} (\rclass{}) as discussed in \secref{verification}.

On line 2 in \lstref{dsl-example}, we declare a new \rclass{} and we refer to it by \texttt{rcls}.
In \dsl{}, we can refer to the \agg{}-axis with the \rclass{} itself, if an \rclass{} has exactly one \agg{}-axis.
On lines 3-4, we declare multiple \emph{abstract-maps} on \texttt{rcls}.
These abstract-maps can be instantiated to \emph{concrete-maps}, whose domain is the same as the axes in the \agg{}-axis represented by \texttt{rcls}.
When instantiated, they map axes in \texttt{rcls} to symbolic values, which can represent axes sizes, \textsf{start} indices, \textsf{end} indices etc.
On line 5, we declare a new tensor containing integer elements, with the shape $\{\texttt{rcls} \mapsto \texttt{size}\}$.
Similarly, the created tensor can also be instantiated with any number of axes (all of which behave in the same way) and symbolic sizes.
We then construct the \textsf{LHS} and \textsf{RHS} expressions on lines 6 and 7, respectively.
On line 8, we specify the precondition that the difference of the \textsf{end} and \textsf{start} indices should be equal to the \textsf{length}.
On line 9, we specify that the \textsf{stride} values should be 1 for all axes.
On line 10, we specify that the \textsf{start} indices should be same on both sides.
Finally, we construct the rewrite rule on line 11.

This example demonstrates how we can represent tensors with arbitrary rank and sizes, specify arbitrary operator attributes, construct tensor expressions, and specify complex preconditions in \dsl{}.
We now discuss our verification methodology built on top of this representation.

\subsection{Verification}

\begin{figure}[h]
    \centering
    \includegraphics[width=\linewidth]{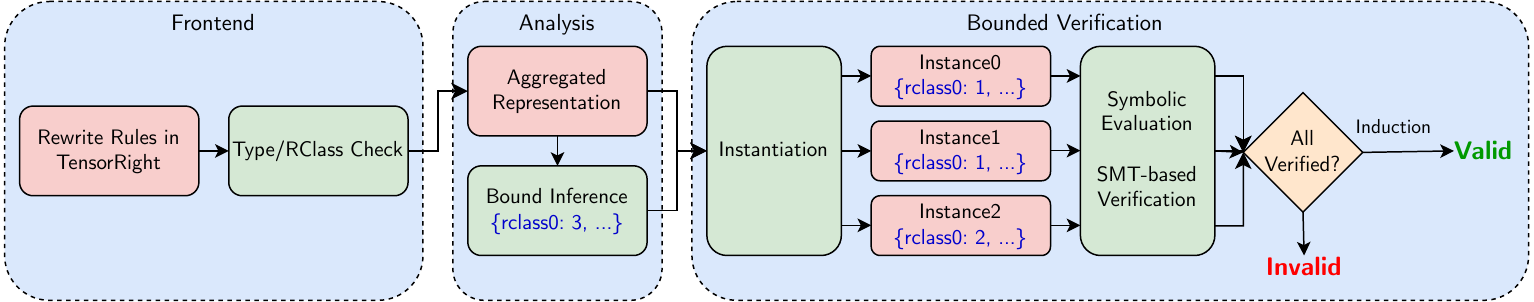}
    \caption{\project{} Overview and Workflow.}
    \figlabel{tr-overview}
%\vspace{-1em}
\end{figure}

\figref{tr-overview} describes our approach to verifying rules with tensors of arbitrary rank and sizes.
With the aggregated representation created in \secref{dsloverview}, the system infers a bound for each \rclass{} and instantiates them to all the ranks within the bound.
Each instance can then be proven with symbolic evaluation by SMT solvers.
This is based on a crucial theorem (\secref{bound}) that for all \rclasses{} in a rule, there exists a mechanically derivable bound on the ranks, such that proving all instances within the bound implies that the rule is correct for all ranks.
With this bound, we convert the unbounded-verification proof obligation to a finite number of \emph{bounded-verification} proof obligations.
The correctness of the rule with ranks beyond the bound is then established with induction.

\begin{wrapfigure}{R}{0.61\textwidth}
%\vspace{-0.5em}
    \centering
    \begin{lstlisting}[xleftmargin=1.5em,language=haskell,label=lst:dsl-lowering,mathescape=true,caption={Frontend Rewrite Rule lowered to our core syntax.},frame=lines,escapechar=^]
  ruleLowered = do
      # Instiate rcls to a concrete set
      ^\fbox{\icon{S}}^ <- rcls
      # Define concrete maps on rcls
      size <- Map ^\fbox{\icon{S}}^ SymInteger
      # Maps for other attributes like start, length
      Y <- newTensor $\texttt{@}$TensorInt "tensor" [rcls --> size]
      lhsSym <- [[...]] # Symbolic Representation of LHS
      rhsSym <- [[...]] # Symbolic Representation of RHS
      pre <- ... # Precondition of the rule
      # Axes and Shape Checks
      assert $\dollar$ lhsShape .== rhsShape
      A <- generalAccess(lhs, rhs)
      verify $\dollar$ pre && lhsValid -> lhsSym[A] .== rhsSym[A]
\end{lstlisting}
%\vspace{-1em}
\end{wrapfigure}

On line 13 in \lstref{dsl-example}, we call \texttt{verifyDSL} on the constructed rule, which is our main verification routine.
First, it infers a bound for every \rclass{} in the rule.
The bound inference algorithm collects the number of unique \emph{boolean conditions} and \emph{tensor accesses} in the rewrite rule.
In this case, there is only one \rclass{}, no boolean conditions, and a single unique access to the tensor.
We infer the bound to be 1.
This means that if the rule is correct for all rank-1 tensors, then the rule is correct for tensors containing any number of axes.
We discuss these conditions and accesses in detail in \secref{verification}.

Second, it specializes the rule for these ranks, ending up with fixed-rank but arbitrary-sized tensors.
\lstref{dsl-lowering} shows some details of the lowered code in our core \dsl{}.
On line 3, we instantiate \texttt{rcls} with a known-rank (1 in this case) and get the set of \emph{concrete}-axes in $\icon{S}$.
On line 5, we declare concrete-maps corresponding to the abstract maps in \lstref{dsl-example}.
These maps now have a concrete-domain, same as $\icon{S}$.
In lines 7-10, we create concrete-ranked input tensors, generate symbolic representations of \textsf{LHS} and \textsf{RHS} with symbolic evaluation, and specify preconditions.
We make an initial assertion on line 12 that both expressions have the same symbolic shape.
One line 14, we express one of the main verification conditions, i.e., the equality between \textsf{LHS} and \textsf{RHS} expressions under a general access \texttt{A}, which is discharged to an SMT solver.
The solver decides that the proof obligation is a tautology and the rule is deemed verified in the unbounded setting.

The example demonstrates how we take the abstract specification of a rewrite rule expressed in \frontend{}, infer a bound for each \rclass{}, instantiate every \agg{}-axis, and discharge bounded-verification proof obligations. We describe the \dsl{} in \secref{dsl}, the denotational semantics in \secref{denotational-sem}, and our verification methodology in detail in \secref{verification}.

\section{Rewrite Rule Representation} \seclabel{dsl}

In this section, we show our rank- and size-polymorphic rewrite rule representation constructed by the \frontend{} introduced in \secref{dsloverview}.
We take a set of operators from \xlahlo{} and model them in \project{}. We show a subset of the modeled operators in \figref{syntax} for discussion (more modeled operators can be found in \secref{ext-sem}).
We can extend the set to other \emph{layout-insensitive} operators, whose semantics do not depend on the particular physical layout of the operands. Some operators that do \emph{not} fall into this category include $\reshape$ and $\bitcast$.

The key distinction of our DSL from \xlahlo{} is that we group axes into \emph{\agg{}-axes}, where all the axes in the same group share the same ``role'' and are treated uniformly by the operators.
This allows us to describe rewrite rules with tensors containing any number of axes in a compact manner.
In \secref{verification}, we extend this representation by tagging the \agg{}-axes with \emph{rank classes} to help us with the instantiation of \agg{}-axes into concrete-axes for verification purposes.

\begin{wrapfigure}{R}{0.66\textwidth}
\vspace{-0.5em}
    \begin{minipage}{0.66\textwidth}
        \begin{tabular}{llll}
            $\tau$& $\coloneq$ & $\mathsf{Int}\mid\mathsf{Bool}\mid\mathsf{Real}$&Type\\
            $a$&$\in$&$\mathcal{A}$&Named-axes\\
            $x$&$\in$&$\mathcal{X}=\mathcal{P}(\mathcal{A})$&Aggregated-axes\\
            % $d$&$\in$&$\mathcal{D}$&Aggregated dimensions\\
            $f$&$\in$&$\mathsf{list}[\mathsf{Int}]\to \mathsf{Int}$&Map function\\
            $m$&$\coloneq$&$\mathcal{M}\mid \mathsf{fmap}(f, m+)$&Maps\\
            $X$&$\in$&$\mathcal{P}(\mathcal{X})$&Set of \agg{}-axes\\
            $S,I$&$\in$&$m^\mathcal{X}$&Shapes and indices\\
            $R$&$\in$&$\mathcal{X}^\mathcal{X}$&Relabel maps\\
            % $v$&$\coloneq$&$i:\mathsf{Int}$ & Scalar literal \\
            % &$\mid$&$b:\mathsf{Bool}$ \\
            % &$\mid$&$r: \mathsf{Real}$\\
            % $v$&$\coloneq$&$i:\mathsf{Int}$&Scalar literal\\
            % &$\mid$&$b:\mathsf{Bool}\mid r: \mathsf{Real}$\\
            $v$&$\coloneq$&$i:\mathsf{Int}\mid b:\mathsf{Bool}\mid r: \mathsf{Real}$&Scalar literal\\
            $e$&$\coloneq$&$\mathcal{T}$ (Literal)&Tensor expression\\
              &$\mid$&$\mathcal{V}$ (Variable) \\
              &$\mid$&$\textsf{const}(v, S)$\\
              &$\mid$&$\tiota(S, x)$\\
              &$\mid$&$\expand(e, S)$\\
              &$\mid$&$\binary(\oplus, e_l, e_r)$\\
              &$\mid$&$\pad(e, v, S_l, S_h, S_i)$\\
              &$\mid$&$\slice(e, I_s, I_e, I_p)$\\
              &$\mid$&$\dyslice(e, I, S)$\\
              &$\mid$&$\dyupdateslice(e, e_u, I)$\\
              &$\mid$&$\reduce(\oplus, e, X)$\\
              &$\mid$&$\relabel(e, R)$\\
              &$\mid$&$\concat(e_l, e_h, x)$\\
              % &$\mid$&$\convbase(t, X_b, X_f, X_o, I_s)$\\
              % &$\mid$&$\conv(t, X_b, X_f, X_o, I_s, S_l, S_{ld}, S_h, S_{rd})$\\
              % &$\mid$&$\tdot(t_l, t_r, X_c, X_b)$\\
              % &$\mid$&$\clamp(t_{min}, t, t_{max})$\\
              % &$\mid$&$\select(t_c, t_t, t_f)$\\
              % &$\mid$&$\reverse(t, X)$\\
            $g$&$\in$&$\mathsf{list}[\mathsf{Int}]\to \mathsf{Bool}$&Predicate function\\
            $P$&$\coloneq$&$\mathsf{fold}(g, m+)$&Precondition\\
            $Rule$&$\coloneq$&$e_{lhs}\Rightarrow_{P*} e_{rhs}$&Rewrite rule
        \end{tabular}
        \caption{Core rewrite rule representation with selected operators.}
        \figlabel{syntax}
    \end{minipage}
\vspace{-1em}
\end{wrapfigure}

\subsection{Named Axes} \seclabel{named-axes}
Following named-tensors in PyTorch~\cite{pytorchnamedtensors} and named-axes in JAX~\cite{jaxnamedaxes}, we give explicit names to the axes of a tensor and call them \emph{\ind{}-axes}.
We can treat the \ind{}-axes of a tensor as an unordered set for layout-insensitive operators and express tensor shapes as mappings from names to sizes.
For example, we may give the names $h$ and $v$ to the horizontal and vertical axes respectively of a $2 \times 3$ tensor $t$ and
the shape of this tensor would be the mapping $\{h\mapsto 3, v\mapsto 2\}$.
Such a tensor can be accessed with an \emph{access map}, which is a mapping from \ind{}-axes to indices.
We then have $t[\{h\mapsto 1, v\mapsto1\}]=w$, given the domain of the access map is exactly the set of the \ind{}-axes, there's no out-of-bounds access, and $w$ is the value at that access. 

An operator that works on multiple tensors will need to match them by the \ind{}-axes, as shown in the following examples.

\begin{itemize}
    \item $\binary(+, t_1, t_2)$, where $t_1$, $t_2$ have the shapes $\{a_1\mapsto 2, a_2\mapsto 3\}$ and $\{a_1\mapsto 2, a_2\mapsto 3\}$, respectively.
    The two shapes match as they have the same set of \ind{}-axes and the corresponding \ind{}-axes have the same sizes.
    
    \item $\tdot(t_1, t_2, \{a_2, a_3\}, \{a_4\})$, where the shapes of $t_1$ and $t_2$ are $\{a_1\mapsto 2, a_2\mapsto 3, a_3\mapsto 4, a_4\mapsto 5\}$ and $\{a_5\mapsto 6, a_2\mapsto 3, a_3\mapsto 4, a_4\mapsto 5\}$, respectively. 
    For the $\tdot$ operator, we need to match the contraction and batch axes,
    in this case $\{a_2, a_3\}$ and $\{a_4\}$, respectively.
    The resulting tensor has the shape $\{a_1\mapsto 2, a_4\mapsto 5, a_5\mapsto 6\}$.
\end{itemize}

Note that sometimes, we need to rename the axes to avoid name clashes.
For example, with the $\tdot$ operator, if the two tensors share some \emph{spatial} \ind{}-axes (neither contraction nor batch axes), we need to rename them before applying the operator to make sure that the \ind{}-axes in the resulting tensor are unique.

\subsection{Aggregated Axes}
As we've shown in the $\tdot$ example in \secref{named-axes}, we have matched some \emph{set} of axes between the two tensors.
This partitions the \ind{}-axes in a tensor into sets of axes that have the same ``role'' in the expression.
For example, $a_2$ and $a_3$ have the same roles as contraction axes.
Based on this observation, we introduce \emph{\agg{}-axes}, which is a set of \ind{}-axes.

We can partition the set of \ind{}-axes of a tensor into \emph{disjoint} \agg{}-axes and write expressions directly using \emph{uninterpreted} \agg{}-axes.
The \agg{}-axes can be instantiated to some concrete set of \ind{}-axes, and the number of instantiated \ind{}-axes is called the \emph{rank} of an \agg{}-axis.
This allows us to write expressions with an arbitrary number of \ind{}-axes in a uniform and simple way.

For example, in the expression $\tdot(t_1, t_2, \{x_2, x_3\}, \{x_4\})$, if we assume that the set of \ind{}-axes in $t_1$ and $t_2$ are $x_1\cup x_2\cup x_3\cup x_4$ and $x_5\cup x_2\cup x_3\cup x_4$, respectively, then the resulting tensor has $\{x_1,x_4,x_5\}$ as the set of \agg{}-axes.
It's easy to see that we can get back the $\tdot$ example shown in \secref{named-axes} by instantiating all \agg{}-axes to singleton sets.
This instantiation is not arbitrary, and we elaborate on how to specify the constraints on the instantiations with \emph{rank classes} in \secref{verification}.
%A canonical mapping will then be established for the axes in \agg{}-axes with the same \rclass{}.

We can then lift the tensor semantics to aggregated semantics: shapes or indices can be expressed with, or instantiated from \agg{}-axes.
Instead of being a mapping from \ind{}-axes to integers, we now need a \emph{nested mapping} that maps \agg{}-axes to another map from names in the \agg{}-axes to integers.
As a convention, we will refer to the inner mappings as a \emph{map} and the outer mapping as an \emph{\agg{}-map}.
For example, the following is valid \agg{}-map:
\begin{equation*}
    \left\{
    \{i_1, i_2\}\mapsto\{i_1\mapsto 2,i_2\mapsto 3\}, 
    \{i_3\}\mapsto\{i_3\mapsto 4\}
    \right\}
\end{equation*}
\begin{definition}
An \emph{\agg{}-map} $M$ is valid if it is a nested mapping from \agg{}-axes to maps from \ind{}-axes to integers such that:
\begin{itemize}
    \item $\forall x_1, x_2\in \mathsf{dom}(M), x_1\ne x_2 \to x_1\cap x_2=\varnothing$, and
    \item $\forall x\in\mathsf{dom}(M), \mathsf{dom}(M[x]) = x$.
\end{itemize}
\end{definition}
Note that $M[x]$ represents the value mapped to $x$ in $M$. Shape and indices are aliases for \agg{}-maps in specific contexts and they have their additional validity conditions, depending on the context.
Here, we give the validity conditions for tensor shapes and access indices: 
\begin{definition}
A valid \emph{tensor shape} $S$ is a valid \emph{\agg{}-map}, such that $\forall x\in\mathsf{dom}(S), \forall a\in x, S[x][a]\geq 0$. The shape of a tensor $t$ is denoted as $\shape(t)$.
\end{definition}
\begin{definition}
A valid \emph{access} $A$ (used for accessing tensors) with respect to a valid tensor shape $S$, is a valid \emph{\agg{}-map} such that
\begin{itemize}
    \item $\mathsf{dom}(A) = \mathsf{dom}(S)$, and
    \item $\forall x\in\mathsf{dom}(A), \forall a\in x, 0\leq A[x][a]<S[x][a]$.
\end{itemize}
The set of all valid \emph{accesses} given a tensor shape $S$, is denoted by $\access(S)$.
A tensor $t$ is then viewed as a mapping from the set $\access(\shape(t))$ to elements, and the element at the access $A$ is denoted as $t[A]$.
The set of all \agg{}-axes of a tensor $t$ is denoted as $\axes(t)$.
%, which is the same as $\mathsf{dom}(\shape(t))$.
\end{definition}

Operator attributes are also expressed using \agg{}-maps, with each operator having its own validity conditions.
For example, the pseudo operator $\padlow$, which only does low-padding, can pad a tensor with $l_1$, followed by $l_2$, for the axes in $x_1$ using the expression: 
$\padlow(\padlow(t,\{x_1\mapsto l_1\}), \{x_1\mapsto l_2\})$.
The validity condition for $\padlow$ allows padding with negative shapes but disallows creating a tensor with a negative shape (more details in \secref{denotational-sem}).
Assuming $t$ has the shape $s_0$ in the \agg{}-axis $x_1$, the resulting shape in the $\padlow$ expression will be $s_0 + l_1 + l_2$.
Note that we are doing an element-wise combination of maps, where the maps must have the same \ind{}-axes and the resulting map contains the sum of the corresponding axis sizes.
Element-wise combination ($\mathsf{fmap}$) is the only allowed operation on the maps to combine them, as
we define \agg{}-axes as set of named axes that have the same ``role'' in the expression.
Note that we may also combine a map with scalars by lifting the scalar to a constant map.

\subsection{Rewrite Rule}

Similar to many other tensor graph rewrite systems, \project{} models rewrite rules as rewriting an $\mathsf{LHS}$ expression to an $\mathsf{RHS}$ expression, subject to certain preconditions. See the following example:
\begin{equation*}
    \padlow(\padlow(t, 0, \{x_1\mapsto l_1\}), 0, \{x_1\mapsto l_2\}) \Rightarrow_{l_1 \geq 0 \lland l_2 \geq 0} \padlow(t, 0, \{x_1\mapsto l_1 + l_2\})
\end{equation*}
In this rule, we aggregated all the \ind{}-axes in the tensor $t$ into the \agg{}-axis $x_1$.
$l_1$ and $l_2$ are two maps from \ind{}-axes in $x_1$ to padding sizes.
In the $\mathsf{RHS}$, the two maps are combined in an element-wise way.
Similarly, our preconditions are predicates lifted to operate on the maps in an element-wise way.
The condition $l_1\geq 0$ here means that the padding sizes in $l_1$ must be greater than or equal to 0 for all \ind{}-axes.

\section{Denotational Semantics} \seclabel{denotational-sem}
We give the denotational semantics of the \xlahlo{} operators in \figref{semantics}.
We will use denotational semantics notations for deriving a bound on the ranks, but note that since our semantics map from our language to computable tensor objects, we can easily derive a big-step operational semantics and perform symbolic evaluation. Due to space limitations, we will only show the semantics of some selected operators.
More operators are available in \secref{ext-sem}. They are usually simple or can be expressed using existing operators or the techniques introduced here.

The domain of our denotational semantics are tensors, which map accesses to elements. The elements can be boolean, integers, or real numbers.
There is also a special type of element called a \emph{Reduction Element}, denoted as $\mathsf{Red}^\oplus_{I_0,I_1,\cdots} f(\{\mathsf{dom}(I_0)\mapsto I_0,\mathsf{dom}(I_1) \mapsto I_1,\cdots\})$.
Here, $\oplus$ is a binary operator and $I_0, I_1, \cdots$ are called \emph{reduction indices}.
We may sometimes omit the indices and write $\mathsf{Red}_X^\oplus f(X)$, where $X=\{\mathsf{dom}(I_0), \mathsf{dom}(I_1), \cdots\}$ is the set of \agg{}-axes being reduced.

The introduction of a \emph{reduction element} is based on pragmatic reasons.
As the sizes of reduced axes are unbounded, we cannot expand the reduction to sum all the values being reduced.
Thus, we leave the sum uninterpreted and provide special treatment for such elements during verification.
\begin{figure}[h]
\vspace{-1em}
    \centering
    \footnotesize
    \begin{equation*}
        \inference{
        }
        {
            \sem{\const(v, S)} =
            \{A \mapsto v \mid A \in \access(S)\}
        }
        [\textsc{Const}]
        \semlabel{Const}
    \end{equation*}
    
    \begin{equation*}
        \inference{
            \mathsf{let}~\{a\}=x & x\in\mathsf{dom}(S)
        }
        {
                \sem{\tiota(S, x)} = \{A \mapsto A[x][a] \mid A \in \access(S)\}
        }
        [\textsc{Iota}]
        \semlabel{Iota}
    \end{equation*}

    \begin{equation*}
        \inference{
            \mathsf{dom}(S)\cap\axes(e)=\varnothing
        }
        {
                \sem{\expand(e, S)} =
                \{A \mapsto \sem{e}[A\vert_{\axes(e)}] \mid A \in \access(\shape(e)\cup S)\}
        }
        [\textsc{Expand}]
        \semlabel{Expand}
    \end{equation*}
    
    \begin{equation*}
        \inference{
            \shape(e) = \shape(e')
        }
        {
                \sem{\binary(\oplus, e, e')} =
                \{A \mapsto \sem{e}[A] \oplus \sem{e'}[A] \mid A \in \access(e)\}
        }
        [\textsc{BinOp}]
        \semlabel{BinOp}
    \end{equation*}
    
    \begin{equation*}
        \inference{
            \mathsf{let}~S = \shape(e) &
            \mathsf{let}~S'=S+S_l\geq 0&\mathsf{let}~not\mbox{-}pad = \lambda A. A \geq S_l
        }{
            \sem{\padlow(e, v, S_l)}=
            \{A\mapsto \mathsf{if}~not\mbox{-}pad(A)~\mathsf{then}~\sem{e}[A-S_l]
            ~\mathsf{else}~v
            \mid A\in\access(S')\}
        }[\textsc{PadLow}]
        \semlabel{PadLow}
    \end{equation*}
    
    \begin{equation*}
        \inference{
            0\leq I_s\leq I_e\leq \shape(e) & I_p > 0
        }
        {
                \sem{\slice(e, I_s, I_e, I_p)} =
                \{A \mapsto \sem{e}[I_s+A\times I_p] \mid A \in \access(\left\lceil \frac{I_e-I_s}{I_p}\right\rceil)\}
        }
        [\textsc{Slice}]
        \semlabel{Slice}
    \end{equation*}
    
    \begin{equation*}
        \inference{
            I+S\leq \shape(e)&S>0&I\geq 0
        }{
            \dyslice(e, I, S)=\{A\mapsto \sem{e}[A+I]\mid A\in\access(S)\}
        }[\textsc{DySlice}]
        \semlabel{DySlice}
    \end{equation*}
    
    \begin{equation*}
        \inference{
            \mathsf{let}~S_u = \shape(e_u) & 
            I+S_u\leq \shape(e) & S_u > 0 &
            I\geq 0 \\
            \mathsf{let} ~ acc = \lambda A.
            \mathsf{if} ~ A\geq I\wedge A < I + S_u ~ \mathsf{then} ~ \sem{e_u}[A - I] ~
            \mathsf{else} ~ \sem{e}[A]
        }{
            \dyupdateslice(e, e_u, I)=
            \{A\mapsto acc(A) \mid A \in\access(e)\}
        }[\textsc{DyUpdateSlice}]
        \semlabel{DyUpdateSlice}
    \end{equation*}
    
    \begin{equation*}
        \inference{
            \mathsf{let}~S=\shape(e) &
            \mathsf{let}~\{x_0\cdots x_k\}=X\subseteq \axes(e)\\
            \mathsf{let}~acc=\lambda A. \mathsf{Red}^\oplus_{I_0,\cdots,I_k}\sem{e}[\{x_0\mapsto I_0, \cdots,x_k\mapsto I_k\}\cup A]
        }
        {
                \sem{\reduce(\oplus, e, X)} =
                \{A\mapsto acc(A)\mid A\in\access(S \setminus S\vert_X)\}
        }
        [\textsc{Reduce}]
        \semlabel{Reduce}
    \end{equation*}

    \begin{equation*}
        \inference{
            \forall x_1,x_2\in\mathsf{dom}(R), x_1\neq x_2\to R[x_1]\neq R[x_2]&
            \mathsf{dom}(R)=\axes(e) \\
        }{
            \sem{\relabel(e, R)}=\{A\mapsto \sem{e}[A\circ R]\mid A\in\access(S\circ R^{-1})\}
        }[\textsc{Relabel}]
        \semlabel{Relabel}
    \end{equation*}
    
    \begin{equation*}
        \inference{
            \mathsf{let}~\{a\} = x &
            x \in \axes(e) &
            x \in \axes(e') &
            \mathsf{let} ~ S = \shape(e) &
            \mathsf{let} ~ S' = \shape(e') \\
            \forall x' \in \axes(e), x'\neq x \to S[x'] = S'[x']\\
            \mathsf{let} ~ S'' = \{x'\mapsto
            \mathsf{if} ~ x'=x  ~\mathsf{then}~S'[x]~\mathsf{else}~\{a'\mapsto 0\mid a'\in x'\}
            \mid x'\in\axes(e)\}\\
            \mathsf{let} ~ acc = \lambda A. \mathsf{if} ~ A \geq S'' ~ \mathsf{then} ~ \sem{e'}[A-S''] ~ \mathsf{else} ~ \sem{e}[A]
        }{
            \concat(e, e', x)=
            \{A\mapsto acc(A)\mid
            A\in\access(S+S'')\}
        }[\textsc{Concat}]
        \semlabel{Concat}
    \end{equation*}
    \caption{Denotational Semantics of some core operators.}
    \figlabel{semantics}
\vspace{-1em}
\end{figure}

% \subsection{Semantic Rules}

\figref{semantics} shows the denotational semantics of some selected operators.
We overload some functions for convenience: $\shape(e)$ means $\shape(\sem{e})$, $\access(e)$ means $\access(\shape(e))$, and $\axes(e)$ means $\axes(\sem{e})$.
We also introduce some helper functions on valid \agg{}-maps.
Given a comparison operator $\odot$ and a binary operator $\oplus$, we lift them to \agg{}-maps as
\begin{gather*}
    M_1 \odot M_2 =\bigwedge_{x\in\mathsf{dom}(M_1)} \bigwedge_{a\in x} M_1[x][a]\odot M_2[x][a] \\
    M_1\oplus M_2 = \{x\mapsto \{a\mapsto M_1[x][a] \oplus M_2[x][a]\mid a\in M_1[x]\}\mid x\in\mathsf{dom}(M_1)\}
\end{gather*}

All these binary operations implicitly introduce the assumption that the two \agg{}-maps are valid and the domain of the two \agg{}-maps are the same.
We will omit these from our rules.
Sometimes, we may overload the notations to operate with constants. This is treated as operating with a nested map where the inner map are constant maps.

In our rules, we introduce bindings with $\mathsf{let}~name = \cdots$ notation.
Some rules like \semref{Iota} and \semref{Concat} require that an \agg{}-axis is singleton.
In these rules, we use the syntax $\mathsf{let}~\{a\}=x$ to say that $x$ is singleton and bind the singleton element in $x$ to $a$.
We use the notation $S|_K$ to denote the mapping restriction of $S$ to $K$, where $S$ is any map and $K$ is a subset of keys from $S$.

\paragraph{\textbf{Tensor operators}} Next, we explain the semantics of some select operators:
\begin{itemize}
    \item \semref{Const}: The $\const$ operator outputs a tensor with the desired shape $S$, with all accesses being mapped to the same constant element $v$.
    \item \semref{Iota}: The $\tiota$ operator projects the access-index for a specific singleton \agg{}-axis $x$. The resulting tensor holds values starting at 0, incrementing by 1 along that axis.
    \item \semref{Expand}: The $\expand$ operator introduces new \agg{}-axes to a tensor by duplicating the data in the tensor.
    It takes a shape containing the sizes of the new axes and the resulting tensor is accessed as if we are accessing the input tensor after removing these new axes from the access.
    The set of new axes must be disjoint from the original set of axes.
    
    \figref{opsem-expand} demonstrates the $\expand$ operator with an example.
    The input tensor contains two \agg{}-axes $x_h$ and $x_w$, each instantiated with 1 \ind{}-axes.
    We refer to the corresponding \ind{}-axes by $h$ and $w$, respectively.
    The input shape is $\{\{h\} \mapsto \{h \mapsto 5\}, \{w\} \mapsto \{w \mapsto 5\}\}$.
    We expand it by adding a new \agg{}-axis $x_d$, containing 1 \ind{}-axis $d$.
    The resulting tensor duplicates data across this new \ind{}-axis and has a shape of $\{\{h\} \mapsto \{h \mapsto 5\}, \{w\} \mapsto \{w \mapsto 5\}, \{d\} \mapsto \{d \mapsto 3\}\}$.
    \item \semref{BinOp}: The $\binary$ operator performs an element-wise operation on two identically-shaped tensors.
    \item \semref{PadLow}: We present a restricted version of the $\pad$ operator, $\padlow$, which pads only on the low-ends of each axis. The semantics test whether the access is in the padded region. If so, return the padded value, or access the original tensor, offset by the padding shape.
    
    \definecolor{rulegreen}{HTML}{b1d021}
    \figref{opsem-padlow} demonstrates the $\padlow$ operator with an example.
    The input tensor contains two \agg{}-axes $x_h$ and $x_w$, each instantiated with 1 \ind{}-axes.
    We refer to the corresponding \ind{}-axes by $h$ and $w$, respectively.
    The input shape is $\{\{h\} \mapsto \{h \mapsto 4\}, \{w\} \mapsto \{w \mapsto 4\}\}$.
    We perform zero-padding on $h$ and $w$ with padding attributes of $2$ and $1$ respectively.
    The zero-padding values are all present at the lower-ends of the axes.
    The resulting tensor a shape of $\{\{h\} \mapsto \{h \mapsto 6\}, \{w\} \mapsto \{w \mapsto 5\}\}$.
    \item \semref{Slice}: The $\slice$ operator extracts a sub-tensor, which has the same \ind{}-axes as the input tensor and contains the values inside a bounding box within the input.
    The indices for the bounding box are given by the starting indices $I_s$, limit indices $I_e$ (exclusive), and the positive strides $I_p$. The slice picks every $I_p[x][a]$ element along each \ind{}-axis $a\in x\in\mathsf{dom}(I_p)$.
    \item \semref{DySlice}: The $\dyslice$ operator extracts a sub-tensor from the input tensor.
    The indices for the bounding box are given by the starting indices $I$ and size of the bounding box $S$.
    The sizes must be positive and should not cause out-of-bounds accesses.
    Note that this is different from the \xla{} semantics, where our starting indices are not represented as a tensor, but as a map.
    We are then only able to express rewriting rules where the indices are used in an opaque way, or as a constant, or computed with element-wise operations, e.g., $\binary$.
    We found that this change does not affect the effectiveness of \project{} for verification purposes,
    and our approximation is able to express all rewrite rules involving $\dyslice$.
    %due to its dynamic nature, the \xla{} compiler isn't able to reason much about the operations to build the indices \jai{is this true? We would know statically where the values come from, it's just that the rewrite rules do not care about it}, and our approximation is able to express all rewriting rules involving $\dyslice$.
    \item \semref{DyUpdateSlice}: The $\dyupdateslice$ operator generates a result with a slice overwritten by $e_u$, starting at indices $I$. Our $\dyupdateslice$ operator follows the same approximation as $\dyslice$.
    \item \semref{Reduce}: The $\reduce$ operator takes a tensor and a set of \agg{}-axes $X$ as inputs, then returns a tensor mapping to uninterpreted reduction elements as the result. The resulting tensor has the shape $S\setminus S\vert_X$, essentially removing all the \agg{}-axes in $X$.
    We extend the semantics of $\reduce$ to make verification easier in \secref{ext-sem}.
    \item \semref{Relabel}: In \project{}, as we take an unordered view of the axes, we no longer need $\transpose$.
    However, we still need to re-match the axes, for example, when we want to describe some expressions such as $t+\transpose(t)$.
    The $\relabel$ operator is introduced for this axes-matching operation.
    It renames \agg{}-axes and does not change tensor contents.
    \item \semref{Concat}: The $\concat$ operator is another example where we introduce the singleton constraint on an \agg{}-axis.
    The two tensors should have the same shape on other axes.
    For the concatenating axis, the resulting size will be the sum of the operand sizes.
    The resulting tensor will then compute which operand tensor the access belongs to and perform the access.
\end{itemize}

\begin{figure}[h]
    \centering
    \captionsetup{justification=centering}
    \captionsetup[subfigure]{justification=centering}
    \begin{subfigure}{0.49\textwidth}
        \centering
        \includegraphics[width=\textwidth]{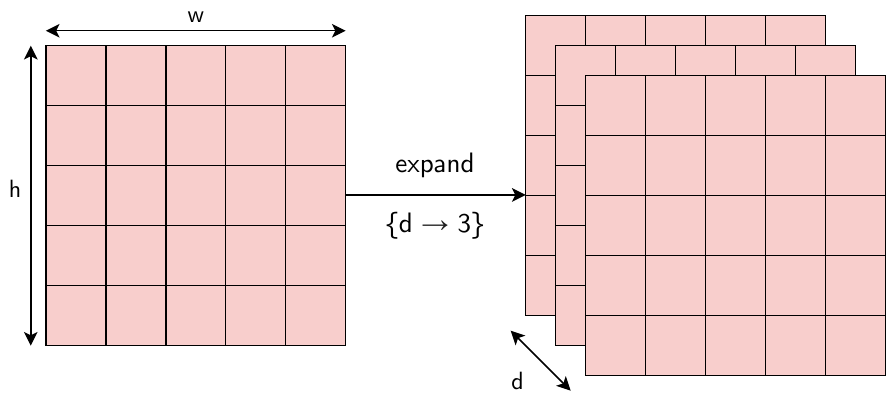}
        \caption{$\expand$: $\{d \rightarrow 3\}$ denotes the expansion shape}
        \figlabel{opsem-expand}
    \end{subfigure}
    \hfill
    \begin{subfigure}{0.49\textwidth}
        \centering
        \includegraphics[width=\textwidth]{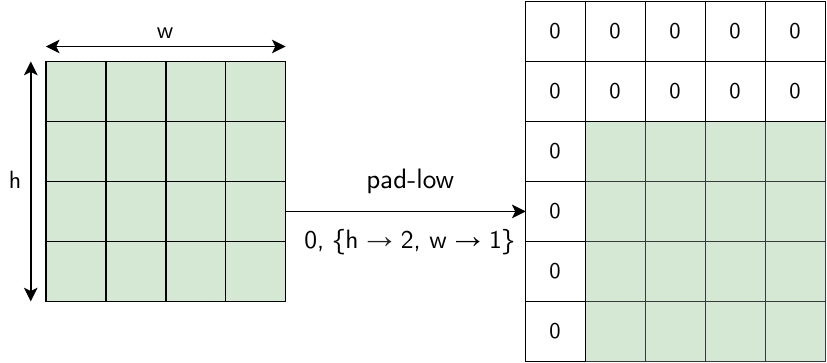}
        \caption{$\padlow$: $\{h \mapsto 2, w \mapsto 1\}$ denotes the padding attributes and 0 denotes the padding value}
        \figlabel{opsem-padlow}
    \end{subfigure}
    \vspace{-1em}
    \caption{Illustration of (a) $\expand$ and (b) $\padlow$ operators. $h$, $w$, and $d$ denote \ind{}-axes.}
    \figlabel{opsem-examples}
    \vspace{-1em}
\end{figure}

\paragraph{\textbf{Handling reduction elements}} 
In real-world \xla{} rewrite rules, the reduction might not always be the top-level operation and a reduction of a tensor may be performed in several steps.
For instance, in the rewrite rule $\reduce(\concat(\reduce(\mathsf{A}), \reduce(\mathsf{B}))) = \reduce(\concat(\mathsf{A},\mathsf{B}))$, the reduction in $\mathsf{LHS}$ is done in two steps.
To handle rules like this, we provide a \emph{limited} set of arithmetic rules for reduction elements in \figref{redelem}.
If no known rule applies, \project{} will report an error, indicating that the rule isn't supported.
The equivalence of two reduction elements is currently verified by a stronger condition, where we view them as sets and verify that there is a one-one mapping between them.
Note that not all correct rules meet the strengthened condition, but we find that it covers most rewrite rules with reductions.
Additionally, user-provided hints are needed for verifying the set equivalence, as discussed in \secref{verifyreduction}.
This is the only manual hint from the user during our verification.
Providing better mechanisms for handling reductions is a future work.

\begin{figure}[h]
\vspace{-1em}
\begin{minipage}{0.4\textwidth}
\begin{equation*}
v*\mathsf{Red}^{+}_X f(X) \to \mathsf{Red}^{+}_X v * f(X)
\end{equation*}
\end{minipage}
\begin{minipage}{0.55\textwidth}
\begin{equation*}
\mathsf{Red}^{+}_X f(X)*\mathsf{Red}^{+}_Y g(Y) \to \mathsf{Red}^{+}_{X,Y} f(X) * g(Y)
\end{equation*}
\end{minipage}
\begin{minipage}{0.55\textwidth}
\begin{equation*}
\mathsf{Red}^{\oplus}_X(\mathsf{Red}^{\oplus}_Y f(X,Y)) \to \mathsf{Red}^{\oplus}_{X,Y} f(X,Y)
\end{equation*}
\end{minipage}
\caption{Rules on reduction elements.}
\figlabel{redelem}
\vspace{-1em}
\end{figure}

\section{Verification of Rewrite Rules} \seclabel{verification}
After defining the representation of rewrite rules in \secref{dsl} and its denotational semantics in \secref{denotational-sem}, 
this section describes how \project{} verifies the rewrite rules given the semantics.
We will first overview our verification approach, which is based on $k$-induction~\cite{sheeran2000checking}, then provide proof sketches for our induction steps.

\subsection{Overview of the Verification}
To prove that a rewrite rule is correct, intuitively, we need to verify that the two expressions have the same denotation, possibly under some assumptions.
Given a rule $\mathsf{LHS}\Rightarrow_C\mathsf{RHS}$, we prove that
\begin{equation} \eqnlabel{valid-rule}
    \forall v\in vars, C\wedge\mathsf{valid\hbox{-}expr}(\mathsf{LHS})\to\sem{\mathsf{LHS}}=\sem{\mathsf{RHS}}
\end{equation}
Here, $vars$ is the set of all variables appearing in the rule. It contains all possible tensor variables and operator attributes, such as slice attributes and expand shapes.
Note that we only consider the case where the term prior to rewriting is valid, i.e., when $\mathsf{LHS}$ is valid.

The challenge here is that a rewrite rule can usually be applied to tensors with arbitrary number of axes and arbitrary sizes in each axis.
To handle arbitrary number of axes, as discussed in \secref{dsl}, we express rewrite rules with a finite number of \agg{}-axes and each of them may be instantiated to arbitrary ranks.
The equivalence of the two expressions then boils down to verifying that they are equivalent under all \emph{valid} instantiations $I$, where the $\mathsf{LHS}$ is valid:
\begin{equation*}
\begin{aligned}
    \bigwedge_{I}\mathsf{valid}(I)&&\mathrm{where}&&
    \mathsf{valid}(I) = \forall v\in vars(I), C\wedge\mathsf{valid\hbox{-}expr}(\mathsf{LHS}(I))\to\sem{\mathsf{LHS}(I)}=\sem{\mathsf{RHS}(I)}
\end{aligned}
\end{equation*}

For a given instantiation, we now have concrete ranks but \emph{unbounded sizes}.
To handle arbitrary sizes, we model each tensor as an uninterpreted function and use unbounded integers to model the indices and sizes in each \ind{}-axis.
This is supported natively by SMT solvers with good performance.
We then leverage a symbolic execution approach to convert the rewrite rule into a set of constraints.
We call this a \emph{bounded-verification} proof obligation, which is used to verify a \emph{given instantiation} of our rewrite rule.
This is described in detail in \secref{bounded-verification}.

We still need to answer two questions: (1) what does it mean for an instantiation to be valid and (2) how to verify the rule for a possibly infinite number of instantiations?
Our approach to these questions is to introduce a new concept called \emph{rank class} (\rclass{}).
A \emph{rank class} is a (required) property of an \agg{}-axis, such that all the \agg{}-axes with the same \rclass{} are always instantiated to the same rank.
An instantiation of a rewrite rule can then be expressed as instantiating these \rclasses{} to some rank.
We then derive a \emph{sufficient} rank for each \rclass{}, such that any instantiation with higher ranks could be proven inductively, given that we have verified all instances within the rank.
This approach follows the $k$-induction technique~\cite{sheeran2000checking}.

To show the intuition for an \rclass{}, consider the rewrite rule $(A+A^\top)^\top=A+A^\top$, where $t$ has the shape $\{x_1\mapsto m_1, x_2\mapsto m_2\}$.
Note that $\transpose$ is a no-op with unordered axes semantics, and $\relabel$ is provided for renaming and matching axes.
\begin{equation*}
\begin{gathered}
\relabel(\binary(+, t, \relabel(t, \{x_1\mapsto x_2, x_2\mapsto x_1\})), \{x_1\mapsto x_2, x_2\mapsto x_1\}) \Rightarrow \\
\binary(+, t, \relabel(t, \{x_1\mapsto x_2, x_2\mapsto x_1\}))
\end{gathered}
\end{equation*}
It is easy to see that we must instantiate $x_1$, $x_2$ with the same number of axes to make sure that both \textsf{LHS} and \textsf{RHS} are valid.
In other words, $x_1$ and $x_2$ must have the same rank, in which case we say that these two aggregated axes are in the same \emph{rank class} (\rclass{}) and
the \rclass{} constraints the possible instantiations.
When we instantiate a rule with relabeling, we also need to be able to establish a consistent mapping between the \ind{}-axes of $x_1$ and $x_2$. These facts are expressed in the following two definitions.

\begin{definition}
A \emph{rank class (\rclass{})} $c$ is a property of a family of \agg{}-axes, such that
\begin{itemize}
    \item Each \agg{}-axis $x$ is in exactly one \rclass{} $c$, written as $x:c$.
    \item For all $x_0:c$, $x_1:c$, $x_0$ and $x_1$ are instantiated to the same rank $r$ in a valid instantiation. In such an instantiation, the rank of the \rclass{} $c$ is defined to be $r$.
\end{itemize}
\end{definition}

\begin{definition}
An \rclass{} $c$ provides a canonical bijection mapping between each pair of the aggregated axes associated with it.
We denote such a mapping that maps from the named axes in $x_0:c$ to $x_1:c$, as $\mathsf{MapAxes}(c, x_0, x_1)$. The canonical mapping must satisfy:
\begin{itemize}
    \item $\mathsf{MapAxes}(c, x_0, x_1) \circ \mathsf{MapAxes}(c, x_1, x_0) = \mathsf{id}$, and
    \item $\mathsf{MapAxes}(c, x_1, x_2) \circ \mathsf{MapAxes}(c, x_0, x_1) = \mathsf{MapAxes}(d, x_0, x_2)$.
\end{itemize}
\end{definition}

where $\circ$ refers to function composition.
Such a mapping can be trivially constructed during the instantiation as the aggregated axes of the same rank class are instantiated to the same rank.
In practice, without loss of generality, to instantiate an aggregated axis $x_i:c$ to rank $r$, we can instantiate it to the set of axes $a_{i,0},\cdots,a_{i,r-1}$. The canonical mapping between two aggregated axes $x_i:c$ and $x_j:c$ can then be established as $\{a_{i,k}\mapsto a_{j,k}\mid k\in\{0\cdots r-1\}\}$, and
the semantics of $\relabel$ follows this mapping to relabel the instantiated axes.

We can then reduce our problem to verifying the rule for all possible \rclass{} instantiations:
\begin{equation} \eqnlabel{valid-rule-rank}
\begin{aligned}
    \bigwedge_{r_1, r_2, \cdots}\mathsf{valid}(\overline{\{c_i \mapsto r_i\}})
\end{aligned}
%\vspace{-0.5em}
\end{equation}
where $\overline{\{c_i \mapsto r_i\}}$ is a map containing the ranks for all \rclasses{} $c_1, c_2, \cdots$ in the rule.
$\mathsf{valid}(\overline{\{c_i \mapsto r_i\}})$ denotes the proof obligation for a concrete-ranked instance, where the \rclass{} $c_i$ is instantiated to rank $r_i$, for all $i \in \{1 \cdots p\}$.
We still need to verify the rule for \emph{all possible} ranks of all \rclasses{}.

\paragraph{Our Approach:} To simplify the discussion, let's assume that the rule only has one \agg{}-axis $x$, with \rclass{} $c$ (i.e., $x : c$). We can then rewrite \eqnref{valid-rule-rank} as
\begin{equation} \eqnlabel{valid-rule-simple}
    \bigwedge_{i} ~ \mathsf{valid}(i)
\end{equation}
where $\mathsf{valid}(i)$ (short for $\mathsf{valid}(\{c \mapsto i\}))$ is true if and only if the rule is valid when rank of $c$ (and $x$) is $i$.
We observe that for every \rclass{} in a rule, there exists a \emph{bound} corresponding to a sufficient rank required for unbounded verification, i.e., there exists a rank $k$, such that,
\begin{equation} \eqnlabel{induction-hypothesis}
    \forall i \geq k, ~ \mathsf{valid}(i) \rightarrow \mathsf{valid}(i+1)
\end{equation}
This means that for all $i \geq k$, if the rule is valid when $c$ has rank $i$, then the rule is valid when $c$ has rank $i+1$.
Given such a $k$, we can do unbounded verification for the rule using $k$-induction:
\begin{itemize}
    \item Basis: Use bounded verification to prove that the rule is valid for all ranks until $k$:
    \[
        \bigwedge_{i = 1 \cdots k} \mathsf{valid}(i)
    \]
    \item Induction case: Use induction on the rank of $c$, with \eqnref{induction-hypothesis} as induction hypothesis:
    \begin{equation*}
        \mathsf{valid}(k) \wedge \left[\forall i \geq k, ~ \mathsf{valid}(i) \rightarrow \mathsf{valid}(i+1) \right] ~
        \Rightarrow~\bigwedge_{i \geq k} \mathsf{valid}(i)
    \end{equation*}
\end{itemize}
This would imply that the rule is correct for an arbitrary number of \ind{}-axes in $c$.
Now what remains is finding a sufficient rank $k$ for any \rclass{} in a rule.
We show how to derive such a bound with the help of an example rewrite rule.

\subsection{Bound Computation Example} \seclabel{bound-example}

\newcommand{\padcombine}[0]{\textsc{PadLowCombine}}

Consider an input tensor \textsf{Y} which has one \agg{}-axis, say $x$, having the \rclass{} $c$ (i.e. $x : c$). The rewrite rule \padcombine{} is shown below:
\begin{equation} \eqnlabel{bound-example-rule}
    \padlow(\padlow(\mathsf{Y}, 0, L_1), 0, L_2) \Rightarrow_{L_1 \geq 0 \lland L_2 \geq 0} \padlow(\mathsf{Y}, 0, L_1+L_2)
\end{equation}
where $L_1=\{x\mapsto l_1\}$ and $L_2=\{x \mapsto l_2\}$, for some maps $l_1$ and $l_2$.

The \padcombine{} rule merges two $\padlow$ operators into a single $\padlow$ operator.
The precondition requires that both padding attributes should be non-negative.
For this rule, the precondition also implies that the \textsf{LHS} expression is valid.
%\jai{This precondition is not the weakest. WP is $\neg (L_1 < 0 \land L_2 \geq 0)$} \charith{Can't we add this as a case study in the evaluation}.
The variables appearing in this rule are the padding attributes $L_1, L_2$ and the input tensor \textsf{Y}. Using \eqnref{valid-rule}, we can express the validity condition for this rewrite rule as
%\jai{what about \textsf{valid-expr}(\textsf{LHS})? The precondition subsumes it, I think?}
\begin{equation} \eqnlabel{pad-rule}
    \forall ~ \mathsf{Y}, L_1, L_2, ~ L_1 \geq 0 \land L_2 \geq 0 \rightarrow \sem{\mathsf{LHS}} = \sem{\mathsf{RHS}}
\end{equation}

It is easy to see that both \textsf{LHS} and \textsf{RHS} have the same shape, i.e., they have the same domain of accesses.
Thus, we can rewrite $\sem{\mathsf{LHS}} = \sem{\mathsf{RHS}}$ by interpreting the tensors under a general, valid access to these tensors.
Let $A \in \mathsf{Access}(\mathsf{LHS})$ be an arbitrary access from the domain of the output tensors.
It has the form $A=\{x\mapsto a\}$, where $a$ maps named-axes in $x$ to (symbolic) indices.
We use the denotational semantics of $\padlow$ described in \secref{denotational-sem}, to symbolically execute these expressions.
\begin{align}
    & ~ \sem{\mathsf{LHS}} = \sem{\mathsf{RHS}} \nonumber \\ 
    \Leftrightarrow & ~ \forall A, ~ \sem{\mathsf{LHS}}[A] = \sem{\mathsf{RHS}}[A] \nonumber \\
    \Leftrightarrow & ~ \forall A, ~ \sem{\padlow(\padlow(\mathsf{Y}, 0, L_1), 0, L_2)}[A] = \sem{\padlow(\mathsf{Y}, 0, L_1 + L_2)}[A] \nonumber \\
    %\Leftrightarrow & ~ \forall A, ~ \lif \ A \geq L_2 \ \lthen \ \sem{\padlow(\mathsf{Y}, 0, L_1)}[A - L_2] \ \lelse \ 0 = \nonumber \\
    %& \quad\quad\quad\quad\quad\quad \lif \ A \geq L_1 + L_2 \ \lthen \ \mathsf{Y}[A - L_1 - L_2] \ \lelse \ 0 \nonumber \\
    \Leftrightarrow & ~ \forall A, ~ \lif \ A \geq L_2 \ \lthen \ (\lif \ A - L_2 \geq L_1 \ \lthen \ \mathsf{Y}[A - L_1 - L_2] \ \lelse \ 0) \ \lelse \ 0 = \nonumber \\
    & \quad\quad\quad\quad\quad\quad \lif \ A \geq L_1 + L_2 \ \lthen \ \mathsf{Y}[A - L_1 - L_2] \ \lelse \ 0 \nonumber \\
    \Leftrightarrow & ~ \forall a, ~ \lif \ a \geq l_2 \ \lthen \ (\lif \ a - l_2 \geq l_1 \ \lthen \ \mathsf{Y}[x \mapsto a - l_1 - l_2] \ \lelse \ 0) \ \lelse \ 0 = \nonumber \\
    & \quad\quad\quad\quad\quad\quad \lif \ a \geq l_1 + l_2 \ \lthen \ \mathsf{Y}[x \mapsto a - l_1 - l_2] \ \lelse \ 0 \eqnlabel{pad-expanded}
\end{align}
In the last step, we made use of the fact that the aggregated maps like $A, L_1, L_2$ contain only one \agg{}-axis $x$.
\eqnref{pad-expanded} holds for any number of \ind{}-axes in $x$.

\definecolor{accesscolor}{HTML}{330066}
\paragraph{Observation} We can syntactically partition the above equation as follows:
\begin{align*}
    & ~ \forall a, ~ \textcolor{blue}{\lif} \ \textcolor{magenta}{a \geq l_2} \ \textcolor{blue}{\lthen} \ (\textcolor{blue}{\lif} \ \textcolor{magenta}{a - l_2 \geq l_1} \ \textcolor{blue}{\lthen} \ \textcolor{accesscolor}{\mathsf{Y}[}\textcolor{red}{x \mapsto a - l_1 - l_2}\textcolor{accesscolor}{]} \ \textcolor{blue}{\lelse} \ \textcolor{blue}{0}) \ \textcolor{blue}{\lelse} \ \textcolor{blue}{0} = \\
    & \quad\quad\quad\quad\quad\quad \textcolor{blue}{\lif} \ \textcolor{magenta}{a \geq l_1 + l_2} \ \textcolor{blue}{\lthen} \ \textcolor{accesscolor}{\mathsf{Y}[}\textcolor{red}{x \mapsto a - l_1 - l_2}\textcolor{accesscolor}{]} \ \textcolor{blue}{\lelse} \ \textcolor{blue}{0} \eqnlabel{pad-expanded}
\end{align*}
We explain each part below: 

$\textcolor{accesscolor}{\mathsf{Y}[\_]}$: represents accesses to the tensor \textsf{Y}.

$\textcolor{red}{x \mapsto a - l_1 - l_2}$: represents an \emph{access expression} for a tensor access. In this case, $\textcolor{red}{a - l_1 - l_2}$ is the access map for the \agg{}-axis $x$. The rank of this access map depends on the number of \ind{}-axes in $x$. We observe that we can rewrite this expression as follows:
    \begin{equation*}
        \textcolor{red}{a - l_1 - l_2} = \textcolor{red}{\fmap(e, a, l_1, l_2)}
    \quad\mathrm{where}\quad
        e\df\lambda v, p, p'. (v - p - p')
    \end{equation*}
    where $\fmap$ takes a function and applies it to a list of maps. For instance, if $m = \{i \mapsto v_i, j \mapsto v_j\}$ and $f = \lambda v. (v+1)$, then $\fmap(f, m) = \{i \mapsto v_i + 1, j \mapsto v_j + 1\}$.
Here, $e$ is independent of the rank of $x$ and only $a,l_1,l_2$ change as the rank of $x$ changes. Thus, we are able to capture all the rank-independent information in the function $e$. We call such a function an \emph{\indtrans{}} because it transforms output index-values to input index-values.

$\textcolor{magenta}{a \geq l_1 + l_2}$: these are boolean values, referred to as \emph{\cond{}s}, occurring inside \textsf{if-then-else} blocks.
They capture the dependency of the output tensor value on the input tensor values, based on the value of the access.
They originate from the operator semantics.
For instance, $not\hbox{-}pad$ in the $\padlow$ semantics takes an access $A$ and tells if it lies in the padded area or not.
We observe that we can rewrite this condition as follows:
\begin{equation*}
        \textcolor{magenta}{a \geq l_1 + l_2} = \textcolor{magenta}{\fold(g_1, a, l_1, l_2)}
        \quad\mathrm{where}\quad
        g_1 \df \lambda v, p, p'. (v \geq p + p')
\end{equation*}
    where $\fold$ takes a boolean valued function, applies it to a list of maps, and returns \textsf{true} if all values are \textsf{true}, \textsf{false} otherwise.
    $\fold$ can be defined as:
    \[
        \fold(g, m_1, m_2, \cdots) = \bigwedge_{i \in \mathsf{dom}(m_1)} g(m_1(i), m_2(i), \cdots)
    \]
    Similarly, we can write $\textcolor{magenta}{a \geq l_2}$ as $\textcolor{magenta}{\fold(g_2, a, l_2)}$, where $g_2 \df \lambda v, p. (v \geq p)$.
    Here, $g_1$ and $g_2$ are independent of the rank of $x$ and only $a,l_1,l_2$ change as the rank of $x$ changes.
    We capture all the rank-independent information in the functions $g_1$ and $g_2$.

$\textcolor{blue}{\lif} \ \_ \ \textcolor{blue}{\lthen} \ (\textcolor{blue}{\lif} \ \_ \ \textcolor{blue}{\lthen} \  \_ \ \textcolor{blue}{\lelse} \ \textcolor{blue}{0}) \ \textcolor{blue}{\lelse} \ \textcolor{blue}{0} = \textcolor{blue}{\lif} \ \_ \ \textcolor{blue}{\lthen} \ \_ \ \textcolor{blue}{\lelse} \ \textcolor{blue}{0}$: this is a function which returns a boolean, denoting if the values of \textsf{LHS} and \textsf{RHS} at the access $A$ are equal or not. We call this $\scalarf$ since it contains the core, \emph{scalar} computation in the expressions and may consist of arithmetic and conditionals. For this rule, we can define $\scalarf$ as
    \begin{equation*}
    \begin{gathered}
        \scalarf(y, b_1, b_2) \df (\textcolor{blue}{\lif} \ b_2 \ \textcolor{blue}{\lthen} \ (\textcolor{blue}{\lif} \ b_1 \ \textcolor{blue}{\lthen} \ y \ \textcolor{blue}{\lelse} \ \textcolor{blue}{0}) \ \textcolor{blue}{\lelse} \ \textcolor{blue}{0}) = (\textcolor{blue}{\lif} \ b_1 \ \textcolor{blue}{\lthen} \ y \ \textcolor{blue}{\lelse} \ \textcolor{blue}{0})\\
        y \df \textcolor{accesscolor}{\mathsf{Y}[}\textcolor{red}{x \mapsto \fmap(e, a, l_1, l_2)}\textcolor{accesscolor}{]} \quad
        b_1 \df \textcolor{magenta}{\fold(g_1, a, l_1, l_2)} \quad
        b_2 \df \textcolor{magenta}{\fold(g_2, a, l_2)}
    \end{gathered}
    \end{equation*}
    Based on the arguments to $\scalarf$, we say $\scalarf$ has \emph{1 access to \textsf{Y}} and \emph{2 \cond{}s}.
    Note that there are 2 occurrences each of $y$ and $b_1$ in the $\scalarf$ but we only care about distinct accesses and \cond{}s.
    Just like \indtrans{}s, $\scalarf$ is also independent of the rank of $x$.
    This property will be crucial for our bound computation algorithm.
We use this observation to write \eqnref{pad-expanded} in terms of $\scalarf$ and substitute it back in \eqnref{pad-rule} to get the validity condition of the rule:
\begin{align} \eqnlabel{pad-scalarf}
    & \forall ~ \mathsf{Y}, l_1, l_2, ~ l_1 \geq 0 \land l_2 \geq 0 \rightarrow \nonumber \\
    & \quad\quad\forall a, ~ \scalarf(\mathsf{Y}[x \mapsto \fmap(e, a, l_1, l_2)], \ \fold(g_1, a, l_1, l_2), \ \fold(g_2, a, l_2))
\end{align}

\paragraph{Bound Computation}
We first look at the validity condition of the rule when $x$ is instantiated with some rank $i$, i.e., $\mathsf{valid}(i)$.
We instantiate the \agg{}-axis $x$ to $x^i$, which contains $i$ \ind{}-axes, say $\{\underline{1}, \cdots, \underline{i}\}$.
We also instantiate the input tensor, padding attributes, and the general access map.
Thus, we can rewrite $\mathsf{valid}(i)$ as
\begin{align} \eqnlabel{pad-scalarf-i}
    & \forall ~ \mathsf{Y}^i, l_1^i, l_2^i, ~ l_1^i \geq 0 \land l_2^i \geq 0 \rightarrow \nonumber \\
    & \quad\quad\forall a^i, ~ \scalarf(\mathsf{Y}^i[x^i \mapsto \fmap(e, a^i, l_1^i, l_2^i)], \ \fold(g_1, a^i, l_1^i, l_2^i), \ \fold(g_2, a^i, l_2^i))
\end{align}
As noted before, $\scalarf$, $e$, $g_1$, and $g_2$ are independent of the rank of $x$, so they remain unchanged irrespective of the value of $i$.
We want to find a $k$ such that \eqnref{induction-hypothesis} holds.
The idea is to start with [$\mathsf{valid}(k) \rightarrow \mathsf{valid}(k+1)$] and try to find a $k$ which satisfies the induction hypothesis.
We instead work with its contrapositive,
\begin{equation*}
    \mathsf{valid}(k) \rightarrow \mathsf{valid}(k+1)
    \Leftrightarrow \neg \mathsf{valid}(k+1) \rightarrow \neg \mathsf{valid}(k)
\end{equation*}

Intuitively, $\neg\mathsf{valid}(i)$ is true if there is a counterexample for the rule at rank $i$.
We want to find a sufficient $k$ such that we can \emph{lower} a counterexample at rank $k+1$ (and all higher ranks) to a counterexample at rank $k$. 
On expanding the validity conditions using \eqnref{pad-scalarf-i}, we get:
\begin{align*}
    & \exists ~ \mathsf{Y}^{k+1}, l_1^{k+1}, l_2^{k+1}, ~ l_1^{k+1} \geq 0 \land l_2^{k+1} \geq 0 \bigwedge \exists a^{k+1},\\
    & \quad\quad\neg\scalarf(\mathsf{Y}^{k+1}[x^{k+1} \mapsto \fmap(e, a^{k+1}, l_1^{k+1}, l_2^{k+1})], \ \fold(g_1, a^{k+1}, l_1^{k+1}, l_2^{k+1}), \ \fold(g_2, a^{k+1}, l_2^{k+1})) \\
    & \phantom{\forall ~ \mathsf{Y}^k, l_1^k, l_2^k, ~ l_1^k \geq 0 \land l_2^k \geq 0 \rightarrow \bigwedge \exists a^k}\quad\quad\quad\quad{\big\downarrow} \\
    & \exists ~ \mathsf{Y}^k, l_1^k, l_2^k, ~ l_1^k \geq 0 \land l_2^k \geq 0 \bigwedge \exists a^k, \\
    & \quad\quad\neg\scalarf(\mathsf{Y}^k[x^k \mapsto \fmap(e, a^k, l_1^k, l_2^k)], \ \fold(g_1, a^k, l_1^k, l_2^k), \ \fold(g_2, a^k, l_2^k))
\end{align*}
This means:
\begin{itemize}
    \item We are given a tensor $\mathsf{Y}^{k+1}$ which has $k+1$ \ind{}-axes in $x^{k+1}$, and whose shape is of the form $\shape(\mathsf{Y}) = \{x^{k+1} \mapsto m\}$, where $m = \{\underline{1} \mapsto n_1, \ \cdots, \ \underline{k+1} \mapsto n_{k+1}\}$
    % \begin{equation*}
    %     m = \{\underline{1} \mapsto n_1, \ \cdots, \ \underline{k+1} \mapsto n_{k+1}\}
    % \end{equation*}
    \item We are given padding attributes $l_1^{k+1}$ and $l_2^{k+1}$ such that the precondition is satisfied.
    \item We are given a map $a^{k+1}$ such that output tensors do not match at the access $\{x^{k+1} \mapsto a^{k+1}\}$.
    \item We then need to construct a tensor $\mathsf{Y}^{k}$ which has $k$ \ind{}-axes in $x^{k}$.
    We also need to construct padding attributes $l_1^k$ and $l_2^k$ such that the precondition is still satisfied, and a map $a^k$ such that the output tensors do not match at the access $\{x^k \mapsto a^k\}$.
\end{itemize}

\paragraph{Counterexample Construction}
Our counterexample construction algorithm involves \emph{projecting} the $(k+1)$-ranked \rclass{} to a $k$-ranked \rclass{}, i.e., we would choose $k$ \ind{}-axes from $\{\underline{1}, \cdots, \underline{k+1}\}$. There are $k+1$ such projections but all projections may not lead to a counterexample.
We express our construction through a set of equations and derive constraints on the projection.

Let $\Gamma \subset \{\underline{1}, \cdots, \underline{k+1}\}$ be a projection of size $k$. The \ind{}-axes in $\Gamma$ are currently unknown.
We can then express the $k$-ranked attributes as follows: $\shape(\mathsf{Y}^k) = \{x^k \mapsto m|_{\Gamma}\}$, $a^k = a^{k+1}|_{\Gamma}$, $l_1^k = l_1^{k+1}|_{\Gamma}$, and $l_2^k = l_2^{k+1}|_{\Gamma}$. These are unknown as well.
To construct a $k$-ranked counterexample, we first make sure that the arguments to $\scalarf$ have the same values in both ranks.
We equate $\scalarf$ arguments in the $k$-ranked counterexample to the corresponding arguments in the $(k+1)$-ranked counterexample and collect \emph{constraints} on $\Gamma$.
A constraint is a \ind{}-axis that needs to be in the projection for the $k$-ranked counterexample to exist. Thus,
\begin{itemize}
    \item $\fold(g_1, a^k, l_1^k, l_2^k)$ and $\fold(g_1, a^{k+1}, l_1^{k+1}, l_2^{k+1})$ need to be equisatisfiable.
    Let $C_1$ be the set of constraints we get from this equation.
    We first expand the definition of $\fold$,
    \begin{align*}
        \fold(g_1, a^{k+1}, l_1^{k+1}, l_2^{k+1}) &= \bigwedge_{i = 1}^{k+1} a^{k+1}(\underline{i}) \geq l_1^{k+1}(\underline{i}) + l_2^{k+1}(\underline{i}) \\
        \fold(g_1, a^k, l_1^k, l_2^k) &= \bigwedge_{\underline{j} \in \Gamma} a^k(\underline{j}) \geq l_1^{k}(\underline{j}) + l_2^{k}(\underline{j})
    \end{align*}
    % \begin{align*}
    %     \fold(g_1, a^{k+1}, l_1^{k+1}, l_2^{k+1}) &= \bigwedge_{i = 1}^{k+1} g_1(a^{k+1}(\underline{i}), l_1^{k+1}(\underline{i}), l_2^{k+1}(\underline{i})) \\
    %     \fold(g_1, a^k, l_1^k, l_2^k) &= \bigwedge_{\underline{j} \in \Gamma} g_1(a^k(\underline{j}), l_1^{k}(\underline{j}), l_2^{k}(\underline{j}))
    % \end{align*}
    Let $b = \fold(g_1, a^{k+1}, l_1^{k+1}, l_2^{k+1})$, which contains $k+1$ clauses, and $b' = \fold(g_1, a^k, l_1^k, l_2^k)$, which contains $k$ clauses.
    The value of $b$ is known since it depends entirely on the $(k+1)$-ranked counterexample.
    Let $r$ be the number of clauses in $b$ which evaluate to \textsf{true}.
    The remaining $(k+1) - r$ clauses evaluate to \textsf{false}.
    We do a case analysis on $r$:

    \begin{itemize}
        \item $r < k$: $b$ is \textsf{false} for this case. We want $b'$ to be \textsf{false} as well.
        We can see that for any projection $\Gamma$, $b'$ will be \textsf{false}.
        There are no constraints in this case, so $C_1 = \emptyset$.
        
        \item $r = k$: $b$ is \textsf{false} for this case. We want $b'$ to be \textsf{false} as well. There is exactly one \ind{}-axis, say $\underline{l}$, for which $a^{k+1}(\underline{l}) \geq l_1^{k+1}(\underline{l}) + l_2^{k+1}(\underline{l})$ is \textsf{false}.
        $\underline{l}$ needs to be in the projection for $b'$ to be $\textsf{false}$.
        Leaving out $\underline{l}$ will make $b'$ \textsf{true}, which is not desirable.
        For this case, $C_1 = \{\underline{l}\}$
        
        \item $r = k+1$: $b$ is \textsf{true} for this case, so we want $b'$ to be \textsf{true} as well.
        We can see that for any projection $\Gamma$, $b'$ will be \textsf{true}.
        There are no constraints in this case, so $C_1 = \emptyset$.
    \end{itemize}
    
    As seen above, we get at most one constraint from this equation, so $|C_1| \leq 1$.
    
    \item $\fold(g_2, a^k, l_2^k)$ and $\fold(g_2, a^{k+1}, l_2^{k+1})$ need to be equisatisfiable.
    Let $C_2$ be the set of constraints we get from this equation.
    We do a similar analysis and get at most one constraint from this equation, so $|C_2| \leq 1$.
    The \ind{}-axes in $C_2$ may or may not be same as \ind{}-axes in $C_1$.
    
    \item $\mathsf{Y}^k[x^k \mapsto \fmap(e, a^k, l_1^k, l_2^k)]$ needs to be set to $\mathsf{Y}^{k+1}[x^{k+1} \mapsto \fmap(e, a^{k+1}, l_1^{k+1}, l_2^{k+1})]$.
    This does not introduce any constraint, irrespective of the projection.
    %Currently, the $\scalarf$ for this rule only had 1 access to \textsf{Y}.
    There could have been constraints introduced if the $\scalarf$ had more than 1 access to \textsf{Y}.
    We discuss more about the general case in \secref{bound}.
\end{itemize}

The final set of constraints is computed as $C = C_1 \cup C_2$.
If $|C| > k$, then we cannot get a valid projection.
Thus, we need $|C| \leq k$ for a valid counterexample lowering. 
We know that
$|C| = |C_1 \cup C_2| \leq |C_1| + |C_2| \leq 2$.
From this, we get $2 \leq k$ as a sufficient condition for a valid counterexample lowering. Finally, we can fully construct the $k$-ranked counterexample as follows:
\begin{itemize}
    \item Projection: The \ind{}-axes in $C$ need to be a part of the projection $\Gamma$, but the other axes are unspecified.
    To get the final projection $\Gamma$, extend $C$ by any $k - |C|$ \ind{}-axes from $\{\underline{1}, \cdots, \underline{k+1}\} \mysetminus C$.
    \item Tensor shapes and attributes: Compute the $k$-ranked attributes as follows:
    $\shape(\mathsf{Y}^k) = \{x^k \mapsto m|_{\Gamma}\}$, $a^k = a^{k+1}|_{\Gamma}$, $l_1^k = l_1^{k+1}|_{\Gamma}$, and $l_2^k = l_2^{k+1}|_{\Gamma}$.
    This ensures that the shape and the access map are valid and the precondition is satisfied in the $k$-ranked counterexample.
    \item Tensor values: Let $v = \mathsf{Y}^{k+1}[x^{k+1} \mapsto \fmap(e, a^{k+1}, l_1^{k+1}, l_2^{k+1})]$.
    We only require $\mathsf{Y}^k$ to have the value $v$ at the access $\{x^k \mapsto \fmap(e, a^k, l_1^k, l_2^k)\}$ and the values at other points are unspecified.
\end{itemize}

Therefore, for all $k \geq 2$, $\mathsf{valid}(k)$ implies $\mathsf{valid}(k+1)$.
This allows us to reduce the unbounded-verification proof obligation to two bounded-verification proof obligations: $\mathsf{valid}(1)$ and $\mathsf{valid}(2)$.

%\vspace{-0.5em}
\subsection{Bound Computation for the General Case}\seclabel{bound}

In \secref{bound-example}, we derived a bound for the \padcombine{} rule and reduced the unbounded-verification proof obligation to bounded-verification proof obligations.
This section first presents a theorem for unbounded verification for a general rewrite rule and an algorithm to compute a bound in the general case.
We then briefly discuss how we handle the complexities of the general case.
The detailed proof is in \secref{proof}.
The \textsc{InferBound} routine is described in \algref{infer-bound}.

\begin{lemma}\lemmalabel{inferbound}
Let $R$ be any rewrite rule written in our DSL. Let $m$ be a map containing ranks of \rclasses{} in $R$.
For any \rclass{} $c$ in the rule, if $k = \textsc{InferBound}(R, c)$, then
\begin{equation*}
    \forall i \geq k, ~ \mathsf{valid}(m[c \mapsto i]) \rightarrow \mathsf{valid}(m[c \mapsto i+1])
\end{equation*}
\end{lemma}
Here, $m[c \mapsto j]$ denotes the map where the rank of $c$ is updated to $j$, while all other ranks are unchanged.
In other words, for all $i \geq k$ and for \emph{any} ranks of the other \rclasses{}, if the rule is valid when $c$ has rank $i$, then it implies that the rule is valid when $c$ has rank $i+1$.
\lemmaref{inferbound} allows us to use $k$-induction on the \rclass{} ranks to verify the rule for arbitrary ranks of all \rclasses{}.

%\charith{should this be theorem 1 or 2?}
\begin{theorem}
Let $R$ be any rewrite rule written in our DSL. If $c_1 \cdots c_p$ are the \rclasses{} appearing in the rule and $k_i = \textsc{InferBound}(R, c_i)$ for all $i \in \{1 \cdots p\}$, then $R$ is a valid rule in the unbounded setting if and only if
\begin{equation*}
    \bigwedge_{1 \leq r_1 \leq k_1} \cdots \bigwedge_{1 \leq r_p \leq k_p} \mathsf{valid}(\{c_1 \mapsto r_1, \cdots, c_p \mapsto r_p\})
\end{equation*}
This follows from using \lemmaref{inferbound} as induction hypothesis for all \rclasses{}.
\end{theorem}

The \textsc{InferBound} routine in \algref{infer-bound} takes a rewrite rule $R$ and \rclass{} $c$ as input.
On line 3, we use the \textsc{TensorsWithRClass} subroutine to get all the input tensors which have an \agg{}-axis having the \rclass{} $c$.
We iterate through all the tensors in lines 4-9.
For each tensor, we use the \textsc{NumTensorAccess} subroutine to get the number of distinct accesses to that tensor and add its contribution to the bound.
On line 10, we use the \textsc{NumConds} subroutine to get the number of conditions having an \agg{}-axis with the \rclass{} $c$ and add it to the bound.
We also make sure that the computed bound is at least 1 since we do not want empty \agg{}-axes.

We now briefly discuss the complexities that we encounter while tackling a general rewrite rule and how the bound computed by this algorithm is sufficient.

\SetKwComment{Comment}{/*}{ */}
\SetKwInOut{Input}{Inputs}
\SetKwInOut{Output}{Output}
\begin{wrapfigure}{R}{0.56\textwidth}
%\vspace{-1em}
    \begin{minipage}{0.56\textwidth}
        \begin{algorithm}[H]
        \caption{Computing the bound for an \rclass{}} \alglabel{infer-bound}
        \Input{$\,$ Rewrite rule $R$ \& \rclass{} $c$}
        \Output{$\,$ $bound$, i.e., a sufficient rank for $c$}
        \SetKwProg{Fn}{Function}{ :}{}
        \Fn{\textsc{InferBound} ($R$, $c$)}{
            $bound \gets 0$\;
            $tensors \gets \textsc{TensorsWithRClass}(R, c)$\;
            \For{$t \in tensors$}{
                $n \gets \textsc{NumTensorAccess}(R, t)$\;
                \If{$n > 1$}{
                    $bound \gets bound + \binom{n}{2}$\;
                }
            }
            $bound \gets bound + \textsc{NumConds}(R, c)$\;
            \textbf{return} $\textsc{Max}(bound, 1)$ 
        }
        \textbf{End Function}
        \end{algorithm}
    \end{minipage}
\vspace{-1em}
\end{wrapfigure}

\paragraph{Arbitrary Operator Compositions}
For any rule written in our DSL, we can symbolically evaluate both \textsf{LHS} and \textsf{RHS} expressions under a general, valid access.
We observe that the result can be expressed in terms of a $\scalarf$ function, which could have any number of accesses and \cond{}s, similar to the \padcombine{} rule.
Therefore, we need to appropriately handle arbitrary number of accesses and conditions.

\vspace{-0.2em}
\paragraph{Arbitrary Number of Conditions}
%The \padcombine{} rule had 2 \cond{}s in the $\scalarf$. We analyzed them in isolation, and each of them can give at most 1 constraint.
In general, a rule could have an arbitrary number of \cond{}s, say $m$.
Similar to the \padcombine{} rule, we observe that any \cond{} introduces at most 1 constraint. 
We assume that all $m$ conditions are independent and we get $m$ constraints in the worst case.
%\mangpo{This assumes that one condition can introduce only one constraint. We should explain why this is true for all conditions.}

\vspace{-0.2em}
\paragraph{Arbitrary Number of Tensor Accesses}
%The \padcombine{} rule had 1 access to \textsf{Y} in the $\scalarf$, which did not introduce any constraints.
In general, a rule could have an arbitrary number of accesses, say $n$, to a tensor.
We need to make sure that during the projection, accesses containing unequal values do not get projected down to the same point in the tensor, since a point cannot have two different values.
We do a pairwise analysis of the accesses ($\binom{n}{2}$ such pairs) and each pair can give at most 1 constraint.
Thus, we get $\binom{n}{2}$ constraints from the accesses in the worst case. 

\vspace{-0.2em}
\paragraph{Arbitrary Number of \rclasses{}}
In general, a rule could have any number of \agg{}-axes and \rclasses{}.
This would require computing the minimum rank for each \rclass{}.
We do so by analyzing each \rclass{} in isolation, i.e., finding a sufficient $k$ for which we can do a $(k+1)$ to $k$ counterexample projection while keeping the ranks of other \rclasses{} unchanged.
The bound for this case would be $\binom{n}{2} + m$, where $n$ is the number of accesses and $m$ is the number of conditions in which the \rclass{} appears.
The computed bound is independent of the ranks of the other \rclasses{}.
Therefore, to ensure correctness in the unbounded setting in presence of multiple \rclasses{}, bounded-verification instances corresponding to all possible combinations of \rclass{} ranks within the bounds need to be verified.
This means that if $c_1 \cdots c_p$ are the \rclasses{} appearing in a rule $R$ and $k_i = \textsc{InferBound}(R, c_i)$ for all $i \in \{1 \cdots p\}$, then we need to verify $\prod_{i = 1}^p k_i$ number of bounded-verification instances.
%\mangpo{Conclude that hence the product.}

\vspace{-0.2em}
\paragraph{Arbitrary Number of Input Tensors}
In general, a rule could have any number of input tensors.
We still analyze each \rclass{} in isolation, but while counting the number of accesses, we only consider accesses to the tensors which contain that \rclass{}.
We also make the observation that accesses across tensors do not lead to any constraints, which allows us to do a pairwise analysis of accesses per tensor.
We then add the contribution of each tensor to the bound.
This is necessary since we take the union of all constraints from all tensors.

\subsection{Bounded Verification of Rewrite Rules} \seclabel{bounded-verification}

We reduced the unbounded-verification proof obligation to a finite set of bounded-verification proof obligations in \secref{bound}.
\project{} infers a sufficient rank for every \rclass{} and instantiates them with all ranks up to that bound, so we end up with fixed-rank but arbitrary-sized tensors.
We handle tensors of unbounded size using uninterpreted functions from accesses to values. For any rewrite rule, \project{} symbolically executes the \textsf{LHS} and \textsf{RHS} tensor expressions using operator semantics and interprets them under a general access with symbolic indices, chosen from the domain of accesses of the two expressions. During the symbolic execution, \project{} decomposes the verification into two kinds of checks below.

\paragraph{Checks performed during symbolic execution}
Each tensor operator constructs the shape of the output tensor using the shapes of the input tensor(s).
\project{} checks that the final \textsf{LHS} and \textsf{RHS} tensors have the
same rank and the same \ind{}-axes,
i.e., \texttt{(= (axes lhs) (axes rhs))}.
This check is performed entirely by the symbolic execution engine, rather than by the solver.
    
\paragraph{Checks delegated to the SMT solver}
During symbolic execution, \project{} collects assertions which are then sent to an SMT solver as verification conditions.
These assertions check that the axes sizes for \textsf{LHS} and \textsf{RHS} are the same; that accesses fall within axes sizes; and that the values stored in tensor expressions are the same.
These verification conditions are described below:

\begin{itemize}
    \item Assertions related to axes sizes: under the precondition, assuming \textsf{LHS} is valid, assert that the \textsf{LHS} shape and \textsf{RHS} shape are the same.
    This can be represented as
    \texttt{(=> (\&\& precond lhsValid) (= (shape lhs) (shape rhs)))}.
    Asserting the equality of two shapes involves checking if the shapes have the same \ind{}-axes and the corresponding axes have the same sizes.
    As discussed, the former check is done entirely by the symbolic execution engine.
    However, the solver is needed for the latter check.
    
    \item Assertions related to access ranges: under the precondition, assuming \textsf{LHS} is valid, all valid accesses to \textsf{LHS} lead to valid accesses to \textsf{RHS}.
    This can be represented as
    \texttt{(=> (\&\& precond lhsValid lhsAccessValid) rhsAccessValid)}.
    
    \item Assertions related to final tensor expressions: under the precondition, assuming that \textsf{LHS} is valid, \textsf{RHS} should be valid and they contain the same values under a general access.
    This can be represented as
    \texttt{(=> (\&\& precond lhsValid) (\&\& rhsValid rewriteEquivalent))}.
\end{itemize}

If the symbolic execution and SMT checks succeed, the rule is deemed verified for that rank.

\subsection{Verifying Rules with Reduction Operators} \seclabel{verifyreduction}

As discussed in \secref{denotational-sem}, automatically verifying expressions with reductions is challenging because:
\begin{itemize}
    \item The sizes of reduced axes are unbounded.
    \item Reductions can be performed in multiple steps, such as when tiling the tensors or distributing reduction over concatenation. One such example is the rule $\reduce(\concat(\mathsf{A},\mathsf{B})) \Rightarrow \reduce(\concat(\reduce(\mathsf{A}),\reduce(\mathsf{B})))$.
\end{itemize}

In \project{}, the key idea to verify rewrite rules with reductions is to represent the reduction results as \emph{uninterpreted} reduction elements, $\mathsf{Red}_Xf(X)$ (see \secref{denotational-sem}).
We observe that the equivalence of two reduction elements $\mathsf{Red}_Xf(X)$ and $\mathsf{Red}_Yg(Y)$, can \emph{often} be proven by showing that the $\mathsf{LHS}$ and $\mathsf{RHS}$ are sums of the same values, i.e.,
$f(X)$ and $g(Y)$ represent the same (multi-)set.

One way to prove set equivalence is to establish a bijection between $X$ and $Y$ and show each pair of values in $f(X)$ and $g(Y)$ are equal, regardless of tensor instantiation and operator attributes.
In \project{}, the user provides a relation between $X$ and $Y$ as a hint.
We then use an SMT solver to verify that it is a bijection: (1) for all valid $x\in X$, a unique $y\in Y$ exists under the relation, and vice versa and (2) the relation can take on all valid $x\in X$ and $y\in Y$.
After establishing the bijection, we prove each pair of elements $f(x)$ and $g(y)$ are always the same, regardless of tensor instantiation.
Successfully passing the checks reduces our proof to the case discussed in \secref{bounded-verification}.

A majority of rules with reductions (13 out of 17) in our system can be proven by establishing the bijection \emph{with} user-provided \emph{hints}.
However, there are rules where bijectivity cannot be proven due to limitations of SMT solvers on quantified formulas (1 out of 17), or no such bijection relation exists due to the fact that the cardinalities of the sets of valid reduction indices in $\mathsf{LHS}$ and $\mathsf{RHS}$ are different (3 out of 17).
In these cases, it is up to the user to further complete the proof of set equivalence based on the verifier output.
Note that this is generally much easier than proving full correctness from scratch.

\section{Discussion} \seclabel{discussion}

\paragraph{Choice of Tensor Compiler}
We chose \xla{} since it is a production quality compiler and is integrated into leading ML frontend-frameworks like TensorFlow, PyTorch, and JAX.
The \xla{} compiler takes model graphs from these frontends and converts them into \xlahlo{}, which is much more expressive than these frameworks.
Operators in these frameworks either have direct counterparts in \xlahlo{} (e.g., $\concat$, $\expand$, $\slice$), or can be expressed using existing \xlahlo{} operators (e.g., \textsf{squeeze}, \textsf{shrink}, \textsf{split}).
Some \xlahlo{} operators are more general than their counterparts in other frameworks.
For instance, \textsf{tf.tensordot} in TensorFlow allows specifying only the contracting axes, whereas \textsf{DotGeneral} in \xlahlo{} allows specifying both contracting and batch axes.
\textsf{tf.pad} in TensorFlow and \textsf{Pad} in ONNX allow specifying only low and high padding attributes, whereas $\pad$ in \xlahlo{} allows specifying interior padding as well.
Moreover, many operators in the \textsf{jax.lax} module \cite{jaxlax} are thin wrappers around equivalent \xlahlo{} operators.
As for other IR frameworks like ONNX \cite{onnx}, their operators also are either similar to \xlahlo{} operators, or can be expressed using \xlahlo{} operators.
Therefore, \xlahlo{} supports a more general set of operators than other frameworks.

\paragraph{Minimum vs Sufficient Rank for Unbounded Verification}
\secref{bound} shows how we can use the \textsc{InferBound} routine to compute a bound for every \rclass{} in a rewrite rule and get a set of bounded-verification instances.
These bounded-verification instances are sufficient to imply correctness in the unbounded setting.
It is worth noting that this computed bound is only a \emph{sufficient} rank and not the \emph{minimum} rank required for such a property to hold.
For instance, in \secref{bound-example}, we compute a bound for the \textsc{PadLowCombine} rule by calculating the number of conditions and the number of accesses.
We assume the two conditions to be independent, each contributing 1 constraint in the worst case, hence getting 2 as the final bound.
However, we observe that given the precondition, the \cond{} $b_1$ implies the \cond{} $b_2$, which reduces the bound to 1.
In general, the conditions may have dependencies and such insights can help us derive a smaller, or maybe even the minimum bound.

\section{Evaluation}
\seclabel{eval}

We evaluated the \project{} verification framework on the following aspects:

\newcommand{\totalRulesXla}[0]{175}
\newcommand{\supportedTR}[0]{121}
\newcommand{\implementedTR}[0]{118}
\newcommand{\verifiedTR}[0]{115}
\newcommand{\failTR}[0]{6}
\newcommand{\supportedTASO}[0]{14}
\newcommand{\supportedPET}[0]{18}

\begin{itemize}
    \item \textbf{Q1:} How expressive is \dsl{} compared to other automatic tensor graph rewrite verification systems? (\secref{eval-expressive})
    \item \textbf{Q2:} How good is \project{} at performing unbounded verification? (\secref{eval-verification})
    \item \textbf{Q3:} Can \project{} be used to aid compiler developers in rapid development? (\secref{case-studies})
\end{itemize}

To answer these questions, we selected all rules from the \xla{}'s Algebraic Simplifier (\textsf{AS}) for evaluation.
The $\mathsf{AS}$ rewrite pass has \totalRulesXla{} rules. 
These rules are implemented to speed up execution and allow further optimizations like fusion in other compiler passes.

\subsection{Expressiveness of \dsl{}}
\seclabel{eval-expressive}

We compared the expressiveness of \project{} with two other automatic tensor graph rewrite engines, TASO~\cite{taso} and PET~\cite{pet}, across all the \totalRulesXla{} rules. 
These rules are categorized into 5 classes, as shown in \tabref{expressiveness}.
We assessed whether each system can represent rules from these classes.
While TASO and PET do not support automatic, unbounded verification, we evaluated whether they can express these rules.
We found that TASO can represent \supportedTASO{} rules, and PET can represent \supportedPET{} rules.
Comparatively, \project{} can represent \supportedTR{} rules.

We also checked whether TASO and PET can perform bounded verification on the rules they can represent.
TASO and PET can prove 6 and 16 rules, respectively.
In contrast, \project{} verified \verifiedTR{} rules in the unbounded setting. Verification statistics are detailed in \secref{eval-verification}.

%\vspace{-0.2em}
\paragraph{Categories.}
\tabref{expressiveness} categorizes all the rules into 5 classes and summarizes representable rules in the systems. Numbers in parentheses indicate verifiable rules, with \project{} supporting unbounded verification and others using bounded verification strategies.
There are two high-level classes: element-wise and non-element-wise rules.
Element-wise rules are expressed with element-wise arithmetic operations.
They are further divided into rules with basic operators (e.g., $+$, $*$, $\mathsf{div}$, $\mathsf{rem}$) and rules with advanced operators without good support by solvers (e.g., $\mathsf{exp}$, $\mathsf{power}$).
Non-element-wise rules involve operators that change axes sizes or perform reductions. This category is further divided into reductions (e.g., $\conv$, $\tdot$, and $\reduce$), layout-sensitive operators (e.g., $\reshape$, $\bitcast$), and others. We separately categorized rules with preconditions. Note that rules with preconditions usually perform non-element-wise operations.

\begin{table}[t]
    \caption{Number of supported rules per disjoint category. The numbers in parentheses indicate rules that have been implemented and verified. $^{\dagger}$ means verified using cvc5 solver.}
    \centering
    \footnotesize
    \begin{tabular}{r|cccc}
    \toprule
        \multirow{2}{*}{\bf Category (disjoint)} & \multicolumn{4}{c}{\bf Number of Supported Rules } \\
         & \xla{} & \project{} & TASO & PET  \\
        \midrule
        Elementwise: simple    & 56 & 46 (45) & 11 (5) & 11 (11) \\
        Elementwise: advanced  & 10 & 1 (1$^{\dagger}$)  & 0 (0) & 0 (0) \\
        Non-elementwise: reductions  & 22 & 20 (17) & 0 (0) & 2 (0) \\
        Non-elementwise: layout-sensitive & 29 & 0 (0) & 0 (0) & 2 (2) \\
        Non-elementwise: others & 58 & 54 (52) & 3 (1) & 3 (3) \\
        \midrule
        Rules with Preconditions & 61 & 40 (34) & 0 (0) & 0 (0) \\
        \midrule
        Total & \totalRulesXla{} & \supportedTR{} (\verifiedTR{}) & \supportedTASO{} (6) & \supportedPET{} (16) \\
        \bottomrule
    \end{tabular}
    \tablabel{expressiveness}
\vspace{-1em}
\end{table}

%\vspace{-0.2em}
\paragraph{\project{}}
We implemented \project{} in Haskell using the Grisette~\cite{grisette} symbolic evaluation engine.
We represented \supportedTR{} rules in \project{} and verified \verifiedTR{} of them in the unbounded setting,
using Z3~\cite{z3} for most verifications and cvc5~\cite{cvc5} for one advanced element-wise rule.

\tabref{expressiveness} shows that \project{} was able to represent 46 element-wise simple rules and verified 45 of them.
In the advanced element-wise class, \project{} was able to verify only one rule involving $\mathsf{exp}$ with cvc5.
The limitations are due to unsupported operators in Z3 or cvc5 (e.g., $\mathsf{log}$, $\mathsf{power}$) and the inability to reason about precision of floating point expressions efficiently.

Among non-element-wise rules, \project{} was able to represent 74 and verified 69 rules.
For reduction rules, 20 rules were represented, with 3 unproven due to insufficient normalization lemmas and handling cases with extra zeroes.
\dsl{} cannot represent layout-sensitive rules.
We further discuss this in \secref{limitation}.
The unsupported rules in the other category are due to unimplemented operators (e.g., $\mathsf{scatter}$, $\mathsf{gather}$), while
two rules failed to verify due to timeouts.

%\vspace{-0.2em}
\paragraph{Comparison to Prior Works}
We compared \project{} with TASO and PET, which use axiomatic and statistical proof mechanisms, respectively.
\tabref{expressiveness} shows that \project{} was able to support significantly more rules than TASO and PET.
The main hurdles for TASO and PET were unsupported operators, too strict operator definitions (e.g., not supporting all the attributes for $\tdot$ and $\conv$), and inability to handle preconditions.
Unlike \project{}, TASO and PET cannot express rules requiring preconditions, precluding them from supporting 61 \xla{} rules.

The axiomatic approach requires axioms to prove equivalence of tensor expressions. Out of the 14 representable rules, 8 rules needed new axioms in TASO.
It can be even more cumbersome for an axiomatic approach like TASO when we need new axioms with preconditions.

The statistical approach in PET has benefits and drawbacks.
It does not need verification conditions proven by SMT solvers, making it potentially more flexible for operations not modeled in SMT.
However, PET can only verify linear expressions, which limits its scope.
It is potentially feasible to add support for preconditions by only generating test inputs meeting the precondition.

\subsection{Verification Capabilities of \project{}}
\seclabel{eval-verification}

\paragraph{Experimental Setup}
All evaluations were conducted on a system equipped with an Intel Core i9-13900K processor and \SI{128}{GB} of RAM.
We supported boolean, integer, and real-valued tensors in our DSL, and we verified the rules for all valid tensor types for that rule.
The timeout per SMT solver query was set to 10 seconds.

Out of the \supportedTR{} rules that we can express in our DSL, we implemented \implementedTR{} rules and verified \verifiedTR{} rules in the unbounded setting.
\figref{verification-times} shows cumulative distribution of total verification times for the \verifiedTR{} verified rules.
\project{} was able to verify 108 rules under 1 second, with verification times ranging from a minimum of \SI{0.023}{\second} to a maximum of \SI{23.33}{\second}.
\figref{num-tasks} shows the number of bounded-verification proof obligations (tasks) discharged for the verified rules.
The number of tasks is simply the product of the computed bounds of all \rclasses{} in a rule.
110 of the rules only required 1 task to guarantee correctness in the unbounded setting.

\project{} was unable to verify \failTR{} rules.
This included 3 timeouts, mainly due to those rules having operators like \textsf{div} and \textsf{rem}, which solvers are slow at handling.
The remaining 3 rules (not implemented) cannot be proven correct due to missing rules on reduction elements.

\begin{figure}[t]
    \centering
    %\captionsetup{justification=centering}
    %\captionsetup[subfigure]{justification=centering}
    \begin{subfigure}{0.38\textwidth}
        \centering
        \includegraphics[width=\textwidth]{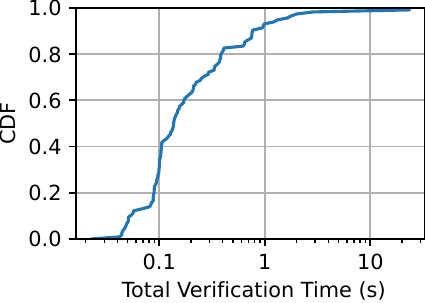}
        \caption{}
        \figlabel{verification-times}
    \end{subfigure}
    %\hfill
    \hspace{5em}
    \begin{subfigure}{0.38\textwidth}
        \centering
        \includegraphics[width=\textwidth]{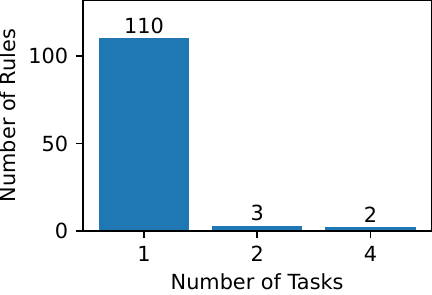}
        \caption{}
        \figlabel{num-tasks}
    \end{subfigure}
    %\vspace{-1em}
    \caption{(a) The cumulative distribution of total verification time and (b) the number of tasks (bounded-verification proof obligations) discharged.}
    %\figlabel{verification-stats}
    \vspace{-1em}
\end{figure}

\subsection{Generalizing Rewrite Rules} \seclabel{case-studies}

Using \project{}, we found that some \xla{} rewrite rules are overly constrained.
Compiler engineers intentionally impose these constraints to avoid reasoning about cases where spurious bugs might be introduced.
We used \project{} to generalize the following rule by relaxing its precondition.

\begin{wrapfigure}{R}{0.41\textwidth}
%\vspace{-1.0em}
    \small
    $
    \begin{array}{c}
    \begin{array}{l}
        \boldsymbol{\mathsf{FoldConvInputPad(XLA){:}}}
        \\
        \ \ \ \ \llet S_{ol} = S_l + S_{lp} \lin \\ 
        \ \ \ \ \llet S_{oh} = S_h + S_{hp}  \lin \\ 
        \ \ \ \ \ \ \ \ \conv(\pad(t, 0, S_{lp}, S_{hp}, S_{ip}), t',
        \\
        \phantom{\ \ \ \ \ \ \ \ \conv(}B, F, O, 
        %\\
        %\phantom{\ \ \ \ \ \ \ \ \mathsf{convolution}(}
        S_l, S_h, S_i, S'_i)
        \\
        \phantom{\mathsf{concatenate}(}
        \Longrightarrow_{\fbox{$S_{ip} = 0 \lland S_i = 1$}}
        %_{\substack{m_{ol} = m_l + (m_i + \overline{1}) \times m_{lp}\\ 
                                %   \ \ m_{oh} = m_h + (m_i + \overline{1}) \times m_{hp}\\ 
                                %   \ \ \ \ \ \ m_{oi} = (m_i + m_{pi}) + (m_i \times m_{pi}) }} 
                                %   \\[1cm]
        \\
        \ \ \ \ \ \ \ \ \conv(t, t', B, F, O, 
        \\
        \phantom{\ \ \ \ \ \ \ \ \conv(}
        S_{ol}, S_{oh}, S_i, S'_{i})
    \end{array}
    \\
    \\
    \begin{array}{l}
        \boldsymbol{\mathsf{FoldConvInputPad(Generalized){:}}}
        \\
        \ \ \ \ \fbox{$\llet S_{ol} = S_l + S_i \times S_{lp} \lin$} \\ 
        \ \ \ \ \fbox{$\llet S_{oh} = S_h + S_i \times S_{hp}  \lin$} \\ 
        \ \ \ \ \fbox{$\llet S_{oi} = S_i + S_i \times S_{ip}  \lin$} \\ 
        \ \ \ \ \ \ \ \ \conv(\pad(t, 0, S_{lp}, S_{hp}, S_{ip}), t',
        \\
        \phantom{\ \ \ \ \ \ \ \ \conv(}B, F, O, 
        %\\
        %\phantom{\ \ \ \ \ \ \ \ \mathsf{convolution}(}
        S_l, S_h, S_i, S'_i)
        \\
        \phantom{\mathsf{concatenate}(}
        \Longrightarrow
        %_{\substack{m_{ol} = m_l + (m_i + \overline{1}) \times m_{lp}\\ 
                                %   \ \ m_{oh} = m_h + (m_i + \overline{1}) \times m_{hp}\\ 
                                %   \ \ \ \ \ \ m_{oi} = (m_i + m_{pi}) + (m_i \times m_{pi}) }} 
                                %   \\[1cm]
        \\
        \ \ \ \ \ \ \ \ \conv(t, t', B, F, O, 
        \\
        \phantom{\ \ \ \ \ \ \ \ \conv(}
        S_{ol}, S_{oh}, S_{oi}, S'_{i})
    \end{array}
    \end{array}
    $
    \caption{Fold input $\pad$ into $\conv$.}
    \figlabel{fold-conv}
\vspace{-2em}
\end{wrapfigure}

\figref{fold-conv} (top) presents the \textsc{FoldConvInputPad} rule as it exists in \xla{} using the \dsl{} notation.
The goal of the rule is to fold the $\pad$ operator into the operand arguments of $\conv$ itself (\xlahlo{} convolutions support padding as operands).
This rule
does not support internal padding in the input tensor and gives up if this constraint is violated.
This is largely because it is non-trivial to think about how the internal padding gets folded into the dilation attribute.
\project{} was able to prove a more general version of this rule as shown in
\figref{fold-conv} (bottom).
The differences are put in boxes.

The key to generalizing the rule is to calculate the padding arguments that get fed into the $\conv$ operator.
This is a function of the $\pad$ operator's interior, high, and low padding, as well as padding that may already exist in the $\conv$ operator.
\figref{fold-conv} (bottom) shows how to calculate the maps $S_{ol}, S_{oh}, S_{oi}$ for the rule to be general.
These maps are more complicated than the non-general version, but this allows the compiler writer to get rid of the precondition of the rule shown in \figref{fold-conv} (bottom).
It is not immediately clear why this formulation might be correct. 
Therefore, we encode it in our Grisette implementation and successfully prove that the generalized rule with these calculations is valid.

\section{Related Works} \seclabel{related-works}

\project{} is inspired by prior works on representing and verifying compiler transformations.

\vspace{-0.2em}
\paragraph{Tensor Language Formalisms.} Glenside~\cite{glenside} formalizes the syntax and provides a composable abstraction to represent tensor graph rewrites in a purely functional form. 
TASO~\cite{taso} uses s-expression based representations to functionally model tensor operators.
It does not provide semantics for tensor operators and mostly relies on axioms built around the operators to perform verification related tasks.
Such an axiomatic approach would not scale with addition of new tensor operators, as it requires axioms describing operator properties and how the operators interact with each other.
\project{} on the other hand only requires the users to specify operator semantics for every new operator once.
PET~\cite{pet} verifies rewrites rule via a statistical approach. PET symbolically infers the bounding boxes of the output tensor, where each box contains elements represented by the same linear expression of its input elements.
Leveraging the linear property, PET statistically verifies the equivalence of the corresponding boxes by checking $m + 1$ specific positions in the box,
where $m$ is the number of axes of the output tensor.

There have been many works providing semantics for hardware instructions~\cite{x86-sem1, x86-sem2} or general purpose compiler IRs such as LLVM IR~\cite{llvm-verify3}.
ATL~\cite{atl-popl} is among the first works to provide denotational semantics to model a tensor language. It is closely modeled after the widely adopted Halide language~\cite{halide}.
In contrast, \project{} models its core language around the production \xla{} compiler's High Level Operators.
To the best of our knowledge, it is the first formalism supporting \xlahlo{}'s operators in their full generality, modeling all the parameterizations of operators.
Similar to ATL, \project{} provides denotational semantics of tensor operators with arbitrary rank and size, which is key to the proof that reduces unbounded verification into a bounded setting.

\paragraph{Verification of Rewrites with Proof Assistants.} ATL~\cite{atl-popl} is among the first works to successfully prove correctness of tensor graph rewrites with input tensors of arbitrary shape using the Coq proof assistant.
Comparatively, \project{} does automatic verification given the rewrite specification and accepts preconditions which are prevalent in practical rewrite rules developed by compiler engineers.
We note that ATL's Coq based approach supports layout-sensitive rewrites such as those that involve reshapes, which \project{} does not cover currently.
We provide a methodology to support those operators in \secref{limitation}. 
There are examples from other domains on mechanized proofs on rewrite systems, covering relational algebra~\cite{relational-algebra-coq, fiat} and compiler construction tools~\cite{comp-construction}.

\vspace{-0.5em}
\paragraph{Automated Verification of Rewrites.} We take inspiration from many successful works focusing on automatically verifying rewrites for different program representations, mainly with the aid of SMT solvers.
Alive~\cite{alive} focuses on verifying rewrite rules in LLVM's Instruction Combiner pass, which is LLVM's peephole optimization.
They mainly focus on scalar LLVM IR instructions.
Many works on superoptimization use automated verification of rewrites as part of their synthesis process. 
For example, the STOKE project~\cite{stoke} and others~\cite{bansal-superopt} verify rewrites expressed in x86 instructions, Souper~\cite{souper} verifies rewrites expressed in LLVM IR instructions,  Minotaur~\cite{minotaur} extends this to vector LLVM IR instructions and ~\cite{halide-rewriter} proves rewrite rules in Halide IR. TASO~\cite{taso}, PET~\cite{pet}, and TENSAT~\cite{tensat} are examples of systems that automatically synthesize tensor graph rewrites.
\project{} is influenced by the success of these systems and for the first time proposes an automated process for verifying tensor graph rewrites on input tensors of arbitrary ranke and size.
Further, \project{} is the first system to incorporate preconditions in its verification process.

\vspace{-0.5em}
\paragraph{Compiler Verification.} There is a lot of work in building verified general-purpose compilers. CompCert~\cite{CompCert-ERTS-2018} is a formally verified C compiler with many verification efforts and extensions~\cite{verify-c, verify-c2, verify-c3, verify-c4}. CakeML~\cite{cakeml} is a formally verified ML compiler.
There are works on building associated verified transformations~\cite{llvm-verify, llvm-verify2, llvm-verify3}. Comparatively, less works have explored verifying compiler transformations in tensor compilers.
ATL~\cite{atl-pldi,atl-popl} is one of the first successes on this front that builds upon a proof-assistant-aided verification process. 
To the best of our knowledge, \project{} is one of the first efforts at automatically verifying tensor graph rewrites closely resembling the industrial strength \xla{} tensor compiler.

%\vspace{-0.2em}
\section{Limitations and Future Work}\seclabel{limitation}

Currently, \project{} does not support layout-sensitive rules with operators like $\reshape$ or $\bitcast$, which change operand layouts or do not respect element boundaries.
The $\reshape$ operator is particularly challenging to verify because it can collapse or flatten an arbitrary number of axes, complicating the representation and verification of rank-polymorphic rules.
This complexity arises because $\reshape$ changes the interpretation of the input tensor by linearizing and de-linearizing its accesses, depending on the rank of input.
However, some reshape rules in \xla{} do not use the operator's full generality.
This suggests a pragmatic approach to extend \project{} to support these simpler cases, addressing much of its practical usage in \xla{}.
The exploration of reshape's full generality is left as potential future work.
    
As shown in \secref{verifyreduction}, \project{} can currently verify a subset of reduction rules.
Users need to provide hints to establish relations between reduction indices, satisfying assumptions like no duplicate values being reduced or 1-1 relations.
However, there are limitations, such as handling cases with extra zeroes being reduced, difficulty in proving bijectivity due solver limitations, and instances where no bijection relation exists.
In these cases, users need to complete the proof of set equivalence based on verifier output, which is generally easier than proving full correctness from scratch.
Future work may explore better proof strategies for reduction rules, possibly using a $k$-induction approach to establish bounds and finitize sizes for reduction axes.

\section{Conclusion} \seclabel{conclusion}

In this paper, we presented \project{}, the first automatic verification system that allows users to succinctly express and verify tensor graph rewrite rules in their full generality.
To do so, we designed \dsl{}, which allows specification of rank- and size-polymorphic rewrite rules using a novel axis definition, called \agg{}-axes.
\dsl{} consists of highly parameterized tensor operators, closely resembling those in \xlahlo{}.
We also provided denotational semantics for \dsl{} and used them to convert the unbounded-verification proof obligation to a finite set of bounded-verification proof obligations.
To the best of our knowledge, this is the first time a sizable subset of tensor operators from a production-quality tensor IR (\xlahlo{}) was formalised.
We demonstrated that \project{} can verify the majority of complex rewrite rules from the production \xla{} compiler's algebraic simplifier in the unbounded setting, 
vastly surpassing the closest automatic, bounded-verification technique.

\section{Data-Availability Statement}

An artifact \cite{tr-artifact} associated with this paper was evaluated and is freely available.

\begin{acks}
We thank the anonymous reviewers for their constructive feedback.
We would also like to thank Wanyu Zhao for her feedback on early drafts of this paper.
This work was supported in part by ACE, one of the seven centers in JUMP 2.0 and the CONIX Research Center, one of the six centers in JUMP, which are Semiconductor Research Corporation (SRC) programs sponsored by DARPA; by NSF under grants CCF-2338739, CCF-2122950, ITE-2132318, and CCF-2437238; by DARPA under grants FA8750-16-2-0032 and D22AP00146-00 as well as gifts from Adobe, Facebook, and Intel.
\end{acks}

% Already set in predefs.tex
\bibliographystyle{ACM-Reference-Format}
\bibliography{references}

\clearpage
\appendix

\section{Extended Operator Semantics} \seclabel{ext-sem}

We provide the denotational semantics for the following additional operators:

\begin{itemize}
    \item \semref{Pad}: The $\pad$ operator performs low ($S_l$), interior ($S_i$), and high ($S_h$) padding on the input tensor. 
    \begin{equation*}
        \inference{
            \mathsf{let} ~ S = \shape(e) & S_i \geq 0 \\
            \mathsf{let} ~ I = S + (S - 1) \times S_i & \mathsf{let} ~ L = I + S_l & \mathsf{let} ~ S_o = L + S_h \geq 0 \\
            \displaystyle \mathsf{let} ~ A' = \lambda A. ~ \frac{A - S_l}{S_i + 1} & \mathsf{let} ~ not\hbox{-}int\hbox{-}pad = \lambda A. ~ (A - S_l) ~ \% ~ (S_i + 1) = 0 \\
            \displaystyle \mathsf{let} ~ not\hbox{-}pad\hbox{-}area = \lambda A. ~ A \geq S_l \lland A < L \lland not\hbox{-}int\hbox{-}pad(A)
        }
        {
            \begin{gathered}
                \sem{\pad(e, v, S_l, S_h, S_i)} = \\
                \{A \mapsto
                \mathsf{if} ~ not\mbox{-}pad\mbox{-}area(A) ~ \mathsf{then} ~ \sem{e}[A'(A)] ~ \mathsf{else} ~ v ~|~ A \in \mathsf{Access}(S_o)\}
            \end{gathered}
        }
        [\textsc{Pad}]
        \semlabel{Pad}
    \end{equation*}
    
    \item \semref{Dot}: The $\tdot$ operator performs sum-of-products over the specified set of contracting \agg{}-axes ($X_c$).
    $X_b$ denotes the set of batch \agg{}-axes that stay independent across the computation.
    For this operation to be valid, $X_c$ and $X_b$ must be disjoint, and
    $X_c \cup X_b$ is the set of exactly those \agg{}-axes which are common in the input tensors $e_1$ and $e_2$.
    For verification purposes, we describe the semantics of the $\tdot$ operator by first expanding the two tensors to the same shape, followed by an element-wise multiplication, and then reducing on the contracting \agg{}-axes.
    The remaining set of \agg{}-axes in $e_1$, i.e., $S_1\setminus S_1\vert_{X_c\cup X_b}$, and in $e_2$, i.e., $S_2\setminus S_2\vert_{X_c\cup X_b}$, also called the spatial-axes, need to be disjoint.
    This condition is implicit in the $\expand$ operator below.

    \begin{equation*}
        \inference{
          \mathsf{let} ~ S_1 = \shape(e_1) &
          \mathsf{let} ~ S_2 = \shape(e_2) \\
          X_c\cup X_b = \axes(e_1) \cap \axes(e_2) &
          X_c\cap X_b=\varnothing \\
          \mathsf{let} ~ t =
          \sem{\reduce(+,\binary(\times, \expand(e_1,S_2\setminus S_2\vert_{X_c\cup X_b}), \expand(e_2,S_1\setminus S_1\vert_{X_c\cup X_b})), X_c)}
        }
        {
            \sem{\tdot(e_1, e_2, X_c, X_b)} = t
        }
        [\textsc{Dot}]
        \semlabel{Dot}
    \end{equation*}

    \item \semref{Conv}: We first model a basic convolution operator without padding and dilation.
    This can be expressed on top of $\slice$ (for extracting the window) and $\tdot$ (dot product of the window and the kernel) operators.
    Then we express the general convolution operator by using the base operator with the $\pad$ operator.
    In the inference rules, the subscript $i$ relates to the input and the subscript $w$ relates to the weights.
    $X_b$, $X_o$, $X_f$, and $X_{sp}$ are the sets of batch, output-feature, input-feature, and spatial \agg{}-axes, respectively.
    
    \[
    \inference{
        \mathsf{let} ~ S_i = \shape(e_i) &
        \mathsf{let} ~ S_w = \shape(e_2) &
        \mathsf{let} ~ X_i = \axes(e_i) \\
        \mathsf{let} ~ X_w = \axes(e_w) &
        X_b = X_i\setminus X_w &
        X_o = X_w\setminus X_i \\
        F\subseteq X_i\cap X_w &
        \mathsf{let} ~ X_{sp} = X_i \setminus (X_b\cup X_f) &
        \mathsf{dom}(I_p) = X_{sp} \\
        \displaystyle\mathsf{let} ~ S' = S_i \vert_{X_b} \cup S_w \vert_{X_o} \cup \left\lfloor\frac{S_i \vert_{X_{sp}} - S_w \vert_{X_{sp}}}{I_p} + 1\right\rfloor\\
        \mathsf{let}~sub = \lambda A_{sp}. \sem{\mathsf{dot}(\slice(e_i, A_{sp} \times I_p, A_{sp} \times I_p + S_w\vert_{X_{sp}}, 1), e_w, X_f\cup X_{sp}, \varnothing)}
    }
    {
        \sem{\convbase(e_i, e_w, X_b, X_f, X_o, I_p)} =
        \{A\mapsto sub(A\vert_{X_{sp}})[A\vert_{X_b\cup X_o}]\mid A\in\mathsf{Access}(S')\}
    }
    [\textsc{ConvBase}]
    \semlabel{ConvBase}
    \]
    %\jai{Explain the difference between dilation and padding}
    The $\conv$ operator uses the $\convbase$ operator by padding its operands with the specified arguments.
    Note that the inputs $S_i$ and $S_i'$ to the $\conv$ operator denote \emph{dilations}, which are the actual padding amounts incremented by 1.
    Therefore, we appropriately subtract 1 from these attributes before passing them to the $\pad$ operator.
    \[
    \inference{
        \sem{\convbase(\pad(e_1, 0, S_{l}, S_h, S_{i} - 1), \pad(e_2, 0, 0, 0, S'_{i} - 1), X_b, X_f, X_o, I_p)} = t
    }
    {
        \sem{\conv(e_i, e_w, X_b, X_f, X_o, S_{l}, S_{h}, S_i, S'_{i}, I_p)} = t
    }
    [\textsc{Conv}]
    \semlabel{Conv}
    \]
    
    \item \semref{Reverse}: The $\reverse$ operator reverses the order of elements in the input tensor along the specified set of \agg{}-axes $X_r$.
    \[
        \inference{
            \mathsf{let} ~ S = \shape(e) & X_r \subseteq \axes(e) \\
            \mathsf{let} ~ A' =  \lambda A. ~ (A \mysetminus A|_{X_r}) \cup (S|_{X_r} - A|_{X_r} - 1)
        }
        {
            \begin{gathered}
                \sem{\reverse(e, X_r)} =
                \{A \mapsto \sem{e}[A'(A)] ~|~ A \in \mathsf{Access}(S)\}
            \end{gathered}
        }
        [\textsc{Reverse}]
        \semlabel{Reverse}
    \]
    
    \item \semref{Select}: The $\select$ operator constructs an output tensor from elements of two input tensors $e_1$ and $e_2$, based on the values of a predicate tensor $e_b$.
    This operator is only valid if $e_b$ contains elements of type $\mathsf{Bool}$ and $e_1,e_2$ contain elements of the same type.
    \[
        \inference{
            \shape(e_b) = \shape(e_1) = \shape(e_2) \\
            \dtype(e_b) = \mathsf{Bool} & \dtype(e_1) = \dtype(e_2)
        }
        {
            \begin{gathered}
                \sem{\select(e_b, e_1, e_2)} = \\
                \{A \mapsto \mathsf{if} ~ \sem{e_b}[A] ~ \mathsf{then} ~ \sem{e_1}[A] ~ \mathsf{else} ~ \sem{e_2}[A] ~|~ A \in \mathsf{Access}(e)\}
            \end{gathered}
        }
        [\textsc{Select}]
        \semlabel{Select}
    \]
    
    \item \semref{Clamp}: The $\clamp$ operator clamps the input tensor ($e$) to within the range specified by the minimum ($e_{min}$) and maximum ($e_{max}$) tensors.
    \[
        \inference{
            \shape(e_{min}) = \shape(e) = \shape(e_{max}) \\
            \sem{\binary(\mathsf{min}, \binary(\mathsf{max}, e, e_{min}), e_{max})} = t
        }
        {
            \begin{gathered}
                \sem{\clamp(e_{min}, e, e_{max})} = t
            \end{gathered}
        }
        [\textsc{Clamp}]
        \semlabel{Clamp}
    \]
    
    \item \semref{ReduceSI}: The $\reduce$ operator takes a tensor and a set of \agg{}-axes $X$ as inputs and returns a tensor which contains uninterpreted reduction elements.
    The resulting tensor has the shape $S\setminus S\vert_X$, essentially removing all the \agg{}-axes in $X$.
    This results in the following semantics:
    \begin{equation*}
        \inference{
            \mathsf{let}~S=\shape(e) &
            \mathsf{let}~\{x_0\cdots x_k\}=X\subseteq \axes(e)\\
            \mathsf{let}~acc=\lambda A. \mathsf{Red}^\oplus_{I_0,\cdots,I_k}\sem{e}[\{x_0\mapsto I_0, \cdots,x_k\mapsto I_k\}\cup A]
        }
        {
                \sem{\reduce(\oplus, e, X)} =
                \{A\mapsto acc(A)\mid A\in\access(S \setminus S\vert_X)\}
        }
        [\textsc{ReduceOrig}]
        \semlabel{ReduceOrig}
    \end{equation*}
    The indices $I_0, \cdots, I_k$ in the function $acc$ are \emph{symbolic reduction} indices, hence making the reduction element uninterpreted.
    The term $\sem{e}[\{x_0\mapsto I_0, \cdots,x_k\mapsto I_k\}\cup A]$ represents a (multi-)set containing elements obtained by instantiating the indices $I_0, \cdots, I_k$ with concrete-maps.
    As the sizes of reduced axes are unbounded, we cannot compute this set of elements and hence cannot expand the reduction to sum all the elements being reduced.
    Thus, we leave the sum uninterpreted and provide special treatment for such elements during verification.

    Verifying rules with reductions would involve proving equivalence of reduction elements.
    Proving the equivalence of two reduction elements $\mathsf{Red}_Xf(X)$ and $\mathsf{Red}_Yg(Y)$ can \emph{often} be done by showing that the $\mathsf{LHS}$ and $\mathsf{RHS}$ are sums of the same values, i.e., $f(X)$ and $g(Y)$ represent the same (multi-)set.
    One way to prove set equivalence is to establish a bijection between $X$ and $Y$ and show each pair of values in $f(X)$ and $g(Y)$ are equal, regardless of tensor instantiation and operator attributes.
    
    In \project{}, the user provides a relation between $X$ and $Y$ as a hint.
    As we see in \semref{ReduceOrig}, $X$ and $Y$ contain symbolic indices.
    The current semantics suggest that we have to generate fresh symbolic variables for each invocation of $\reduce$.
    However, this precludes the user from defining hints as they cannot refer to these indices.
    %To allow to user to provide a relation between $X$ and $Y$, we need a way to refer to these indices.
    To solve this problem, we provide the symbolic reduction indices as input to the operator semantics.
    One can think of these extra arguments as symbolic \emph{names}, based on which we can define relations between them.
    The modified semantics are shown below.
    \begin{equation*}
        \inference{
            \mathsf{let}~S=\shape(e) &
             \mathsf{let}~\{x_0 \mapsto SI_0, \cdots, x_k \mapsto SI_k\}=SI \\
            \mathsf{let}~X=\mathsf{dom}(I) \subseteq \axes(e)\\
            \mathsf{let}~acc=\lambda A. \mathsf{Red}^\oplus_{SI_0,\cdots,SI_k}\sem{e}[SI \cup A]
        }
        {
                \sem{\reduce(\oplus, e, SI)} =
                \{A\mapsto acc(A)\mid A\in\access(S \setminus S\vert_X)\}
        }
        [\textsc{ReduceSI}]
        \semlabel{ReduceSI}
    \end{equation*}
    The $\reduce$ operator now takes an \agg{}-map $I$ as input, instead of the set of \agg{}-axes to reduce.
    $I$ contains the symbolic reduction indices $SI_0, \cdots, SI_k$, which are used the construct the uninterpreted reduction element.
    For any rewrite rule with reductions, the user can now use these supplied reduction indices from \textsf{LHS} and \textsf{RHS} to define hints, also called \emph{si-relations} (symbolic-index relations).
\end{itemize}
\section{Validity of Rewrite Rules} \seclabel{rewrite-validity}

Given the semantics of \dsl{}, we now define what it means for a rewrite rule to be valid.
\figref{extended-syntax} lists all the tensor operators supported in \dsl{}.
As a part of our verification methodology, we symbolically execute all operator semantics and generate symbolic representations of the \textsf{LHS} and \textsf{RHS} expressions.
The validity condition of a rewrite rule is simply the assertion of equality between the \textsf{LHS} and \textsf{RHS} expressions, under the specified precondition.
We now formalize this notion.

\begin{figure}
    \centering
    \begin{tabular}{llll}
            $\tau$& $\coloneq$ & $\mathsf{Int}\mid\mathsf{Bool}\mid\mathsf{Real}$&Type\\
            $a$&$\in$&$\mathcal{A}$&Named-axes\\
            $x$&$\in$&$\mathcal{X}=\mathcal{P}(\mathcal{A})$&Aggregated-axes\\
            % $d$&$\in$&$\mathcal{D}$&Aggregated dimensions\\
            $f$&$\in$&$\mathsf{list}[\mathsf{Int}]\to \mathsf{Int}$&Map function\\
            $m$&$\coloneq$&$\mathcal{M}\mid \mathsf{fmap}(f, m+)$&Maps\\
            $X$&$\in$&$\mathcal{P}(\mathcal{X})$&Set of \agg{}-axes\\
            $S,I$&$\in$&$m^\mathcal{X}$&Shapes and indices\\
            $R$&$\in$&$\mathcal{X}^\mathcal{X}$&Relabel maps\\
            % $v$&$\coloneq$&$i:\mathsf{Int}$ & Scalar literal \\
            % &$\mid$&$b:\mathsf{Bool}$ \\
            % &$\mid$&$r: \mathsf{Real}$\\
            % $v$&$\coloneq$&$i:\mathsf{Int}$&Scalar literal\\
            % &$\mid$&$b:\mathsf{Bool}\mid r: \mathsf{Real}$\\
            $v$&$\coloneq$&$i:\mathsf{Int}\mid b:\mathsf{Bool}\mid r: \mathsf{Real}$&Scalar literal\\
            $e$&$\coloneq$&$\mathcal{T}$ (Literal)&Tensor expression\\
              &$\mid$&$\mathcal{V}$ (Variable) \\
              &$\mid$&$\textsf{const}(v, S)$\\
              &$\mid$&$\tiota(S, x)$\\
              &$\mid$&$\expand(e, S)$\\
              &$\mid$&$\binary(\oplus, e_l, e_r)$\\
              &$\mid$&$\pad(e, v, S_l, S_h, S_i)$\\
              &$\mid$&$\slice(e, I_s, I_e, I_p)$\\
              &$\mid$&$\dyslice(e, I, S)$\\
              &$\mid$&$\dyupdateslice(e, e_u, I)$\\
              &$\mid$&$\reduce(\oplus, e, I)$\\
              &$\mid$&$\relabel(e, R)$\\
              &$\mid$&$\concat(e_l, e_h, x)$\\
               &$\mid$&$\convbase(e_i, e_w, X_b, X_f, X_o, I_p)$\\
               &$\mid$&$\conv(e_i, e_w, X_b, X_f, X_o, S_l, S_h, S_{i}, S_i', I_p)$\\
               &$\mid$&$\tdot(e_1, e_2, X_c, X_b)$\\
               &$\mid$&$\clamp(e_{min}, e, e_{max})$\\
               &$\mid$&$\select(e_b, e_1, e_2)$\\
               &$\mid$&$\reverse(e, X_r)$\\
            $g$&$\in$&$\mathsf{list}[\mathsf{Int}]\to \mathsf{Bool}$&Predicate function\\
            $P$&$\coloneq$&$\mathsf{fold}(g, m+)$&Precondition\\
            $Rule$&$\coloneq$&$e_{lhs}\Rightarrow_{P*} e_{rhs}$&Rewrite rule
        \end{tabular}
        \caption{Core rewrite rule representation with extended operators.}
        \figlabel{extended-syntax}
\end{figure}

% \begin{definition} \deflabel{syntactic-valid}
%     Statically Valid Rewrite Rule: A Rewrite Rule $\mathsf{LHS} \Rightarrow_C \mathsf{RHS}$ is statically valid if and only if
%     \begin{itemize}
%         \item $\mathsf{LHS}$ is a semantically valid expression
%         \item $\mathsf{RHS}$ is a semantically valid expression
%         \item $\shape(\mathsf{LHS}) = \shape(\mathsf{RHS})$, i.e, they have the same shape
%     \end{itemize}
% \end{definition}

\begin{definition}[\textsc{Valid Rewrite Rule}] \deflabel{semantic-valid}
    A rewrite rule $\mathsf{LHS} \Rightarrow_C \mathsf{RHS}$ is valid if and only if
    \begin{equation} \eqnlabel{rew-valid}
        \forall v \in \vars, ~ C \wedge \valexp(\mathsf{LHS}) \rightarrow \sem{\mathsf{LHS}} = \sem{\mathsf{RHS}} 
    \end{equation}
    where $\vars$ is the set of all variables appearing in the rewrite rule, such as tensor variables and operator attributes.
    The predicate $\valexp$ takes a tensor expression and returns a condition under which the expression is valid. 
    It is defined inductively on the structure of tensor expressions, collecting all assertions from the operator semantics, as shown in \figref{op-valid}.
\end{definition}

\begin{figure}
    \raggedright
    \textsc{Basis:}
    \begin{align*}
        &\valexp(\underline{t}) \coloneq \true & \underline{t} \in \mathcal{T} \\
        &\valexp(var) \coloneq \shape(var) \geq 0 & var \in \mathcal{V} \\
        &\valexp(\textsf{const}(v, S)) \coloneq S \geq 0 & \\
        &\valexp(\tiota(S, x)) \coloneq S \geq 0 \wedge x \in S & \\
    \end{align*}
    \textsc{Induction Step:}
    \begin{align*}
        &\valexp(\expand(e, S)) \coloneq \valexp(e) \wedge S \geq 0 \wedge \mathsf{dom}(S) \cap \axes(e) = \varnothing  \\
        &\valexp(\binary(\oplus, e_l, e_r)) \coloneq \valexp(e_l) \wedge \valexp(e_r) \wedge \shape(e_l) = \shape(e_r) \\
        % \pad(e, v, S_l, S_h, S_i)\\
        &\valexp(\slice(e, I_s, I_e, I_p)) \coloneq \valexp(e) \wedge 0 \leq I_s \leq I_e \leq \shape(e) \wedge I_p > 0 \\
        &\valexp(\dyslice(e, I, S)) \coloneq \valexp(e) \wedge I + S \leq \shape(e) \wedge S > 0 \wedge I \geq 0 \\
        % \dyupdateslice(e, e_u, I) \\
        &\phantom{\valexp(\dyslice(e, I, S))} \vdots
        % \reduce(\oplus, e, I)\\
        % \relabel(e, R)\\
        % \concat(e_l, e_h, x)\\
        % \convbase(e_i, e_w, X_b, X_f, X_o, I_p)\\
        % \conv(e_i, e_w, X_b, X_f, X_o, S_l, S_h, S_{i}, S_i', I_p)\\
        % \tdot(e_1, e_2, X_c, X_b)\\
        % \clamp(e_{min}, e, e_{max})\\
        % \select(e_b, e_1, e_2)\\
        % \reverse(e, X_r)
    \end{align*}
    \caption{The predicate $\valexp$ is defined inductively over the structure of tensor expressions. Other cases follow similarly, derived from the operator semantics.}
    \figlabel{op-valid}
\end{figure}

\defref{semantic-valid} states that a rewrite rule must be valid for \emph{all} valuations of all the variables appearing in a rewrite rule.
Note that we only consider valuations where the term prior to rewriting is valid, i.e., when $\mathsf{LHS}$ is valid.
The inner equality $\sem{\mathsf{LHS}} = \sem{\mathsf{RHS}}$ asserts that, once all the variables are instantiated, the resulting tensors on both sides are equal.
\secref{discussion-app} extends this definition to account for validity of the \textsf{RHS} expression.

Based on the denotational semantics of operators in \dsl{}, we can symbolically compute the shape of any tensor expression as a function of the shapes of input tensors and operator attributes.
We only consider rewrite rules where under the precondition, the \textsf{LHS} and \textsf{RHS} expressions have the same shape, since the rule would be invalid if the shapes do not match.
%As the rule is statically valid, both \textsf{LHS} and \textsf{RHS} have the same shape.
Under this assumption, semantic validity reduces to a pointwise equality of tensor values:
\begin{equation} \eqnlabel{rule-valid}
    \forall v \in \vars, C \wedge \valexp(\mathsf{LHS}) \rightarrow \forall A \in \mathsf{Access}(\mathsf{LHS}), ~ \sem{\mathsf{LHS}}[A] = \sem{\mathsf{RHS}}[A]
\end{equation}

\begin{definition}[\textsc{Invalid Rewrite Rule}] \deflabel{semantic-invalid}
    A rewrite rule $\mathsf{LHS} \Rightarrow_C \mathsf{RHS}$ is invalid if and only if
    \[
         \exists v \in \vars, C \wedge \valexp(\mathsf{LHS}) \wedge \exists A \in \mathsf{Access}(\mathsf{LHS}), ~ \sem{\mathsf{LHS}}[A] \neq \sem{\mathsf{RHS}}[A]
    \]
    where $\vars$ is the set of variables appearing in the rewrite rule.
    %This equation is the negation of \eqnref{rule-valid}.
\end{definition}

\defref{semantic-invalid} states that if a rewrite rule is invalid, then there exists some valuation of all variables and an access $A$, such that (1) the precondition is \true{}, (2) \textsf{LHS} is valid, and (3) \textsf{LHS} and \textsf{RHS} do not match at the access $A$.
Such a valuation of variables and the access $A$ together comprise a \emph{counterexample} for the rule.
% We refer to $A$ as the \emph{output counterexample access}.

\section{Preliminaries}

In \secref{rewrite-validity}, we defined what it means for a rewrite rule to be valid (\eqnref{rew-valid}).
A key challenge is that rewrite rules must apply to tensors with arbitrary ranks and arbitrary sizes along each axis.
Thus, we must verify rewrite rules in the \emph{unbounded setting}, where both the number of axes (rank) and the size along each axis are unbounded.
To handle arbitrary number of axes, we partition the \ind{}-axes of a tensor into finite groups or classes, called \agg{}-axes, where each \agg{}-axis can be instantiated to any rank.
Once the \agg{}-axes in a rewrite rule are instantiated, we obtain a concrete-rank instantiation of the rule, though the sizes along each axis remain unbounded.
Verifying the rewrite rule in the unbounded setting therefore reduces to verifying the rule for all valid instantiations of the \agg{}-axes.
In this section, we present our key observation (\secref{observation}), and in \secref{scalarf} we show how \project{} verifies any rewrite rule for a possibly infinite number of instantiations.

\subsection{Assumptions} \seclabel{assumptions}

To simplify the discussion, we make the following assumptions, which will be relaxed later.
Let $\mathsf{LHS} \Rightarrow_C \mathsf{RHS}$ be a rewrite rule such that:
\begin{itemize}
    \item The rewrite rule contains only one input tensor, denoted \textsf{X}
    \item \textsf{X} contains exactly one \agg{}-axis, say $x$, which has the \rclass{} $c$, i.e. $x : c$.
    %$x$ can have an arbitrary number of \ind{}-axes.
    \item The rewrite rule does not contain any $\reduce$, $\relabel$, $\expand$, $\tdot$, $\convbase$, and $\conv$ operators: these operators can add, remove or rename \agg{}-axes.
    For now, we keep the discussion restricted to a single \agg{}-axis.
\end{itemize}

\subsection{Observation} \seclabel{observation}

We can symbolically execute any rewrite rule $\mathsf{LHS} \Rightarrow_C \mathsf{RHS}$ under a general access $A$ using the operator semantics.
We observe that the equality $\sem{\mathsf{LHS}}[A] = \sem{\mathsf{RHS}}[A]$ (the inner equality in \eqnref{rule-valid}) can always be expressed in the following canonical form:
\definecolor{accesscolor}{HTML}{330066}
\begin{align} \eqnlabel{scalar-form}
    \textcolor{blue}{\scalarf}(&\textcolor{accesscolor}{\mathsf{X}[}\textcolor{red}{x \mapsto \fmap(e_1, \mathcal{M}_1)}\textcolor{accesscolor}{]}, ~ \cdots, ~\textcolor{accesscolor}{\mathsf{X}[}\textcolor{red}{x \mapsto \fmap(e_n, \mathcal{M}_n)} \textcolor{accesscolor}{]}, \nonumber \\
    &\textcolor{magenta}{\fold(g_1, \mathcal{N}_1)}, ~ \cdots, ~ \textcolor{magenta}{\fold(g_m, \mathcal{N}_m)})
\end{align}

We now illustrate, through examples, how symbolic evaluation of rewrite rules results in expressions of this form.

\subsection{Symbolic Evaluation Examples} \seclabel{symeval-examples}

\subsubsection{Pad-Low}

\paragraph{Setting}
Let $x : c$ be an \agg{}-axis with \rclass{} $c$.
It can be instantiated to an arbitrary number of \ind{}-axes.
Let \textsf{X} be a tensor with shape $\{x \mapsto s\}$, where $s$ contains sizes of \ind{}-axes in $x$.
Let the rewrite rule in consideration be,
\begin{align*}
    \padlow(\mathsf{X}, 0, S_l) \Rightarrow \padlow(\mathsf{X}, 1, S_l)
\end{align*}
where $S_l = \{x \mapsto s_l\}$ and $s_l$ contains padding sizes of \ind{}-axes in $x$.
The difference between \textsf{LHS} and \textsf{RHS} is that the former is padded with a value of 0, while the latter is padded with a value of 1.
Therefore, this rewrite rule is invalid.

\paragraph{Shape of the Output Tensors}
We compute the shape of the \textsf{LHS} (same as \textsf{RHS}) as follows:
\begin{align*}
    \shape(\padlow(\mathsf{X}, 0, S_l)) &= \shape(\mathsf{X}) + S_l \\
    &= \{x \mapsto s\} + \{x \mapsto s_l\} \\
    &= \{x \mapsto s + s_l\}
\end{align*}

Therefore, an arbitrary, valid access $A$ on the output tensors has the form $A = \{x \mapsto a\}$, where $a$ contains the index-values for the \ind{}-axes in $x$.
We symbolically evaluate the expression $\sem{\mathsf{LHS}}[A] = \sem{\mathsf{RHS}}[A]$ under the access $A$:
\begin{align*}
    & \sem{\mathsf{LHS}}[A] = \sem{\mathsf{RHS}}[A] \\
    \Leftrightarrow & ~ (\mathsf{if} ~ not\hbox{-}pad(A) ~ \mathsf{then} ~ \mathsf{X}[x \mapsto (a - s_l)] ~ \textsf{else} ~ 0) = \\
    & \quad\quad\quad\quad\quad(\mathsf{if} ~ not\hbox{-}pad(A) ~ \mathsf{then} ~ \mathsf{X}[x \mapsto (a - s_l)] ~ \textsf{else} ~ 1)
\end{align*}

We express the above in the canonical form given by \eqnref{scalar-form}, as follows:
\begin{align*}
    e(v, l) &\df v - l \\
    y &\df \mathsf{X}[x \mapsto \fmap(e, [a, s_l])]  = \mathsf{X}[x \mapsto a - s_l] \\
    g(v, l) &\df v \geq l \\
    b &\df not\hbox{-}pad(A) = A \geq S_l \\
    &= a \geq s_l = \fold(g, [a, s_l]) \\
    \scalarf(y, b) &\df (\mathsf{if} ~ b ~ \mathsf{then} ~ y ~ \mathsf{else} ~ 0) = (\mathsf{if} ~ b ~ \mathsf{then} ~ y ~ \mathsf{else} ~ 1) \\
    \sem{\mathsf{LHS}}[A] = \sem{\mathsf{RHS}}[A] &\df
    \scalarf(\mathsf{X}[x \mapsto \fmap(e, [a, s_l])], \fold(g, [a, s_l]))
\end{align*}

% Following examples demonstrate how the \agg{}-axis $x$ can be instantiated to concrete ranks:
% \begin{itemize}
%     \item $x$ has rank 2: Let $x$ contain two \ind{}-axes $i,j$. Then
%     \begin{align*}
%         a &= \{i \mapsto v_i, j \mapsto v_j\} \\
%         s_l &= \{i \mapsto l_i, j \mapsto l_j\} \\
%         &\scalarf(\mathsf{X}[x \mapsto \fmap(e, [a, s_l])], ~ \fold(g, [a, s_l])) \\
%         = \ & \scalarf(\mathsf{X}[x \mapsto \{i \mapsto v_i - l_i, j \mapsto v_j - l_j\}], ~ v_i \geq l_i \wedge v_j \geq l_j) 
%     \end{align*}
    
%     \item $x$ has rank 3: Let $x$ contain two \ind{}-axes $i,j, k$. Then
%     \begin{align*}
%         a &= \{i \mapsto v_i, j \mapsto v_j, k \mapsto v_k\} \\
%         s_l &= \{i \mapsto l_i, j \mapsto l_j, k \mapsto l_k\} \\
%         &\scalarf(\mathsf{X}[x \mapsto \fmap(e, [a, s_l])], ~ \fold(g, [a, s_l])) \\
%         = \ & \scalarf(\mathsf{X}[x \mapsto \{i \mapsto v_i - l_i, j \mapsto v_j - l_j, k \mapsto v_k - l_k\}], ~ v_i \geq l_i \wedge v_j \geq l_j \wedge v_k \geq l_k) 
%     \end{align*}
% \end{itemize}

\subsubsection{Slice}

\paragraph{Setting}
Let $x : c$ be an \agg{}-axis with \rclass{} $c$.
It can be instantiated to an arbitrary number of \ind{}-axes.
Let \textsf{X} be a tensor with shape $\{x \mapsto s\}$, where $s$ contains sizes of \ind{}-axes in $x$.
Let the rewrite rule in consideration be,
\begin{align*}
    \slice(\padlow(\mathsf{X}, 0, S_l), S_l, S + S_l, S_p) \Rightarrow \mathsf{X}
\end{align*}
where $S = \shape(\mathsf{X})$, $S_l = \{x \mapsto s_l\}$, and $S_p = \{x \mapsto 1\}$.

% $x \mapsto 3$ denotes that every \ind{}-axis in $x$ is low-padded with a size of 3.
The \textsf{LHS} first pads the input tensor by 0 with padding sizes $s_l$.
It then slices the resulting tensor by removing $s_l$ elements from the beginning along every \ind{}-axis, while skipping no elements till the end.
It is easy to see that the \textsf{LHS} is a no-op since it adds padding values and then removes them.
The \textsf{RHS} is simply the input tensor.
Therefore, this rewrite rule is valid.

\paragraph{Shape of the output tensors}
We compute the shape of the \textsf{LHS} (same as \textsf{RHS}) as follows:
\begin{align*}
    \shape(\slice(\padlow(\mathsf{X}, 0, S_l), S_l, S + S_l, S_p) &= S + S_l - S_l \\
    &= S \\
    &= \{x \mapsto s\}
\end{align*}

Therefore, an arbitrary, valid access $A$ on the output tensors has the form $A = \{x \mapsto a\}$, where $a$ contains the index-values for the \ind{}-axes in $x$.
We symbolically evaluate the expression $\sem{\mathsf{LHS}}[A] = \sem{\mathsf{RHS}}[A]$ under the access $A$:
\begin{align*}
    & \sem{\mathsf{LHS}}[A] = \sem{\mathsf{RHS}}[A] \\
    \Leftrightarrow & ~ \sem{\mathsf{pad}(\mathsf{X}, 0, S_l)}[S_l + A] = \mathsf{X}[A] \\
    \Leftrightarrow & ~ (\mathsf{if} ~ not\hbox{-}pad(S_l + A) ~ \mathsf{then} ~ \mathsf{X}[(A + S_l) - S_l] ~ \textsf{else} ~ 0) = \mathsf{X}[A] \\
    \Leftrightarrow & ~ (\mathsf{if} ~ not\hbox{-}pad(S_l + A) ~ \mathsf{then} ~ \mathsf{X}[A] ~ \textsf{else} ~ 0) = \mathsf{X}[A]
\end{align*}

We express the above in the canonical form given by \eqnref{scalar-form}, as follows:
\begin{align*}
    e(v) &\df v \\
    y &\df \mathsf{X}[x \mapsto \fmap(e, [a])]  = \mathsf{X}[x \mapsto a] \\
    g(v, l) &\df v + l \geq l \\
    b &\df not\hbox{-}pad(A + S_l) = A + S_l \geq S_l \\
    &= a + s_l \geq s_l = \fold(g, [a, s_l]) \\
    \scalarf(y, b) &\df (\mathsf{if} ~ b ~ \mathsf{then} ~ y ~ \mathsf{else} ~ 0) = y \\
    \sem{\mathsf{LHS}}[A] = \sem{\mathsf{RHS}}[A] &\df \scalarf(\mathsf{X}[x \mapsto \fmap(e, [a])], \fold(g, [a, s_l]))
\end{align*}

\subsubsection{Multiple \rclasses{}} \seclabel{ex-mul-rclass}

\paragraph{Setting}
Let $x_1 : c_1$ and $x_2 : c_2$ be \agg{}-axes with \rclasses{} $c_1$ and $c_2$ respectively.
They can be instantiated to an arbitrary number of \ind{}-axes.
Let \textsf{X} be a tensor with shape $\{x_1 \mapsto s_1\}$, where $s_1$ contains sizes of \ind{}-axes in $x_1$.

Let the rewrite rule rule in consideration be,
\begin{align*}
    \padlow(\expand(\mathsf{X}, \{x_2 \mapsto s_2\}), 0, S_{l_1}) \Rightarrow \expand(\padlow(\mathsf{X}, 0, S_{l_2}), \{x_2 \mapsto s_2\})
\end{align*}

where $S_{l_1} = \{x_1 \mapsto s_l, x_2 \mapsto 0\}$, $S_{l_2} = \{x_1 \mapsto s_l\}$, and $s_l$ contains padding sizes for \ind{}-axes in $x_1$.
The \textsf{LHS} first expands the tensor \textsf{X} with a new \agg{}-axis $x_2$ and then low-pads $x_1$ with $s_l$ and leaves $x_2$ unchanged.
The \textsf{RHS} first low-pads $x_1$ with $s_l$ and then expands with a new \agg{}-axis $x_2$.
This is a valid rewrite rule since \textsf{LHS} does not pad $x_2$ after the expansion, while \textsf{RHS} simply replicates the data across $x_2$.

\paragraph{Shape of the output tensors}
We compute the shape of the \textsf{LHS} (same as \textsf{RHS}) as follows:
\begin{align*}
    \shape(\padlow(\expand(\mathsf{X}, \{x_2 \mapsto s_2\}), 0, S_{l_1})) &= \shape(\expand(\mathsf{X}, \{x_2 \mapsto s_2\})) + S_{l_1} \\
    &= \shape(\mathsf{X}) \cup \{x_2 \mapsto s_2\} + S_{l_1} \\
    %&= \{x_1 \mapsto s_1\} \cup \{x_2 \mapsto s_2\} + S_{l_1} \\
    &= \{x_1 \mapsto s_1, x_2 \mapsto s_2\} + \{x_1 \mapsto s_l, x_2 \mapsto 0\} \\
    &= \{x_1 \mapsto s_1 + s_l, x_2 \mapsto s_2\}
\end{align*}

Therefore, an arbitrary, valid access $A$ on the output tensors has the form $A = \{x_1 \mapsto a_1, x_2 \mapsto a_!\}$, where $a_1$ and $a_2$ contain the index-values for the \ind{}-axes in $x_1$ and $x_2$ respectively.
We symbolically evaluate the expression $\sem{\mathsf{LHS}}[A] = \sem{\mathsf{RHS}}[A]$ under the access $A$:
\begin{align*}
    & \sem{\mathsf{LHS}}[A] = \sem{\mathsf{RHS}}[A] \\
    \Leftrightarrow & ~ (\mathsf{if} ~ not\hbox{-}pad(A) ~ \mathsf{then} ~ \sem{\expand(\mathsf{X}, \{x_2 \mapsto s_2\})}[A - S_{l_1}] ~ \mathsf{else} ~ 0) = \\
    & \quad\quad\quad\quad\quad\quad\quad\quad\sem{\padlow(\mathsf{X}, 0, S_{l_2})}[A|_{x_1}] \\
    \Leftrightarrow & ~ (\mathsf{if} ~ not\hbox{-}pad(A) ~ \mathsf{then} ~ \mathsf{X}[(A - S_{l_1})|_{x_1}] ~ \mathsf{else} ~ 0) = \\
    & \quad\quad\quad\quad\quad\quad\quad\quad (\mathsf{if} ~ not\hbox{-}pad(A|_{x_1}) ~ \mathsf{then} ~ \mathsf{X}[A|_{x_1} - S_{l_2}] ~ \mathsf{else} ~ 0) \\
    \Leftrightarrow & ~ (\mathsf{if} ~ not\hbox{-}pad(A) ~ \mathsf{then} ~ \mathsf{X}[x_1 \mapsto a_1 - s_l] ~ \mathsf{else} ~ 0) = \\
    & \quad\quad\quad\quad\quad\quad\quad\quad (\mathsf{if} ~ not\hbox{-}pad(A|_{x_1}) ~ \mathsf{then} ~ \mathsf{X}[x_1 \mapsto a_1 - s_l] ~ \mathsf{else} ~ 0)
\end{align*}

We express the above in the canonical form given by \eqnref{scalar-form}, as follows:
\begin{align*}
    e(v, l) &\df v - l \\
    y &\df \mathsf{X}[x_1 \mapsto \fmap(e, [a_1, s_l])]  = \mathsf{X}[x_1 \mapsto a_1 - s_l] \\
    g_1(v, l) &\df v \geq l \\
    g_2(v, l) &\df v \geq 0 \\
    b_1 &\df not\hbox{-}pad(A) = A \geq S_{l_1} \\
    &= a_1 \geq s_l \wedge a_2 \geq 0 = \fold(g_1, [a_1, s_l]) \wedge \fold(g_2, [a_2]) \\
    g_3 &\df \true \\
    b_2 &\df not\hbox{-}pad(A|_{x_1}) = a_1 \geq s_l \\
    &= a_1 \geq s_l \wedge \true = \fold(g_1, [a_1, s_l]) \wedge \fold(g_3, []) \\
    \scalarf(y, b_1, b_2) &\df (\mathsf{if} ~ b_1 ~ \mathsf{then} ~ y ~ \mathsf{else} ~ 0) = (\mathsf{if} ~ b_2 ~ \mathsf{then} ~ y ~ \mathsf{else} ~ 0) \\
    \sem{\mathsf{LHS}}[A] = \sem{\mathsf{RHS}}[A] &\df \scalarf(\fmap(e, [a_1, s_l]), \fold(g_1, [a_1, s_l]) \wedge \fold(g_2, [a_2]), \\
    &\qquad\qquad\qquad \fold(g_1, [a_1, s_l]) \wedge \fold(g_3, []))
\end{align*}

%\jai{Mention how $b_2$ was ``padded'' with extra conditions}
Here, the condition $b_2$ has been ``padded'' with a \emph{trivial} $\fold$ over the \agg{}-axis $x_2$, using a constant predicate function $g_3 = \true$.
This notational convention ensures that each condition contains exactly one $\fold$ per \agg{}-axis, some of which may be trivial.
Such trivial folds are introduced for uniformity and can be safely removed during postprocessing.

\subsubsection{Key Takeaway}
The observation that rewrite rules in \dsl{} can be symbolically evaluated and expressed in the form given by \eqnref{scalar-form}, extends to \emph{all} rewrite rules in the language.
This generalization can be formally proven by induction on the structure of tensor expressions.
We omit the proof, as the structure of the argument follows directly from the operator semantics.

\section{Scalar Function Explained} \seclabel{scalarf}

In \secref{symeval-examples}, we illustrated through examples how any rewrite rule $\mathsf{LHS} \Rightarrow_C \mathsf{RHS}$ in \dsl{} can be symbolically evaluated under a general access $A$ using the operator semantics.
The assumptions from \secref{assumptions} still apply here and will be incrementally relaxed in later sections.
The equality $\sem{\mathsf{LHS}}[A] = \sem{\mathsf{RHS}}[A]$ can be expressed in the following form:
\definecolor{accesscolor}{HTML}{330066}
\begin{align*}
    \textcolor{blue}{\scalarf}(&\textcolor{accesscolor}{\mathsf{X}[}\textcolor{red}{x \mapsto \fmap(e_1, \mathcal{M}_1)}\textcolor{accesscolor}{]}, ~ \cdots, ~\textcolor{accesscolor}{\mathsf{X}[}\textcolor{red}{x \mapsto \fmap(e_n, \mathcal{M}_n)} \textcolor{accesscolor}{]}, \\
    &\textcolor{magenta}{\fold(g_1, \mathcal{N}_1)}, ~ \cdots, ~ \textcolor{magenta}{\fold(g_m, \mathcal{N}_m)})
\end{align*}

We explain each component below:
\begin{itemize}
    \item \textcolor{accesscolor}{\textsf{X}[\_]}: this represents an access to the tensor \textsf{X}.
    
    \item $\textcolor{red}{x \mapsto \fmap(e_i, \mathcal{M}_i)}$: this represents the access expression for the $i^{\text{th}}$ access to \textsf{X}.
    Since \textsf{X} has only one \agg{}-axis $x$, the access expression contains the same \agg{}-axis, with the access map $\textcolor{red}{\fmap(e_i, \mathcal{M}_i)}$.
    
    The $\fmap$ construct takes a function and applies it pointwise to a list of maps with the same domain.
    It as defined a follows:
    \[
        \fmap(f, [m_1, m_2, \cdots]) = \{i \mapsto f(m_1(i),m_2(i),\cdots) \ \vert \ i \in \mathsf{dom}(m_1) \}
    \]
    For example, if $\mathcal{M} = [\{i \mapsto v_i, j \mapsto v_j\}, \{i \mapsto l_i, j \mapsto l_j\}]$ and $f = \lambda v,l. (v-l)$, then $\fmap(f, \mathcal{M}) = \{i \mapsto v_i - l_i, j \mapsto v_j - l_j\}$.
    %$\mathcal{M}_i$ is the list of maps appearing in this access.
    
    Here, $\mathcal{M}_i$ is the list of maps used in the $i^{\text{th}}$ access, which include access maps, attribute maps, etc. 
    Each function $e_i$ is independent of the rank of $x$, and only the maps in $\mathcal{M}_i$ vary with the rank of $x$.
    Thus, we are able to capture all the \emph{rank-independent} information in the function $e_i$.
    We refer to $e_i$ as an \emph{\indtrans{}} because it transforms output index-values to input index-values.
    % Every \ind{}-axis in $x$ is treated in the \emph{same way}, i.e., every \ind{}-axis index changes in the same way through $e_i$.
    For the $\fmap$ construct to be well-formed, the arity of $e_i$ must match the number of maps in $\mathcal{M}_i$.
    
    When \textsf{X} has multiple \agg{}-axes, each access must have separate access maps for each \agg{}-axis.
    For example, if $\axes(\mathsf{X}) = \{x_1, \cdots, x_p\}$, then the $i^{\text{th}}$ access to \textsf{X} has the following form:
    \[
        \mathsf{X}[x_1 \mapsto \fmap({}^1e_i, {}^1\mathcal{M}_i), \ \cdots \ , x_p \mapsto \fmap({}^pe_i, {}^p\mathcal{M}_i)]
    \]
    where ${}^je_i$ denotes the \indtrans{} for $x_j$ in the $i^{\text{th}}$ access, and ${}^j\mathcal{M}_i$ is the corresponding list of maps.
    Each ${}^je_i$ may vary across different \agg{}-axes.
    However, it is applied to all \ind{}-axes uniformly in a given \agg{}-axis.

    \item $\textcolor{magenta}{\fold(g_j, \mathcal{N}_j)}$: 
    these represent boolean values, referred to as \emph{\cond{}s}, that appear within \textsf{if-then-else} blocks.
    These conditions capture the dependency of output tensor values on the input tensor values and originate directly from the operator semantics.
    For instance, $not\hbox{-}pad$ in the $\pad$ semantics (\semref{Pad}) takes an access $A$ and determines if it lies in the padded area or not.
    
    The $\fold$ construct takes a predicate and applies it pointwise to a list of maps with the same domains.
    It returns \true{} if all predicate values are \true{}, and \false{} otherwise.
    It is defined as:
    \[
        \fold(g, [m_1,m_2,\cdots]) = \bigwedge_{i \in \mathsf{dom}(m_1)} g(m_1(i), m_2(i), \cdots)
    \]
    For example, if $\mathcal{N} = [\{i \mapsto v_i, j \mapsto v_j\}, \{i \mapsto l_i, j \mapsto l_j\}]$ and $g = \lambda v,l. (v\geq l)$, then $\fold(g, \mathcal{N}) = v_i \geq  l_i \wedge v_j \geq l_j$.
    %$\mathcal{M}_i$ is the list of maps appearing in this access.
    
    Here, $\mathcal{N}_j$ is the list of maps used in the $j^{\text{th}}$ condition, which include access maps, attribute maps, etc. 
    Each function $g_j$ is independent of the rank of the \agg{}-axis $x$, and only the maps in $\mathcal{N}_j$ vary with the rank of $x$.
    Thus, we are able to capture all the \emph{rank-independent} information in the function $g_j$.
    For the $\fold$ construct to be well-formed, the arity of $g_j$ must match the number of maps in $\mathcal{N}_j$.
    
    In case where the rewrite rule has multiple \agg{}-axes, we observe that all such conditions can be expressed as a conjunction of $\fold$s, one for each \agg{}-axis.
    More specifically, if $x_1 \cdots x_p$ are the \agg{}-axes appearing in a rule, then the $j^{\text{th}}$ condition has the form:
    \[
        \bigwedge_{i = 1\cdots p} \fold({}^ig_j, {}^i\mathcal{N}_j)
    \]
    where ${}^ig_j$ is the $j^{\text{th}}$ condition for the \agg{}-axis $x_i$, and ${}^i \mathcal{N}_j$ is the corresponding list of maps.
    Some important observations:
    \begin{itemize}
        \item Some of the ${}^ig_j's$ may be \emph{trivial}, i.e., they return \true{} for all inputs.
        In that case, the corresponding $\fold$ always evaluates to \true{} and does not contribute to the actual condition.
        \item Introducing trivial folds serves a notational purpose: every condition is uniformly represented as a conjunction of one $\fold$ per \agg{}-axis, regardless of whether that axis contributes meaningful conditions.
        This convention simplifies the representation and such trivial folds can be removed once they're no longer needed.

        \secref{ex-mul-rclass} illustrates an example where a condition only involving $x_1$ was extended with a trivial fold over $x_2$ to maintain uniformity.
        
        \item Each ${}^ig_j$ may vary across different \agg{}-axes.
        However, it is applied to all \ind{}-axes uniformly in a given \agg{}-axis.
    \end{itemize}
    
    \item $\textcolor{blue}{\scalarf}$:
    this is a scalar function which returns a boolean value, indicating whether the values of \textsf{LHS} and \textsf{RHS} tensors match at a given access $A$.
    We refer to it as $\scalarf$ since it captures the core, \emph{scalar computation} in the expression, typically consisting of arithmetic and conditionals over scalars.
    
    Based on the arguments to $\scalarf$, we say: $\scalarf$ comprises \emph{n accesses to \textsf{X}} and \emph{m \cond{}s}.
    Just like \indtrans{}s, $\scalarf$ is independent of the rank of $x$.
\end{itemize}

We now rewrite \eqnref{rule-valid} in terms of the $\scalarf$ function, accesses, and conditions, to get the validity condition of a rewrite rule:
\begin{align} \eqnlabel{rule-scalarf-expanded}
    \forall v \in \vars, C &\wedge \valexp(\mathsf{LHS}) \rightarrow \forall A \in \access(\mathsf{LHS}),  \nonumber \\
    \scalarf(&\mathsf{X}[x \mapsto \fmap(e_1, \mathcal{M}_1)], ~ \cdots, ~ \mathsf{X}[x \mapsto \fmap(e_n, \mathcal{M}_n) ], \nonumber\\
    &\fold(g_1, \mathcal{N}_1), ~ \cdots, ~ \fold(g_m, \mathcal{N}_m))
\end{align}

\subsection{Verification Methodology}

\begin{definition}[\textsc{\rcrank{} Map}]
    Let $R$ be a rewrite rule with \rclasses{} $c_1, \cdots, c_p$.
    An \emph{\rcrank{} map} is a mapping $\{c_1 \!\mapsto\!r_1, \cdots, c_p\!\mapsto\!r_p\}$, that assigns a concrete rank $r_i$ to each \rclass{} $c_i$, representing a specific rank instantiation of the rule $R$.
\end{definition}

\begin{definition}[\textsc{Validity for Concrete \rclass{} Ranks}]
    Let $R$ be a rewrite rule and $m$ be an \rcrank{} map for $R$.
    Then the predicate $\mathsf{valid}(R, m)$ holds if and only if:
    \[
        \mathsf{valid}(R, m) \Leftrightarrow R \text{ is valid for the number of \ind{}-axes as specified by } m 
    \]
    This corresponds to a concrete, bounded verification instance: the ranks of all \rclasses{} (and thus of all \agg{}-axes) are fixed according to $m$, while the sizes of the \ind{}-axes remain arbitrary.
\end{definition}

Let $R = \mathsf{LHS} \Rightarrow_C \mathsf{RHS}$ be the rewrite rule that we wish to verify in the unbounded setting.
For simplicity, the assumptions from \secref{assumptions} still apply.
We identify a \emph{sufficient} rank $k$ such that:
\begin{align} \eqnlabel{app-induction}
    \forall ~ i \geq k, ~ \mathsf{valid}(R, \{c \mapsto i\}) \rightarrow \mathsf{valid}(R, \{c \mapsto i\!+\!1\})
\end{align}

This means that for all $i \geq k$, if the rule is valid when the \rclass{} $c$ has rank $i$, then it is also valid when $c$ has rank $i\!+\!1$.
Given such a sufficient rank $k$, we can verify the rule in the unbounded setting using induction on tensor ranks:
\begin{itemize}
    \item \textsc{Basis}: Use bounded verification to prove that the rule is valid for all ranks up to $k$:
    \[
        \bigwedge_{i = 1 \cdots k} \mathsf{valid}(\{c \mapsto i\})
    \]
    \item \textsc{Induction Step}: Use induction on the rank of $c$ with \eqnref{app-induction} as induction hypothesis:
    \begin{equation*}
        \mathsf{valid}(\{c \mapsto k\}) \wedge \left[\forall i \geq k, ~ \mathsf{valid}(\{c \mapsto i\}) \rightarrow \mathsf{valid}(\{c \mapsto i\!+\!1\}) \right] ~
        \Rightarrow~\bigwedge_{i \geq k} \mathsf{valid}(\{c \mapsto i\})
    \end{equation*}
\end{itemize}
In other words, the rule is valid for any number of \ind{}-axes in the \rclass{} $c$.
What remains is to determine a sufficient rank $k$ for each \rclass{} in a given rule.

We begin by analyzing the validity condition of the rule when the \agg{}-axis $x$ is instantiated with a concrete rank $i$, i.e., when $\mathsf{valid}(\{c \mapsto i\})$ holds.
We first instantiate $x$ to $x^i$, a concrete \agg{}-axis containing $i$ distinct \ind{}-axes, namely $\{\underline{1}, \cdots, \underline{i}\}$.
We similarly instantiate the precondition $C$, input tensor \textsf{X}, list of maps $\mathcal{M}_j$ and $\mathcal{N}_j$, operator attributes, and tensor expressions.
Thus, we can expand the validity $\mathsf{valid}(\{c \mapsto i\})$ by rewriting \eqnref{rule-scalarf-expanded} as:
\begin{align} \eqnlabel{rule-scalarf-expanded-i}
    \forall v^i \in \vars^i, C^i &\wedge \valexp(\mathsf{LHS}^i) \rightarrow \forall A^i \in \access(\mathsf{LHS}^i),  \nonumber \\
    \scalarf(&\mathsf{X}^i[x^i \mapsto \fmap(e_1, \mathcal{M}_1^i)], ~ \cdots, ~ \mathsf{X}^i[x^i \mapsto \fmap(e_n, \mathcal{M}_n^i) ], \nonumber\\
    &\fold(g_1, \mathcal{N}_1^i), ~ \cdots, ~ \fold(g_m, \mathcal{N}_m^i))
\end{align}
As noted previously, the scalar function $\scalarf$, the index transformers $e_j$, and the condition predicates $g_j$ are independent of the rank of $x$.
The only components that vary with the rank $i$ are the variables $\vars^i$, access $A^i$, and list of maps $\mathsf{M}_j^i$ and $\mathcal{N}_j^i$.

To compute a sufficient rank $k$, we start from the inductive goal [$\mathsf{valid}(\{c \mapsto k\}) \rightarrow \mathsf{valid}(\{c \mapsto k\!+\!1\})$].
We instead work with its contrapositive,
\begin{equation*}
    \mathsf{valid}(\{c \mapsto k\}) \rightarrow \mathsf{valid}(\{c \mapsto k\!+\!1\})
    \Leftrightarrow \neg \mathsf{valid}(\{c \mapsto k\!+\!1\}) \rightarrow \neg \mathsf{valid}(\{c \mapsto k\})
\end{equation*}

Intuitively, $\neg\mathsf{valid}(\{c \mapsto i\})$ holds if there exists a counterexample to the rewrite rule at rank $i$.
Our goal is to find a sufficient rank $k$ such that any counterexample at rank $k\!+\!1$ (or higher) can be \emph{lowered} to a counterexample at rank $k$.

On expanding the validity condition using \eqnref{rule-scalarf-expanded-i}, the contrapositive becomes:
\begin{align*}
    \exists ~ v^{k+1} \in \vars^{k+1}, ~ &C^{k+1} \wedge \valexp(\mathsf{LHS}^{k+1}) \wedge \exists A^{k+1} \in \access(\mathsf{LHS}^{k+1}),  \nonumber \\
    \neg\scalarf(&\mathsf{X}^{k+1}[x^{k+1} \mapsto \fmap(e_1, \mathcal{M}_1^{k+1})], ~ \cdots, ~ \mathsf{X}^{k+1}[x^{k+1} \mapsto \fmap(e_n, \mathcal{M}_n^{k+1}) ], \nonumber\\
    &\fold(g_1, \mathcal{N}_1^{k+1}), ~ \cdots, ~ \fold(g_m, \mathcal{N}_m^{k+1})) \\
    &\quad\quad\quad\quad{\big\downarrow} \\
    \exists ~ v^{k} \in \vars^{k}, ~ &C^{k} \wedge \valexp(\mathsf{LHS}^{k}) \wedge \exists A^{k} \in \access(\mathsf{LHS}^{k}),  \nonumber \\
    \neg\scalarf(&\mathsf{X}^{k}[x^{k} \mapsto \fmap(e_1, \mathcal{M}_1^{k})], ~ \cdots, ~ \mathsf{X}^k[x^k \mapsto \fmap(e_n, \mathcal{M}_n^{k}) ], \nonumber\\
    &\fold(g_1, \mathcal{N}_1^{k}), ~ \cdots, ~ \fold(g_m, \mathcal{N}_m^{k})) \\
\end{align*}

This means that we are given a counterexample at rank $k\!+\!1$, which comprises:
\begin{itemize}
    \item A tensor $\mathsf{X}^{k+1}$ with $k\!+\!1$ \ind{}-axes in $x^{k+1}$, having the shape $\{x^{k+1} \mapsto m\}$, where $m = \{\underline{1} \mapsto n_1, \ \cdots, \ \underline{k\!+\!1} \mapsto n_{k+1}\}$,
    % \begin{equation*}
    %     m = \{\underline{1} \mapsto n_1, \ \cdots, \ \underline{k+1} \mapsto n_{k+1}\}
    % \end{equation*}
    \item A valuation of variables $\vars^{k+1}$ such that the precondition $C^{k+1}$ is satisfied and the expression $\mathsf{LHS}^{k+1}$ is valid, and
    \item An access $A^{k+1} \in \access(\mathsf{LHS}^{k+1})$, where the output tensors do not match.
\end{itemize}

To lower this counterexample from rank $k\!+\!1$ to rank $k$, we must construct:
\begin{itemize}
    \item A tensor $\mathsf{X}^{k}$ with $k$ \ind{}-axes in $x^{k}$,
    \item A valuation of variables $\vars^k$ such that the precondition $C^k$ is satisfied and the expression $\mathsf{LHS}^{k}$ is valid, and
    \item An access $A^k \in \access(\mathsf{LHS}^k)$, where the output tensors do not match.
\end{itemize}

Our goal is to derive lower bounds on $k$ while ensuring that any $(k\!+\!1)$-ranked counterexample can be lowered to a $k$-ranked counterexample.
This lower bound serves as a sufficient rank for the inductive argument.

\subsubsection{Counterexample Construction}

Our counterexample construction algorithm involves \emph{projecting} the $(k\!+\!1)$-ranked \agg{}-axis $x^{k+1}$ to a $k$-ranked \agg{}-axis $x^k$ by choosing $k$ \ind{}-axes from $\{\underline{1}, \cdots, \underline{k\!+\!1}\}$.
There are $k\!+\!1$ such projections in total, but not all of them necessarily yield a valid counterexample at rank $k$.
We express this construction using a set of equations and use them to derive \emph{constraints} on valid projections.

\paragraph{Projection Notation}
Let $\Gamma \subset x$ be any projection of the \agg{}-axis $x$.
We define projection operations on different constructs as follows:
\begin{itemize}
    \item Map Projection:
    Given a map $m$ defined on an \agg{}-axis $x$ (i.e., $\mathsf{dom}(m) = x$), its projection to $\Gamma$ is defined as:
    \[
        \cproject{m}{x}{\Gamma} \df \{\underline{j} \mapsto m(\underline{j})~\mid~ \underline{j} \in \Gamma\}
    \]
    If $m$ is defined on a different \agg{}-axis $x' \ne x$, then $\cproject{m}{x}{\Gamma} \df m$.

    \item Aggregated-Map Projection:
    Given an \agg{}-map $M$ that contains the \agg{}-axis $x$, its projection to $\Gamma$ is defined as:
    \[
        \cproject{M}{x}{\Gamma} \df (M \mysetminus \{x \mapsto M(x)\}) \cup \{\Gamma \mapsto \cproject{M(x)}{x}{\Gamma}\}
    \]
    If $M$ does not contain $x$, then $\cproject{M}{x}{\Gamma} \df M$.

    \item Tensor Projection:
    Given a tensor $\mathsf{X}$, its projection is a tensor with shape $\cproject{\shape(\mathsf{X})}{x}{\Gamma}$.
    The projected tensor $\cproject{\mathsf{X}}{x}{\Gamma}$ has unspecified values (unless determined by other constraints).

    \item Valuation Projection:
    Given a set of variable valuations $\vars$, its projection is defined pointwise as:
    \[
        \cproject{\vars}{x}{\Gamma} \df \{\cproject{v}{x}{\Gamma} \ \mid \ v \in \vars\}
    \]

    \item Map List projection:
    Given a list of maps $\mathcal{M} = [m_1, \cdots, m_r]$, its projection is defined elementwise as:
    \[
        \cproject{\mathcal{M}}{x}{\Gamma} \df [\cproject{m_1}{x}{\Gamma}, \cdots, \cproject{m_r}{x}{\Gamma}]
    \]
    \item Bijection Projection:
    Given a bijection $\mu : x \leftrightarrow y$, its projection is defined as:
    \[
        \cproject{\mu}{x}{\Gamma} \df \{\Gamma(i) \in y \ \mid \ i \in x\}
    \]
    The resulting function $\cproject{\mu}{x}{\Gamma} : \Gamma \leftrightarrow \Gamma(x)$ is also a bijection.

    \item Successive Projections:
    For any construct $\alpha$ which can be projected (map, aggregated map, tensor, valuation, bijection, etc.), successive projections are defined as:
     \[
        \cproject{\alpha}{x_1 \cdots x_p}{\Gamma_1 \cdots \Gamma_p} \df \cproject{\cproject{\alpha}{x_1}{\Gamma_1} \cdots}{x_p}{\Gamma_p}
    \]
\end{itemize}

\paragraph{Map-List Notation}
Let $\mathcal{M} = [m_1, m_2, \cdots]$ be a list of maps sharing a common domain, and let $i \in \mathsf{dom}(m_1)$ be a \ind{}-axis.
We write $\mathcal{M}(i)$ for the list obtained by applying each map to $i$:
\[
    \mathcal{M}(i) \df [m_1(i), m_2(i), \cdots]
\]

Let $\Gamma \subset \{\underline{1}, \cdots, \underline{k\!+\!1}\}$ be a projection of size $k$, representing a selection of $k$ \ind{}-axes from $x^{k+1}$.
Note that $\Gamma$ is same as $x^k$, and we use them interchangeably.
At this point, the \ind{}-axes in $\Gamma$ are unknown.
As part of constructing a $k$-ranked counterexample from a rank $k\!+\!1$ counterexample, we define the $k$-ranked components by projecting their $(k\!+\!1)$-ranked counterparts using $\Gamma$:
\begin{itemize}
    \item Tensors: $\mathsf{X}^k = \cproject{\mathsf{X}^{k+1}}{x}{\Gamma}$,
    \item Variable Valuations: $\vars^k = \cproject{\vars^{k+1}}{x}{\Gamma}$,
    \item Output Counterexample Access: $A^k = \cproject{A^{k+1}}{x}{\Gamma}$,
    \item Access Map Lists: for all $i \in \{1 \cdots n\}$, $\mathcal{M}_i^k = \cproject{\mathcal{M}_i^{k+1}}{x}{\Gamma}$, and
    \item Condition Map Lists: for all $j \in \{1 \cdots m\}$, $\mathcal{N}_j^k = \cproject{\mathcal{N}_j^{k+1}}{x}{\Gamma}$
\end{itemize}
In all of the projections listed above, we omit the superscript from the \agg{}-axis for brevity, i.e., $\cproject{\alpha}{x}{\Gamma}$ as shorthand for $\cproject{\alpha}{x^{k+1}}{\Gamma}$.

Note that each entry in the map lists $\mathcal{M}_i$ and $\mathcal{N}_j$ also appears within the variable valuations $\vars$, so projecting them separately is consistent with projecting $\vars$.
At this point, the $k$-ranked components are also unknown.

To complete the construction of a valid $k$-ranked counterexample, we require that the arguments passed to $\scalarf$ in the $k$-ranked setting match those in the $(k\!+\!1)$-ranked setting.
This imposes a set of \emph{constraints} on the projection $\Gamma$.
These constraints play a key role in computing a lower bound on the rank $k$ for which counterexample lowering is valid, as described in \secref{proof}.

\subsubsection{Handling Generalizations}

All the assumptions from \secref{assumptions} can be relaxed and the verification methodology can be extended accordingly.
To do so, we simply assume a more general symbolic form of the rule, i.e., the general form of $\scalarf$, accesses, and conditions.
More concretely, suppose the rule has \rclasses{} $c_1, \cdots, c_p$:
\begin{itemize}
    \item Consider any \rcrank{} map $m = \{c_1 \mapsto r_1, \cdots, c_p \mapsto r_p \}$ for the rewrite rule.

    \item For each \rclass{} $c_l$, compute a rank $k_l$ such that
    \[
        \forall i \geq k_l, \ \mathsf{valid}(R, m[c_l \mapsto i]) \rightarrow \mathsf{valid}(R, m[c_l \mapsto i\!+\!1]).
    \]
    Similar to the single-\agg{}-axis case, we reason using the contrapositive:
    \begin{equation*}
        \neg\mathsf{valid}(R, m[c_l \mapsto k_l\!+\!1]) \rightarrow \neg\mathsf{valid}(R, m[c_l \mapsto k_l])
    \end{equation*}
    This states that, given a counterexample at rank $k_l\!+\!1$ for all \agg{}-axes of \rclass{} $c_l$, we must construct a counterexample at rank $k_l$ for all those \agg{}-axes.
    Importantly, all \agg{}-axes with the same \rclass{} must have the same rank in any valid instantiation, so we project all such \agg{}-axes together.
    The ranks of all other \rclasses{} remain fixed throughout the lowering.

    Expanding the validity condition using \eqnref{rule-scalarf-expanded-i}, we obtain:
    \begin{align} \eqnlabel{starter-eq}
        \exists ~ v^{k_l+1} \in \vars^{k_l+1}, ~ &C^{k_l+1} \wedge \valexp(\mathsf{LHS}^{k_l+1}) \wedge \exists A^{k_l+1} \in \access(\mathsf{LHS}^{k_l+1}), \neg \ \scalarf(\ \cdots \ ) \nonumber \\
        &\quad\quad\quad\quad{\big\downarrow} \nonumber \\
        \exists ~ v^{k_l} \in \vars^{k_l}, ~ &C^{k_l} \wedge \valexp(\mathsf{LHS}^{k_l}) \wedge \exists A^{k_l} \in \access(\mathsf{LHS}^{k_l}), \neg \ \scalarf(\ \cdots \ ) 
    \end{align}
    We observe that the computed rank $k_l$ is independent of the ranks of other \rclasses{} in $m$.

    \item After computing this sufficient rank for each \rclass{}, we obtain the \rcrank{} map $\{c_1 \mapsto k_1, \cdots, c_p \mapsto k_p\}$.

    \item Finally, we verify the rewrite rule for all concrete instantiations within these bounds:
    \[
        \bigwedge_{1 \leq r_1 \leq k_1} \cdots \bigwedge_{1 \leq r_p \leq k_p} \mathsf{valid}(R, \{c_1 \mapsto r_1, \cdots, c_p \mapsto r_p\})
    \]
    If this conjunction is a tautology, then the rewrite rule $R$ is verified in the unbounded setting.
\end{itemize}

\section{Proof and Helper Lemmas} \seclabel{proof}

In this section, we compute a sufficient rank for each \rclass{} in any rewrite rule.
We begin under the assumptions listed in \secref{assumptions} and incrementally relax them to handle more general settings.
Our approach is to reason from \eqnref{starter-eq}: we begin with a counterexample at rank $k\!+\!1$ and construct a corresponding counterexample at rank $k$.
In each case, we assume a particular form for the scalar function $\scalarf$, the accesses, and the conditions, and then derive a sufficient rank for that form.
% By showing that counterexamples at higher ranks can be “lowered” to those at sufficient rank, we demonstrate that verifying all instantiations up to that rank suffices for correctness in the unbounded setting.

\subsection{Sufficient Rank for Arbitrary Number of Conditions} \seclabel{arbitrary-cond}

\begin{lemma}
\lemmalabel{arbitrary-cond}
Let $R = \mathsf{LHS} \Rightarrow_C \mathsf{RHS}$ be a rewrite rule such that:
\begin{itemize}
    \item[(1)] $R$ contains at most one input tensor, denoted \textsf{X},
    \item[(2)] \textsf{X} has exactly one \agg{}-axis $x : c$,
    \item[(3)] $R$ does not contain any of the following operators: $\reduce$, $\expand$, $\tdot$, $\conv$, $\convbase$, $\relabel$, and
    \item[(4)] The scalar function $\scalarf$ comprises one access to \textsf{X} and $m$ conditions.
\end{itemize}
Then, the following holds:
\[
    \mathsf{max}(m, 1) \leq k \implies \mathsf{valid}(R, \{c \mapsto k\}) \rightarrow \mathsf{valid}(R, \{c \mapsto k\!+\!1\})
\]
\end{lemma}
\begin{proof}
The $\scalarf$ of $R$ has the following canonical form, comprising one access and $m$ conditions:
\[
    \scalarf(\mathsf{X}[x \mapsto \fmap(e, \mathcal{M})],~ \fold(g_1, \mathcal{N}_1), \cdots, \fold(g_m, \mathcal{N}_m))
\]
where $e$ is an \indtrans{} and $g_1, \cdots, g_m$ are condition predicates.

Our goal is to find a sufficient rank $k$ such that any counterexample at rank $k\!+\!1$ can be lowered to a counterexample at rank $k$.
Let $x^{k+1} = \{\underline{1}, \cdots, \underline{k\!+\!1}\}$ denote a ($k\!+\!1$)-ranked instantiation of $x$.
Let $\Gamma \subset x^{k+1}$ be a projection of size $k$.

Starting from \eqnref{starter-eq}, we have:
\begin{align*}
    \exists ~ v^{k+1} \in \vars^{k+1}, C^{k+1} &\wedge \valexp(\mathsf{LHS}^{k+1}) \wedge \exists A^{k+1} \in \access(\mathsf{LHS}^{k+1}), \\
    &\neg\scalarf(\mathsf{X}^{k+1}[x^{k+1} \mapsto \fmap(e, \mathcal{M}^{k+1})],~ \fold(g_1, \mathcal{N}_1^{k+1}), \cdots, \fold(g_m, \mathcal{N}_m^{k+1})) \\
    &\quad\quad\quad\quad\quad\quad\quad{\big\downarrow} \\
    \exists ~ v^{k} \in \vars^{k}, C^k \wedge &\ \valexp(\mathsf{LHS}^{k}) \wedge \exists A^k \in \access(\mathsf{LHS}^{k}), \\
    &\neg\scalarf(\mathsf{X}^k[x^k \mapsto \fmap(e, \mathcal{M}^k)],~ \fold(g_1, \mathcal{N}_1^{k}), \cdots, \fold(g_m, \mathcal{N}_m^{k}))
\end{align*}

We define the $k$-ranked components via projection with $\Gamma$:
\begin{itemize}
    \item $\mathsf{X}^k = \cproject{\mathsf{X}^{k+1}}{x}{\Gamma}$,
    \item $\vars^k = \cproject{\vars^{k+1}}{x}{\Gamma}$,
    \item $A^k = \cproject{A^{k+1}}{x}{\Gamma}$,
    \item $\mathcal{M}^k = \cproject{\mathcal{M}^{k+1}}{x}{\Gamma}$, and
    \item for all $j \in \{1 \cdots m\}$, $\mathcal{N}_j^k = \cproject{\mathcal{N}_j^{k+1}}{x}{\Gamma}$
\end{itemize}

%These components are unknown since $\Gamma$ is unknown.
Since $C^{k+1} \Rightarrow C^k$ and $\valexp(\mathsf{LHS}^{k+1}) \Rightarrow \valexp(\mathsf{LHS}^{k})$, the precondition remains satisfied and the \textsf{LHS} expression remains valid in the $k$-ranked counterexample (\theoremref{precond-proj}).
What remains is ensuring equivalence of $\scalarf$ arguments in both counterexamples and deriving constraints on $\Gamma$:
\begin{itemize}
    \item Conditions: For all $l \in \{1 \cdots m\}$, we require that $\fold(g_l, \mathcal{N}_l^k)$ and $\fold(g_l, \mathcal{N}_l^{k+1})$ must have the same truth value.
    Let $\mathcal{N}_l^{k+1} = [m_1^{k+1}, m_2^{k+1}, \cdots]$.
    Then:
    \[
        \mathcal{N}_l^{k} = \mathcal{N}_l^{k+1}|^x_{\Gamma} = [m_1^{k+1}|^x_{\Gamma}, m_2^{k+1}|^x_{\Gamma}, \cdots] = [m_1^k, m_2^k, \cdots]
    \]
    Let $C_l$ be the set of constraints we get from this equation. We expand the definition of $\fold$:
    \begin{align*}
        \fold(g_l, \mathcal{N}_l^{k+1}) = \fold(g_l, [m_1^{k+1}, m_2^{k+1}, \cdots]) &= \bigwedge_{i = 1}^{k+1} g_l(m_1^{k+1}(\underline{i}), m_2^{k+1}(\underline{i}), \cdots) \\
        \fold(g_l, \mathcal{N}_l^k) = \fold(g_k, [m_1^k, m_2^k, \cdots]) &= \bigwedge_{\underline{j} \in \Gamma} g_l(m_1^{k}(\underline{j}), m_2^{k}(\underline{j}), \cdots)
    \end{align*}
    Let $b = \fold(g_l, \mathcal{N}_l^{k+1})$, which contains $k\!+\!1$ clauses, and $c = \fold(g_l, \mathcal{N}_l^k)$, which contains $k$ clauses.
    The value of $b$ is known since it depends entirely on the $(k\!+\!1)$-ranked counterexample.
    Let $r$ be the number of clauses in $b$ which evaluate to \true{}.
    The remaining $(k\!+\!1) - r$ clauses evaluate to \false{}.
    We perform a case analysis on $r$:
    \begin{itemize}
        \item $r < k$: Then $b = \false$.
        We require $c = \false$.
        We can see that for any size-$k$ projection $\Gamma$, at least one \false{} clause will be included.
        Hence, $c = \false$ always.
        There are no constraints on $\Gamma$, i.e., $C_l = \emptyset$.
        
        \item $r = k$: Then $b = \false$.
        We require $c = \false$.
        There exists exactly one $\underline{j} \in x^{k+1}$, for which the corresponding clause is \false{}.
        To ensure $c = \false$, $\Gamma$ must include the \ind{}-axis $\underline{j}$.
        Leaving out this \ind{}-axis will lead to $c$ being \true{}.
        Therefore, $C_l = \{\underline{j}\}$.
        
        \item $r = k\!+\!1$: Then $b = \true$.
        We require $c = \true$.
        We can see that for any size-$k$ projection $\Gamma$, $c$ will be \true{}.
        There are no constraints on $\Gamma$, i.e., $C_l = \emptyset$.
    \end{itemize}
    
    We get at most one constraint for each $g_l$, i.e., $|C_l| \leq 1$.
    % Since there are $m$ such conditions, all of the constraints need to be adhered to.
    The set of all constraints from the conditions is given by:
    \[
        C = \bigcup_{l = 1}^m C_l
    \]
    
    \item Accesses:
    We require the following,
    \[
        \mathsf{X}^k[x^k \mapsto \fmap(e, \mathcal{M}^k)] = \mathsf{X}^{k+1}[x^{k+1} \mapsto \fmap(e, \mathcal{M}^{k+1})]
    \]
    This can always be satisfied since we can just assign the value of $\mathsf{X}^{k+1}$ at the access $\{x^{k+1} \mapsto \fmap(e, \mathcal{M}^{k+1})\}$ to $\mathsf{X}^k[x^k \mapsto \fmap(e, \mathcal{M}^k)]$, irrespective of the projection.
    Therefore, no constraints on $\Gamma$ are introduced by this equation.
\end{itemize}

The final set of constraints is $C$.
%If $|C| > k$, then we cannot get a valid projection.
To construct a valid projection of size $k$, we require $|C| \leq k$.
We know that $|C| \leq m$, as each condition results in at most one constraint.
We also require $1 \leq k$ to avoid empty \agg{}-axes.
Therefore, $\mathsf{max}(m, 1) \leq k$ is a sufficient condition for a valid counterexample lowering.

Finally, we fully construct the $k$-ranked counterexample as follows:
\begin{itemize}
    \item Projection: The \ind{}-axes in $C$ must be a part of the projection $\Gamma$, but the other axes are unspecified.
    To get the final projection $\Gamma$, extend $C$ by any $(k - |C|)$ \ind{}-axes from $\{\underline{1}, \cdots, \underline{k\!+\!1}\} \mysetminus C$.
    \item Construct all $k$-ranked components as described above using $\Gamma$.
    Performing a projection ensures that the tensor shapes, access maps, and attributes are valid and the precondition is satisfied in the $k$-ranked counterexample (\theoremref{precond-proj}).
    \item Tensor values: Let $v = \mathsf{X}^{k+1}[x^{k+1} \mapsto \fmap(e, \mathcal{M}^{k+1})]$.
    We only require $\mathsf{X}^k$ to have the value $v$ at the access $\{x^k \mapsto \fmap(e, \mathcal{M}^k)\}$, and the values at other accesses are unspecified.
\end{itemize}
Thus, $k = \mathsf{max}(m, 1)$ is a sufficient rank for verifying the rewrite rule $R$.
\end{proof}

\subsection{Sufficient Rank for Arbitrary Number of Accesses} \seclabel{arbitrary-access}

\begin{lemma}
\lemmalabel{arbitrary-access}
Let $R = \mathsf{LHS} \Rightarrow_C \mathsf{RHS}$ be a rewrite rule such that:
\begin{itemize}
    \item[(1)] $R$ contains at most one input tensor, denoted \textsf{X},
    \item[(2)] \textsf{X} has exactly one \agg{}-axis $x : c$,
    \item[(3)] $R$ does not contain any of the following operators: $\reduce$, $\expand$, $\tdot$, $\conv$, $\convbase$, $\relabel$, and
    \item[(4)] The scalar function $\scalarf$ comprises $n$ accesses to \textsf{X} and $m$ conditions.
\end{itemize}
Then, the following holds:
\[
    \mathsf{max}(m + \binom{n}{2}, 1) \leq k \Rightarrow \mathsf{valid}(R, \{c \mapsto k\}) \rightarrow \mathsf{valid}(R, \{c \mapsto k\!+\!1\})
\]
\end{lemma}
\begin{proof}
The $\scalarf$ of $R$ has the following canonical form, comprising $n$ accesses and $m$ conditions:
\begin{align*}
    \scalarf(&\mathsf{X}[x \mapsto \fmap(e_1, \mathcal{M}_1)], \cdots, \mathsf{X}[x \mapsto \fmap(e_n, \mathcal{M}_n)], ~ \\
    &\fold(g_1, \mathcal{N}_1), \cdots, \fold(g_m, \mathcal{N}_m))
\end{align*}
where $e_1, \cdots, e_n$ \indtrans{}s and $g_1, \cdots, g_m$ are condition predicates.

Our goal is to find a sufficient rank $k$ such that any counterexample at rank $k\!+\!1$ can be lowered to a counterexample at rank $k$.
Let $x^{k+1} = \{\underline{1}, \cdots, \underline{k\!+\!1}\}$ denote a ($k\!+\!1$)-ranked instantiation of $x$.
Let $\Gamma \subset x^{k+1}$ be a projection of size $k$.

Starting from \eqnref{starter-eq}, we have:
\begin{align*}
    & \exists ~ v^{k+1} \in \vars^{k+1}, C^{k+1} \wedge \valexp(\mathsf{LHS}^{k+1}) \wedge \exists A^{k+1} \in \access(\mathsf{LHS}^{k+1}), \\
    &\quad\quad\quad\quad\neg\scalarf(\mathsf{X}^{k+1}[x^{k+1} \mapsto \fmap(e_1, \mathcal{M}_1^{k+1})], \cdots, \mathsf{X}^{k+1}[x^{k+1} \mapsto \fmap(e_n, \mathcal{M}_n^{k+1})], \\
    & \quad\quad\quad\quad\quad\quad\qquad\fold(g_1, \mathcal{N}_1^{k+1}), \cdots, \fold(g_m, \mathcal{N}_m^{k+1})) \\
    &\phantom{\scalarf(\mathsf{X}^{k+1}[x^{k+1} \mapsto \fmap(e_1, \mathcal{M}_1^{k+1})], \cdots}{\big\downarrow} \\
    &\exists ~ v^k \in \vars^k, C^k \wedge \valexp(\mathsf{LHS}^k) \wedge \exists A^k \in \access(\mathsf{LHS}^{k}), \\
    &\quad\quad\quad\quad\neg\scalarf(\mathsf{X}^{k}[x^k \mapsto \fmap(e_1, \mathcal{M}_1^k)],\cdots, \mathsf{X}^{k}[x^k \mapsto \fmap(e_n, \mathcal{M}_n^k)], \\
    &\quad\quad\quad\quad\quad\quad\qquad\fold(g_1, \mathcal{N}_1^k), \cdots, \fold(g_m, \mathcal{N}_m^k))
\end{align*}

We define the $k$-ranked components via projection with $\Gamma$:
\begin{itemize}
    \item $\mathsf{X}^k = \cproject{\mathsf{X}^{k+1}}{x}{\Gamma}$,
    \item $\vars^k = \cproject{\vars^{k+1}}{x}{\Gamma}$,
    \item $A^k = \cproject{A^{k+1}}{x}{\Gamma}$,
    \item for all $i \in \{1 \cdots n\}$, $\mathcal{M}_i^k = \cproject{\mathcal{M}_i^{k+1}}{x}{\Gamma}$, and
    \item for all $j \in \{1 \cdots m\}$, $\mathcal{N}_j^k = \cproject{\mathcal{N}_j^{k+1}}{x}{\Gamma}$
\end{itemize}

Since $C^{k+1} \Rightarrow C^k$ and $\valexp(\mathsf{LHS}^{k+1}) \Rightarrow \valexp(\mathsf{LHS}^{k})$, the precondition remains satisfied and the \textsf{LHS} expression remains valid in the $k$-ranked counterexample (\theoremref{precond-proj}).
What remains is ensuring equivalence of $\scalarf$ arguments in both counterexamples and deriving constraints on $\Gamma$:
\begin{itemize}
    \item Conditions: For all $l \in \{1 \cdots m\}$, we require that $\fold(g_l, \mathcal{N}_l^k)$ and $\fold(g_l, \mathcal{N}_l^{k+1})$ must have the same truth value.
    We assume the conditions to be independent of the accesses in the worst case.
    Therefore, the constraints as given by $C_1 = \bigcup_{l = 1}^m C_l$ still hold (\lemmaref{arbitrary-cond}).

    \item Accesses: For all $l \in \{1 \cdots n\}$,  we require:
    \[
        \mathsf{X}^k[x^k \mapsto \fmap(e_l, \mathcal{M}_l^k)] = \mathsf{X}^{k+1}[x^{k+1} \mapsto \fmap(e_l, \mathcal{M}_l^{k+1})]
    \]
    Let $v_l$ be the value of the tensor $X^{k+1}$ at the access $\{x^{k+1} \mapsto \fmap(e_l, \mathcal{M}_l^{k+1})\}$.
    The projection $\Gamma$ should not lead to any inconsistency with respect to the elements of $\mathsf{X}^k$.
    
    For example, if $v_1 \neq v_2$, then these values correspond to different accesses in the tensor $\mathsf{X}^{k+1}$, i.e.,
    $\fmap(e_1, \mathcal{M}_1^{k+1}) \neq \fmap(e_2, \mathcal{M}_2^{k+1})$.
    Therefore, there exists some $j \in \{1 \cdots (k\!+\!1)\}$ such that $e_1(\mathcal{M}_1^{k+1}(\underline{j})) \neq e_2(\mathcal{M}_2^{k+1}(\underline{j}))$, i.e., the access indices are unequal.
    
    We require that $\fmap(e_1, \mathcal{M}_1^{k})$ and $\fmap(e_2, \mathcal{M}_2^{k})$ are different accesses in $\mathsf{X}^k$, otherwise the same access in $\mathsf{X}^k$ would be assigned two different values $v_1$ and $v_2$ in the counterexample lowering.
    Selecting $\underline{j}$ as one of the \ind{}-axes in the projection will make sure that elements of $\mathsf{X}^k$ are consistent.
    
    Consider any arbitrary pair of values $v_r$ and $v_s$, where $r,s \in \{1 \cdots n\}$ and $r \neq s$.
    These correspond to two accesses in $\mathsf{X}^{k+1}$.
    There are $\binom{n}{2}$ such access pairs.
    Let $C_{r,s}$ be the set of constraints we get from this pair.
    The following cases arise:
    \begin{itemize}
        \item $v_r = v_s$: Then any projection will result in a valid counterexample since the lowering will not lead to any inconsistency. 
        For this case, $C_{r,s} = \emptyset$
        \item $v_r \neq v_s$: Then these values correspond to different accesses in $\mathsf{X}^{k+1}$, i.e.,
        \begin{equation*}
            \fmap(e_r, \mathcal{M}_r^{k+1}) \neq \fmap(e_s, \mathcal{M}_s^{k+1})
        \end{equation*}
        We require that they also correspond to different accesses in $\mathsf{X}^k$, i.e.,
        \begin{equation} \eqnlabel{ensure}
            \fmap(e_r, \mathcal{M}_r^k) \neq \fmap(e_r, \mathcal{M}_s^k)
        \end{equation}

        Let $E = \{\underline{l} \in x^{k+1} ~ | ~ e_r(\mathcal{M}_r^{k+1}(\underline{l})) = e_s(\mathcal{M}_s^{k+1}(\underline{l}))\}$ denote the set of \ind{}-axes having the same index value in accesses $r$ and $s$.
        The following cases arise:
        
        \begin{itemize}
            \item $|E| < k$: Then any size-$k$ projection will ensure that \eqnref{ensure} holds, since there exists at least one \ind{}-axis with different index values.
            Therefore, $C_{r,s} = \emptyset$.
            
            \item $|E| = k$: Then there exists exactly one $\underline{l} \in x^{k+1}$ for which $e_r(\mathcal{M}_r^{k+1}(\underline{l})) \neq e_s(\mathcal{M}_s^{k+1}(\underline{l}))$.
            To make sure that \eqnref{ensure} holds, $\underline{l}$ needs to be in the projection.
            If we leave this \ind{}-axis, then the accesses $r$ and $s$ will be the same in the tensor $\mathsf{X}^k$, but would be assigned two different values $v_r$ and $v_s$.
            Therefore, $C_{r,s} = \{\underline{l}\}$.
            
            \item $|E| = k\!+\!1$: This means that the accesses $r$ and $s$ are the same in $\mathsf{X}^{k+1}$.
            This implies that $v_r = v_s$, which is a contradiction.
            Therefore, such a case cannot arise.
        \end{itemize}
    \end{itemize}
\end{itemize}

% Each pair $(r, s)$ contributes at most one constraint, hence $|C_2| \leq \binom{n}{2}$.
% Since there are $\binom{n}{2}$ such pairs, all of the constraints need to be adhered to.
The set of all constraints from accesses is given by:
\[
    C_2 = \bigcup_{(r,s)} C_{r,s}
\]

We know that $|C_2| \leq \binom{n}{2}$, as each access pair results in at most one constraint.
The final set of constraints is $C = C_1 \cup C_2$.
The set of constraints $C$ is bounded by:
\begin{align*}
    |C| = |C_1 \cup C_2| &\leq |C_1| + |C_2| \leq m + \binom{n}{2}
\end{align*}

To construct a valid projection of size $k$, we require $|C| \leq k$.
We also require $1 \leq k$ to avoid empty \agg{}-axes.
Therefore, $\mathsf{max}(m + \binom{n}{2}, 1) \leq k$ is a sufficient condition for a valid counterexample lowering.

Finally, we fully construct the $k$-ranked counterexample as follows:
\begin{itemize}
    \item Projection: The \ind{}-axes in $C$ must be a part of the projection $\Gamma$, but the other axes are unspecified.
    To get the final projection $\Gamma$, extend $C$ by any $(k - |C|)$ \ind{}-axes from $\{\underline{1}, \cdots, \underline{k\!+\!1}\} \mysetminus C$.
    \item Construct all $k$-ranked components as described above using $\Gamma$.
    Performing a projection ensures that the tensor shapes, access maps, and attributes are valid and the precondition is satisfied in the $k$-ranked counterexample (\theoremref{precond-proj}).
    \item Tensor values: Let $v_l = \mathsf{X}^{k+1}[x^{k+1} \mapsto \fmap(e_l, \mathcal{M}_l^{k+1})]$ be the value at the $l^{\text{th}}$ access to $\mathsf{X}^{k+1}$.
    We only require $\mathsf{X}^k$ to have the value $v_l$ at the access $\{x^k \mapsto \fmap(e_l, \mathcal{M}_l^k)\}$, and the values at other accesses are unspecified.
    This should hold for all $l \in \{1 \cdots n\}$.
\end{itemize}
Thus, $k = \mathsf{max}(m + \binom{n}{2}, 1)$ is a sufficient rank for verifying the rewrite rule $R$.
\end{proof}

\subsection{Sufficient Rank for Arbitrary Number of Aggregated-Axes} \seclabel{arbitrary-agg}

\begin{lemma}
\lemmalabel{arbitrary-agg}
Let $R = \mathsf{LHS} \Rightarrow_C \mathsf{RHS}$ be a rewrite rule such that:
\begin{itemize}
    \item[(1)] $R$ contains at most one input tensor, denoted \textsf{X},
    \item[(2)] All \agg{}-axes appearing in the rule have the same \rclass{}, denoted $c$,
    \item[(3)] $R$ does not contain any of the following operators: $\reduce$, $\tdot$, $\conv$, $\convbase$, and
    \item[(4)] The scalar function $\scalarf$ comprises $n$ accesses to \textsf{X} and $m$ conditions.
\end{itemize}
Then, the following holds:
\[
    \mathsf{max}(m + \binom{n}{2}, 1) \leq k \Rightarrow \mathsf{valid}(R, \{c \mapsto k\}) \rightarrow \mathsf{valid}(R, \{c \mapsto k\!+\!1\})
\]
\end{lemma}

\begin{proof}
Let $x_1, \cdots, x_p$ denote all the \agg{}-axes that appear in the rewrite rule $R$, each belonging to the same \rclass{} $c$ by assumption.
Among these, let $x_1, \cdots, x_h$ be the \agg{}-axes that index the input tensor $\mathsf{X}$.
The remaining \agg{}-axes $x_{h+1}, \cdots, x_p$ occur as arguments to operators such as $\expand$ but do not index $\mathsf{X}$ directly.
Since all \agg{}-axes belong to the same class $c$, they are instantiated to the same rank in any valid instantiation of the rule.

The $\scalarf$ of $R$ has the following canonical form, comprising $n$ accesses and $m$ conditions:
\begin{align*}
    \scalarf(&\mathsf{X}[x_1 \mapsto \fmap({}^1e_1, {}^1\mathcal{M}_1), \cdots, x_h \mapsto \fmap({}^he_1, {}^h\mathcal{M}_1)], \cdots, \\
    & \phantom{\cdots}\mathsf{X}[x_1 \mapsto \fmap({}^1e_n, {}^1\mathcal{M}_n), \cdots, x_h \mapsto \fmap({}^he_n, {}^h\mathcal{M}_n)],\\
    & \phantom{\cdots\cdots} \bigwedge_{j = 1}^p \fold({}^jg_1, {}^j\mathcal{N}_1), \cdots, \bigwedge_{j = 1}^p \fold({}^jg_m, {}^j\mathcal{N}_m))
\end{align*}
where each ${}^je_i$ is an \indtrans{} and each ${}^jg_i$ is a condition predicate.

Our goal is to find a sufficient rank $k$ for the \rclass{} $c$, such that any counterexample at rank $k\!+\!1$ can be lowered to a counterexample at rank $k$.
Let $x_i^{k+1} = \{\underline{1}^i, \cdots, \underline{k\!+\!1}^i\}$ denote a $(k\!+\!1)$-ranked instantiation of each \agg{}-axis $x_i$.
Since all \agg{}-axes belong to the same \rclass{} $c$, for all $i, j \in \{1 \cdots p\}$, there exists a canonical bijection:
\[
    \mu_{i,j}^{k+1} = \mathsf{MapAxes}(c, x_i^{k+1}, x_j^{k+1}) : x_i^{k+1} \leftrightarrow x_j^{k+1}
\]

To construct a valid counterexample, we project each \agg{}-axis $x_i$ since they all have the same \rclass{} and are instantiated to the same rank.
Let $\Gamma_i \subset x_i^{k+1}$ be a projection of size $k$ for $x_i$.
These projections must be pairwise disjoint and must respect the canonical bijections:
\begin{equation} \eqnlabel{ensure-proj}
    \forall i, j \in \{1 \cdots p\}, \quad \mu_{i,j}^{k+1}(\Gamma_i) = \Gamma_j
\end{equation}

Starting from \eqnref{starter-eq}, we have:
\begin{align*}
    & \exists ~ v^{k+1} \in \vars^{k+1}, C^{k+1} \wedge \valexp(\mathsf{LHS}^{k+1}) \wedge \exists A^{k+1} \in \access(\mathsf{LHS}^{k+1}), \\
    &\phantom{\cdots\cdots\cdots\cdots}\neg\scalarf(v_1, \cdots, v_n, \, b_1, \cdots, b_m) \\
    &\phantom{\scalarf(\mathsf{X}^{k+1}[x^{k+1} \mapsto \fmap}{\big\downarrow} \\
    &\qquad\exists ~ v^k \in \vars^k, C^k \wedge \valexp(\mathsf{LHS}^k) \wedge \exists A^k \in \access(\mathsf{LHS}^{k}), \\
    &\phantom{\cdots\cdots\cdots\cdots}\neg\scalarf(w_1, \cdots, w_n, \, c_1, \cdots, c_m)
\end{align*}
where,
\begin{itemize}
    \item For each $l \in \{1 \cdots n\}$, $v_l$ denotes the $l^{\text{th}}$ access to $\mathsf{X}^{k+1}$ in the $(k\!+\!1)$-ranked counterexample:
    \[
        v_l = \mathsf{X}^{k+1}[x_1^{k+1} \mapsto \fmap({}^1e_l, {}^1\mathcal{M}_l^{k+1}), \cdots, x_h^{k+1} \mapsto \fmap({}^he_l, {}^h\mathcal{M}_l^{k+1})]
    \]
    \item For each $l \in \{1 \cdots n\}$, $w_l$ denotes the $l^{\text{th}}$ access to $\mathsf{X}^{k}$ in the $k$-ranked counterexample:
    \[
        w_l = \mathsf{X}^{k}[x_1^{k} \mapsto \fmap({}^1e_l, {}^1\mathcal{M}_l^{k}), \cdots, x_h^{k} \mapsto \fmap({}^he_l, {}^h\mathcal{M}_l^{k})]
    \]
    \item For each $l \in \{1 \cdots m\}$, $b_l$ denotes the $l^{\text{th}}$ condition in the $(k\!+\!1)$-ranked counterexample:
    \[
        b_l = \bigwedge_{j = 1}^p \fold({}^jg_l, {}^j\mathcal{N}_l^{k+1})
    \]
    \item For each $l \in \{1 \cdots m\}$, $c_l$ denotes the $l^{\text{th}}$ condition in the $k$-ranked counterexample:
    \[
        c_l = \bigwedge_{j = 1}^p \fold({}^jg_l, {}^j\mathcal{N}_l^{k})
    \]
\end{itemize}

We define the $k$-ranked components via projections with $\Gamma_1, \cdots, \Gamma_p$:
\begin{itemize}
    \item $\mathsf{X}^k = \cproject{\mathsf{X}^{k+1}}{x_1, \cdots, x_p}{\Gamma_1, \cdots, \Gamma_p}$,
    \item $\vars^k = \cproject{\vars^{k+1}}{x_1, \cdots, x_p}{\Gamma_1, \cdots, x_p}$,
    \item $A^k = \cproject{A^{k+1}}{x_1,\cdots,x_p}{\Gamma_1,\cdots,\Gamma_p}$,
    \item for all $l \in \{1 \cdots n\}, i \in \{1 \cdots h\}$, ${}^i\mathcal{M}_l^k = \cproject{{}^i\mathcal{M}_l^{k+1}}{x_i}{\Gamma_i}$,
    \item for all $l \in \{1 \cdots m\}, j \in \{1 \cdots p\}$, ${}^j\mathcal{N}_l^k = \cproject{{}^j\mathcal{N}_l^{k+1}}{x_j}{\Gamma_j}$, and
    \item for all $i, j \in \{1 \cdots p\}$, the canonical bijection $\mu_{i,j}^k = \mathsf{MapAxes}(c, \Gamma_i, \Gamma_j)$ is simply $\cproject{\mu_{i,j}^{k+1}}{x_i}{\Gamma_i}$.
    The resulting canonical bijections must satisfy:
    \begin{itemize}
        \item $\mu_{i,j}^k \circ \mu_{j,i}^k = \mathsf{id}$, and
        \item $\mu_{j,l}^k \circ \mu_{i,j}^k = \mu_{i,l}^k$.
    \end{itemize}
    To ensure that these properties hold, we must ensure that \eqnref{ensure-proj} is satisfied.
\end{itemize}

Since $C^{k+1} \Rightarrow C^k$ and $\valexp(\mathsf{LHS}^{k+1}) \Rightarrow \valexp(\mathsf{LHS}^{k})$, the precondition remains satisfied and the \textsf{LHS} expression remains valid in the $k$-ranked counterexample (\theoremref{precond-proj}).
What remains is ensuring equivalence of $\scalarf$ arguments in both counterexamples and deriving constraints on $\Gamma_1, \cdots, \Gamma_p$:
\begin{itemize}
    \item Conditions: For all $l \in \{1 \cdots m\}$, $c_l = \bigwedge_{j = 1}^p \fold({}^jg_l, {}^j\mathcal{N}_l^k)$ and $b_l = \bigwedge_{j = 1}^p \fold({}^jg_l, {}^j\mathcal{N}_l^{k+1})$ must have the same truth value.
    Let ${}^1C_l, \cdots, {}^pC_l$ be the set of constraints we get on $\Gamma_1, \cdots, \Gamma_p$ from this equation, respectively.
    
    The value of $b_l$ is known since it depends entirely on the $(k\!+\!1)$-ranked counterexample.
    Both $b_l$ and $c_l$ contain $p$ $\fold$s.
    Let $r$ be the number of $\fold$s in $b_l$ that evaluate to \true{}.
    The remaining $p - r$ $\fold$s evaluate to \false{}.
    We perform a case analysis on $r$:
    \begin{itemize}
        \item $r = p$: Then $b_l = \true$.
        We require $c_l = \true$.
        We can see that for any size-$k$ projections $\Gamma_1, \cdots, \Gamma_p$, all $\fold$s in $c_l$ will be \true{}.
        Therefore, ${}^1C_l = \cdots = {}^pC_l = \emptyset$.
        
        \item $r < p$: Then $b_l = \false$.
        We require $c_l = \false$.
        This implies that there exists some \agg{}-axis, say $x_i^{k+1}$, whose $\fold$ evaluates to \false{} in $b_l$, i.e., $\fold({}^ig_l, {}^i\mathcal{N}_l^{k+1}) = \false$.
        We ensure that the $\fold$ of $x_i^k$, i.e., $\fold({}^ig_l, {}^i\mathcal{N}_l^{k})$ also evaluates to \false{}, resulting in $c_l = \false$.
        This is similar to \lemmaref{arbitrary-cond}, and we get at most one constraint on $\Gamma_i$ in ${}^iC_l$ and no constraints on other projections.
        Therefore, $|{}^iC_l| \leq 1$ and ${}^jC_l = \emptyset$ for all $j \neq i$.
    \end{itemize}
    Each condition contributes at most one constraint to exactly one projection.
    % Since there are $m$ such conditions, all of the constraints need to be adhered to.
    The set of all constraints through the conditions is given by, where ${}^iC$ denotes the constraints on $\Gamma_i$:
    \begin{gather*}
        {}^1C = \bigcup_{l = 1}^m {}^1C_l, \ \cdots \ , {}^pC = \bigcup_{l = 1}^m {}^pC_l
    \end{gather*}
    Since we get at most $m$ constraints partitioned into these $p$ sets, we conclude that the total number of constraints is bounded by:
    \[
        |{}^1C| + \ \cdots \ + |{}^pC| \leq m
    \]

    \item Accesses: For all $l \in \{1 \cdots n\}$, $w_l$ must be the same as $v_l$.
    The projections $\Gamma_1, \cdots, \Gamma_p$ should not lead to inconsistencies with respect to the elements of $\mathsf{X}^k$.
    
    Consider any arbitrary pair of values $v_r$ and $v_s$, where $r,s \in \{1 \cdots n\}$ and $r \neq s$.
    These correspond to two accesses to $\mathsf{X}^{k+1}$.
    There are $\binom{n}{2}$ such access pairs.
    Let ${}^1D_{r,s}, \cdots, {}^hD_{r,s}$ be the constraints we get on $\Gamma_1, \cdots, \Gamma_h$ from this equation, respectively.
    We exclude ${}^{h+1}D_{r,s}, \cdots, {}^pD_{r,s}$ from the analysis because the corresponding \agg{}-axes do not index the tensor \textsf{X}, and hence do not introduce any inconsistencies.
    Thus, their constraints are equal to $\emptyset$.
    The following cases arise:
    \begin{itemize}
        \item $v_r = v_s$: Then any projection will result in a valid counterexample since the lowering will not lead to any inconsistency. 
        For this case, ${}^iD_{r,s} = \emptyset$ for all $i$.
        \item $v_r \neq v_s$: Then these values correspond to different accesses in $\mathsf{X}^{k+1}$, i.e.,
        \begin{align*}
            \{x_1^{k+1} \mapsto \fmap({}^1e_r, {}^1\mathcal{M}_r^{k+1}), \cdots, x_h^{k+1} \mapsto \fmap({}^he_r, {}^h\mathcal{M}_r^{k+1})\} \neq \\
            \{x_1^{k+1} \mapsto \fmap({}^1e_s, {}^1\mathcal{M}_s^{k+1}), \cdots, x_h^{k+1} \mapsto \fmap({}^he_s, {}^h\mathcal{M}_s^{k+1})\}
        \end{align*}
        We require that they also correspond to different accesses in $\mathsf{X}^k$, i.e.,
        \begin{align} \eqnlabel{ensure-agg}
            \{x_1^{k} \mapsto \fmap({}^1e_r, {}^1\mathcal{M}_r^{k}), \cdots, x_h^{k} \mapsto \fmap({}^he_r, {}^h\mathcal{M}_r^{k})\} \neq \nonumber \\
            \{x_1^{k} \mapsto \fmap({}^1e_s, {}^1\mathcal{M}_s^{k}), \cdots, x_h^{k} \mapsto \fmap({}^he_s, {}^h\mathcal{M}_s^{k})\}
        \end{align}
        
        Since the accesses do not match in $\mathsf{X}^{k+1}$, there exists some \agg{}-axis, say $x_i^{k+1}$, for $i \in \{1 \cdots h\}$, such that its corresponding access maps in the two accesses are unequal:
        \[
            \fmap({}^ie_r, {}^i\mathcal{M}_r^{k+1}) \neq \fmap({}^ie_s, {}^i\mathcal{M}_s^{k+1})
        \]
        We ensure that the access maps of $x_i^k$ are unequal in the two accesses, to ensure that \eqnref{ensure-agg} holds:
        \[
            \fmap({}^ie_r, {}^i\mathcal{M}_r^{k}) \neq \fmap({}^ie_s, {}^i\mathcal{M}_s^{k})
        \]
        This is similar to \lemmaref{arbitrary-access}, and we get at most one constraint on $\Gamma_i$ in ${}^iD_{r,s}$ and no constraints on other projections.
        Therefore, $|{}^iD_{r,s}| \leq 1$ and ${}^jD_{r,s} = \emptyset$ for all $j \neq i$.
    \end{itemize}
    Each access pair contributes at most one constraint to exactly one projection.
    The set of all constraints through the accesses is given by, where ${}^iD$ denotes the constraints on $\Gamma_i$:
    \begin{gather*}
        {}^1D = \bigcup_{(r,s)} {}^1D_{r,s}, \ \cdots \ , {}^pD = \bigcup_{(r,s)} {}^pD_{r,s}
    \end{gather*}
    Since we get at most $\binom{n}{2}$ constraints partitioned into these $p$ sets, we conclude that the total number of constraints is bounded by:
    \[
        |{}^1D| + \ \cdots \ + |{}^pD| \leq \binom{n}{2}
    \]
\end{itemize}

We now define the combined set of constraints (from both conditions and accesses) on each projection $\Gamma_i$ as $F_i = {}^iC \cup {}^iD$.
The total number of constraints across all projections is bounded by:
\begin{align*}
    |F_1| + \ \cdots \ + |F_p| &= |{}^1C \cup {}^1D| + \ \cdots \ + |{}^pC \cup {}^pD| \\
    &\leq |{}^1C| + |{}^1D| + \ \cdots \ + |{}^pC| + |{}^pD| \\
    &= |{}^1C| + \ \cdots \ + |{}^pC| + |{}^1D| + \ \cdots \ + |{}^pD| \\
    &\leq m + \binom{n}{2}
\end{align*}
Thus, the total number of constraints is at most $m + \binom{n}{2}$, partitioned across the $p$ constraint sets.
This implies that $\mathsf{max}_i |F_i| \leq m + \binom{n}{2}$ (the maximum possible constraints on any projection).
For all the projections to be valid, we require $|F_1| \leq k, \cdots, |F_p| \leq k$, or equivalently, $\mathsf{max}_i |F_i| \leq k$.
We also need $1 \leq k$ to avoid empty \agg{}-axes.
Therefore, $\mathsf{max}(m + \binom{n}{2}, 1) \leq k$ is a sufficient condition for a valid counterexample lowering.

Finally, we fully construct the $k$-ranked counterexample as follows:
\begin{itemize}
    \item Projection: For all $j \in \{1 \cdots p\}$, the \ind{}-axes in $F_j$ must be a part of the projection $\Gamma_j$.
    Additionally, the projections should respect the canonical bijections of the $(k\!+\!1)$-ranked counterexample.
    We compute the final projections using the \textsc{ComputeProjections} routine described in \algref{expand-proj}.
    It takes the final constraints $F_1, \cdots, F_p$ as input and returns the projections $\Gamma_1, \cdots, \Gamma_p$.
    In lines 4-9, we expand each $F_j$ using the canonical bijection images of other constraint sets.
    
    After line 9, the sets $F_1', \cdots, F_p'$ are of the same size.
    In fact, there exists a bijection between all the sets: the bijection between $F_i'$ and $F_j'$ is simply $\cproject{\mu_{i,j}^{k+1}}{x_i}{F_i'}$.
    % We prove this by showing that there exists a bijection between all the sets.
    % Each projection $\Gamma_i$ must include the \ind{}-axes in $F_i$, and the projections must respect the canonical bijections of the $(k\!+\!1)$-ranked counterexample.
    % This means that if $v \in F_i'$, then its images under the bijections must be included in all the other sets $F_j'$.
    % Since the bijections are one-to-one and onto, every element that appears in any $F_i'$
    % has a unique counterpart in each other $F_j'$.
    This ensures that all $F_i'$ contain the exact same number of elements, since their elements are simply relabelings of each other via the bijections.
    
    We also prove that the size of the sets $F_1', \cdots, F_p'$ is bounded by $k$, the projection size.
    For all $i, j$, $|\Delta_{i,j}| = |F_i|$ since $\mu_{i,k}^{k+1}$ is a bijection.
    For all $j \in \{1 \cdots p\}$, the set $F_j'$ is $\displaystyle F_j \cup \bigcup_{i = 1, i \neq j}^p \Delta_{i,j}$, whose size is bounded by:
    \begin{align*}
        |F_j'| &= \left| F_j \cup \bigcup_{i = 1, i \neq j}^p \Delta_{i,j}\right| \\
        &\leq |F_j| + \sum_{i = 1, i \neq j}^p |\Delta_{i, j}| \\
        &= |F_j| + \sum_{i = 1, i \neq j}^p |F_i| = \sum_{i = 1}^p |F_i| \leq m + \binom{n}{2} \leq k
    \end{align*}
    Therefore, the number of constraints in each set does not exceed the projection size $k$.
    On line 10, we add any $(k - |F_1'|)$ \ind{}-axes from $x_1^{k+1} \mysetminus F_1'$ to complete the projection $\Gamma_1$.
    In lines 11-13, we add their bijection images to complete the other projections.
    
    \item Construct all $k$-ranked components as described above using $\Gamma$.
    Performing a projection ensures that the tensor shapes, access maps, and attributes are valid and the precondition is satisfied in the $k$-ranked counterexample (\theoremref{precond-proj}).

    \item Tensor values: For all $i \in \{1 \cdots n\}$, we require $w_i$, the value of $\mathsf{X}^k$ at the $i^{\text{th}}$ access, to be same as $v_i$, the value of $\mathsf{X}^{k+1}$ at the $i^{\text{th}}$ access.
    The values at other accesses in $\mathsf{X}^k$ are unspecified.
\end{itemize}

\SetKwComment{Comment}{/*}{ */}
\begin{algorithm}[t]
\caption{\textsc{ComputeProjections}} \alglabel{expand-proj}
\KwIn{Constraint sets $F_1, \cdots, F_p$}
\KwOut{Final projections $\Gamma_1, \cdots, \Gamma_p$}
\SetKwProg{Fn}{Function}{ :}{}
\For{$i \in \{1 \cdots p\}$}{
    $F_i' \gets F_i$\;
}
\For{$j \in \{1 \cdots p\}$}{
    \For{$i \in \{1 \cdots p\}, i \neq j$}{
        $\Delta_{i,j} \gets \mu_{i,j}^{k+1}(F_i)$ \;
        $F_j' \gets F_j' \cup \Delta_{i,j}$ \;
    }
}
$R \gets \textsc{Choice}(x_1^{k+1} \mysetminus F_1', \ k - |F_1'|)$\;
\For{$i \in \{1 \cdots p\}$}{
    $\Gamma_i \gets F_i' \cup \mu_{1,i}^{k+1}(R)$\;
}
\Return $\Gamma_1, \cdots, \Gamma_p$
\end{algorithm}

Thus, $k = \mathsf{max}(m + \binom{n}{2}, 1)$ is a sufficient rank for verifying the rewrite rule $R$.

\end{proof}

\subsection{Arbitrary Number of RClasses} \seclabel{arbitrary-rclass}

\begin{lemma}
\lemmalabel{arbitrary-rclass}
    Let $R = \mathsf{LHS} \Rightarrow_C \mathsf{RHS}$ be a rewrite rule such that
    \begin{itemize}
        \item[(1)] $R$ contains at most one input tensor, denoted \textsf{X},
        \item[(2)] $R$ does not contain any of the following operators: $\reduce, \tdot, \conv$, and $\convbase$,
        \item[(3)] $M$ is an \rcrank{} map for $R$ and $c$ is an \rclass{} appearing in the rule,
        \item[(4)] The scalar function $\scalarf$ comprises: (1) $n$ accesses to \textsf{X} and (2) $m$ conditions.
    \end{itemize}
    Let $n_\mathsf{X} = n$ if an \agg{}-axis of \textsf{X} has \rclass{} $c$, and $n_\mathsf{X} = 0$ otherwise.
    Then the following holds:
    \[
        \mathsf{max}(m + \binom{n_\mathsf{X}}{2}, 1) \leq k \Rightarrow \mathsf{valid}(R, M[c \mapsto k]) \rightarrow \mathsf{valid}(R, M[c \mapsto k\!+\!1])
    \]
\end{lemma}

\begin{proof}
Let $x_1,\cdots, x_p$ denote all the \agg{}-axes of \rclass{} $c$ that appear in the rewrite rule $R$ .
Among these, let $x_1, \cdots, x_h$ be the \agg{}-axes that index the input tensor \textsf{X}.
The \agg{}-axes $x_{h+1}, \cdots, x_p$ occur as arguments to operators such as $\expand$ but do not index \textsf{X} directly.
All other \agg{}-axes in $R$ belong to \rclasses{} other than $c$.
The $\scalarf$ of $R$ has the following canonical form, comprising $n_\mathsf{X}$ relevant accesses to \textsf{X} and $m$ conditions:
\begin{align*}
    \scalarf(&\mathsf{X}[x_1 \mapsto \fmap({}^1e_1, {}^1\mathcal{M}_1), \cdots, x_h \mapsto \fmap({}^he_1, {}^h\mathcal{M}_1) \cup Y_1], \cdots, \\
    & \phantom{\cdots}\mathsf{X}[x_1 \mapsto \fmap({}^1e_{n_\mathsf{X}}, {}^1\mathcal{M}_{n_\mathsf{X}}), \cdots, x_h \mapsto \fmap({}^he_{n_\mathsf{X}}, {}^h\mathcal{M}_{n_\mathsf{X}}) \cup Y_{n_\mathsf{X}}],\\
    & \phantom{\cdots\cdots} \phi_1 \wedge \bigwedge_{j = 1}^p \fold({}^jg_1, {}^j\mathcal{N}_1), \cdots, \phi_m \wedge \bigwedge_{j = 1}^p \fold({}^jg_m, {}^j\mathcal{N}_m))
\end{align*}
where each ${}^je_i$ is an \indtrans{} and each ${}^jg_i$ is a condition predicate.
The maps $Y_i$ denote the access maps for the remaining \agg{}-axes (those whose \rclass{} is not $c$), while formulas $\phi_i$ collect the $\fold$-conditions for these remaining \agg{}-axes.

Our goal is to find a sufficient rank $k$ for the \rclass{} $c$, such that any counterexample at rank $k\!+\!1$ can be lowered to a counterexample at rank $k$.
Let $x_i^{k+1} = \{\underline{1}^i, \cdots, \underline{k\!+\!1}^i\}$ denote a $(k\!+\!1)$-ranked instantiation of each \agg{}-axis $x_i$.
Since all \agg{}-axes $x_1, \cdots, x_p$ belong to the same \rclass{} $c$, there exists a canonical bijection between any pair:
\[
    \mu_{i,j}^{k+1} = \mathsf{MapAxes}(c, x_i^{k+1}, x_j^{k+1}) : x_i^{k+1} \leftrightarrow x_j^{k+1}
\]

To construct a valid counterexample, we project each \agg{}-axis $x_i$ since they all have the same \rclass{} and are instantiated to the same rank.
Let $\Gamma_i \subset x_i^{k+1}$ be a projection of size $k$ for $x_i$.
These projections must be pairwise disjoint and must respect the canonical bijections:
\begin{equation} \eqnlabel{ensure-proj-1}
    \forall i, j \in \{1 \cdots p\}, \quad \mu_{i,j}^{k+1}(\Gamma_i) = \Gamma_j
\end{equation}

Starting from \eqnref{starter-eq}, we have:
\begin{align*}
    & \exists ~ v^{k+1} \in \vars^{k+1}, C^{k+1} \wedge \valexp(\mathsf{LHS}^{k+1}) \wedge \exists A^{k+1} \in \access(\mathsf{LHS}^{k+1}), \\
    &\phantom{\cdots\cdots\cdots\cdots}\neg\scalarf(v_1, \cdots, v_{n_\mathsf{X}}, \, b_1, \cdots, b_m) \\
    &\phantom{\scalarf(\mathsf{X}^{k+1}[x^{k+1} \mapsto \fmap}{\big\downarrow} \\
    &\qquad\exists ~ v^k \in \vars^k, C^k \wedge \valexp(\mathsf{LHS}^k) \wedge \exists A^k \in \access(\mathsf{LHS}^{k}), \\
    &\phantom{\cdots\cdots\cdots\cdots}\neg\scalarf(w_1, \cdots, w_{n_\mathsf{X}}, \, c_1, \cdots, c_m)
\end{align*}
where,
\begin{itemize}
    \item For each $l \in \{1 \cdots n_\mathsf{X}\}$, $v_l$ denotes the $l^{\text{th}}$ access to $\mathsf{X}^{k+1}$ in the $(k\!+\!1)$-ranked counterexample:
    \[
        v_l = \mathsf{X}^{k+1}[x_1^{k+1} \mapsto \fmap({}^1e_l, {}^1\mathcal{M}_l^{k+1}), \cdots, x_h^{k+1} \mapsto \fmap({}^he_l, {}^h\mathcal{M}_l^{k+1}) \cup Y_l^{k+1}]
    \]
    % where $S_l$ is the access map for the \agg{}-axes in \textsf{X} which do not belong to the \rclass{} $c$.
    \item For all $l \in \{1 \cdots n_\mathsf{X}\}$, $w_l$ denotes the $l^{\text{th}}$ access to $\mathsf{X}^{k}$ in the $k$-ranked counterexample:
    \[
        w_l = \mathsf{X}^{k}[x_1^{k} \mapsto \fmap({}^1e_l, {}^1\mathcal{M}_l^{k}), \cdots, x_h^{k} \mapsto \fmap({}^he_l, {}^h\mathcal{M}_l^{k}) \cup Y_l^k]
    \]
    % Note that we keep $S_l$ unchanged since the rank of other \rclasses{} is unchanged.
    \item For each $l \in \{1 \cdots m\}$, $b_l$ denotes the $l^{\text{th}}$ condition in the $(k\!+\!1)$-ranked counterexample:
    \[
        b_l = \phi_l^{k+1} \wedge \bigwedge_{j = 1}^p \fold({}^jg_l, {}^j\mathcal{N}_l^{k+1})
    \]
    % where $\phi_l$ is the conjunction of $\fold$-condition for the \agg{}-axes which do not belong to the \rclass{} $c$.
    \item For all $l \in \{1 \cdots m\}$, $c_l$ denotes the $l^{\text{th}}$ condition in the $k$-ranked counterexample:
    \[
        c_l = \phi_l^{k} \wedge \bigwedge_{j = 1}^p \fold({}^jg_l, {}^j\mathcal{N}_l^{k})
    \]
    % Note that we keep $\phi_l$ unchanged since the rank of other \rclasses{} is unchanged.
\end{itemize}

We define the $k$-ranked components via projections with $\Gamma_1, \cdots, \Gamma_p$:
\begin{itemize}
    \item $\mathsf{X}^k = \cproject{\mathsf{X}^{k+1}}{x_1, \cdots, x_p}{\Gamma_1, \cdots, \Gamma_p}$,
    \item $\vars^k = \cproject{\vars^{k+1}}{x_1, \cdots, x_p}{\Gamma_1, \cdots, x_p}$,
    \item $A^k = \cproject{A^{k+1}}{x_1,\cdots,x_p}{\Gamma_1,\cdots,\Gamma_p}$,
    \item for all $l \in \{1 \cdots n_\mathsf{X}\}, i \in \{1 \cdots h\}$, ${}^i\mathcal{M}_l^k = \cproject{{}^i\mathcal{M}_l^{k+1}}{x_i}{\Gamma_i}$,
    \item for all $l \in \{1 \cdots n_\mathsf{X}\}$, $Y_l^k = \cproject{Y_l^{k+1}}{x_1,\cdots,x_p}{\Gamma_1,\cdots,\Gamma_p} = Y_l^{k+1}$, since $Y_l$ is independent of the \rclass{} $c$,
    \item for all $l \in \{1 \cdots m\}, j \in \{1 \cdots p\}$, ${}^j\mathcal{N}_l^k = \cproject{{}^j\mathcal{N}_l^{k+1}}{x_j}{\Gamma_j}$,
    \item for all $l \in \{1 \cdots m\}$, $\phi_l^k = \cproject{\phi_l^{k+1}}{x_1,\cdots,x_p}{\Gamma_1,\cdots,\Gamma_p} = \phi_l^{k+1}$, since $\phi_l$ is independent of the \rclass{} $c$, and
    \item for all $i, j \in \{1 \cdots p\}$, the canonical bijection $\mu_{i,j}^k = \mathsf{MapAxes}(c, \Gamma_i, \Gamma_j)$ is simply $\cproject{\mu_{i,j}^{k+1}}{x_i}{\Gamma_i}$.
    The resulting canonical bijections must satisfy:
    \begin{itemize}
        \item $\mu_{i,j}^k \circ \mu_{j,i}^k = \mathsf{id}$, and
        \item $\mu_{j,l}^k \circ \mu_{i,j}^k = \mu_{i,l}^k$.
    \end{itemize}
    To ensure that these properties hold, we must ensure that \eqnref{ensure-proj-1} is satisfied.
\end{itemize}
To obtain the $k$-ranked counterexample, we project each component of the $(k\!+\!1)$-ranked counterexample along the projections $\Gamma_i$.
Components not involving the \rclass{} $c$ are unchanged.

Since $C^{k+1} \Rightarrow C^k$ and $\valexp(\mathsf{LHS}^{k+1}) \Rightarrow \valexp(\mathsf{LHS}^{k})$, the precondition remains satisfied and the \textsf{LHS} expression remains valid in the $k$-ranked counterexample (\theoremref{precond-proj}).
What remains is ensuring equivalence of $\scalarf$ arguments in both counterexamples and deriving constraints on $\Gamma_1, \cdots, \Gamma_p$:
\begin{itemize}
    \item Conditions: For all $l \in \{1 \cdots m\}$, we require $c_l = \phi_l^k \wedge \bigwedge_{j = 1}^p \fold({}^jg_l, {}^j\mathcal{N}_l^k)$ and $b_l = \phi_l^{k+1} \wedge \bigwedge_{j = 1}^p \fold({}^jg_l, {}^j\mathcal{N}_l^{k+1})$ must have the same truth value.
    
    The value of $b_l$ is known since it depends entirely on the $(k\!+\!1)$-ranked counterexample.
    To make sure that $c_l$ and $b_l$ have the same truth value, it is sufficient to make the truth values of $c_l' = \bigwedge_{j = 1}^p \fold({}^jg_l, {}^j\mathcal{N}_l^k)$ and $b_l' = \bigwedge_{j = 1}^p \fold({}^jg_l, {}^j\mathcal{N}_l^{k+1})$ the same since $\phi_l^k = \phi_l^{k+1}$.
    This is similar to \lemmaref{arbitrary-agg}, and we get at most one constraint on exactly one projection.
    
    The set of all constraints through the conditions is given by, where ${}^iC$ denotes the constraints on $\Gamma_i$.
    Since we get at most $m$ constraints partitioned into these $p$ sets, we conclude that the total number of constraints is bounded by:
    \[
        |{}^1C| + \ \cdots \ + |{}^pC| \leq m
    \]

    \item Accesses: For all $l \in \{1 \cdots n_\mathsf{X}\}$, $w_l$ must be the same as $v_l$.
    The projections $\Gamma_1, \cdots, \Gamma_p$ should not lead to inconsistencies with respect to the elements of $\mathsf{X}^k$.
    
    Consider any arbitrary pair of values $v_r$ and $v_s$, where $r,s \in \{1 \cdots n_\mathsf{X}\}$ and $r \neq s$.
    These correspond to two accesses to $\mathsf{X}^{k+1}$.
    There are $\binom{n_\mathsf{X}}{2}$ such access pairs.
    Let ${}^1D_{r,s}, \cdots, {}^hD_{r,s}$ be the constraints we get on $\Gamma_1, \cdots, \Gamma_h$ from this equation, respectively.
    We exclude ${}^{h+1}D_{r,s}, \cdots, {}^pD_{r,s}$ from the analysis because the corresponding \agg{}-axes do not index the tensor \textsf{X}, and hence do not introduce any inconsistencies.
    Thus, their constraints are equal to $\emptyset$.
    The following cases arise:
    \begin{itemize}
        \item $v_r = v_s$: Then any projection will result in a valid counterexample since the lowering will not lead to any inconsistency. 
        For this case, ${}^iD_{r,s} = \emptyset$ for all $i$.
        \item $v_r \neq v_s$: Then these values correspond to different accesses in $\mathsf{X}^{k+1}$, i.e.,
        \begin{align*}
            \{x_1^{k+1} \mapsto \fmap({}^1e_r, {}^1\mathcal{M}_r^{k+1}), \cdots, x_h^{k+1} \mapsto \fmap({}^he_r, {}^h\mathcal{M}_r^{k+1})\} \cup Y_r^{k+1} \neq \\
            \{x_1^{k+1} \mapsto \fmap({}^1e_s, {}^1\mathcal{M}_s^{k+1}), \cdots, x_h^{k+1} \mapsto \fmap({}^he_s, {}^h\mathcal{M}_s^{k+1})\} \cup Y_s^{k+1}
        \end{align*}
        We require that they also correspond to different accesses in $\mathsf{X}^k$, i.e.,
        \begin{align} \eqnlabel{ensure-rclass}
            \{x_1^{k} \mapsto \fmap({}^1e_r, {}^1\mathcal{M}_r^{k}), \cdots, x_h^{k} \mapsto \fmap({}^he_r, {}^h\mathcal{M}_r^{k})\} \cup Y_r^k \neq \nonumber \\
            \{x_1^{k} \mapsto \fmap({}^1e_s, {}^1\mathcal{M}_s^{k}), \cdots, x_h^{k} \mapsto \fmap({}^he_s, {}^h\mathcal{M}_s^{k})\} \cup Y_s^k
        \end{align}
        
        Since the accesses do not match in $\mathsf{X}^{k+1}$, two cases arise:
        (1) $Y_r^{k+1} \neq Y_s^{k+1}$. This implies $Y_r^k \neq Y_s^k$, ensuring that \eqnref{ensure-rclass} holds. This case results in no constraints on any projection,
        (2) there exists some \agg{}-axis, say $x_i^{k+1}$, for $i \in \{1 \cdots h\}$, such that its corresponding access maps in the two accesses are unequal:
        \[
            \fmap({}^ie_r, {}^i\mathcal{M}_r^{k+1}) \neq \fmap({}^ie_s, {}^i\mathcal{M}_s^{k+1})
        \]
        We ensure that the access maps of $x_i^k$ are unequal in the two accesses, to ensure that \eqnref{ensure-rclass} holds:
        \[
            \fmap({}^ie_r, {}^i\mathcal{M}_r^{k}) \neq \fmap({}^ie_s, {}^i\mathcal{M}_s^{k})
        \]
        This is similar to \lemmaref{arbitrary-access}, and we get at most one constraint on $\Gamma_i$ in ${}^iD_{r,s}$ and no constraints on other projections.
        Therefore, $|{}^iD_{r,s}| \leq 1$ and ${}^jD_{r,s} = \emptyset$ for all $j \neq i$.
    \end{itemize}
    Each access pair contributes at most one constraint to exactly one projection.
    The set of all constraints through the accesses is given by, where ${}^iD$ denotes the constraints on $\Gamma_i$:
    \begin{gather*}
        {}^1D = \bigcup_{(r,s)} {}^1D_{r,s}, \ \cdots \ , {}^pD = \bigcup_{(r,s)} {}^pD_{r,s}
    \end{gather*}
    Since we get at most $\binom{n_\mathsf{X}}{2}$ constraints partitioned into these $p$ sets, we conclude that the total number of constraints is bounded by:
    \[
        |{}^1D| + \ \cdots \ + |{}^pD| \leq \binom{n_\mathsf{X}}{2}
    \]
\end{itemize}

We now define the combined set of constraints (from both conditions and accesses) on each projection $\Gamma_i$ as $F_i = {}^iC \cup {}^iD$.
The total number of constraints across all projections is bounded by:
\begin{align*}
    |F_1| + \ \cdots \ + |F_p| &= |{}^1C \cup {}^1D| + \ \cdots \ + |{}^pC \cup {}^pD| \\
    &\leq |{}^1C| + |{}^1D| + \ \cdots \ + |{}^pC| + |{}^pD| \\
    &= |{}^1C| + \ \cdots \ + |{}^pC| + |{}^1D| + \ \cdots \ + |{}^pD| \\
    &\leq m + \binom{n_\mathsf{X}}{2}
\end{align*}

Similar to \lemmaref{arbitrary-agg}, we get $\mathsf{max}(m + \binom{n_\mathsf{X}}{2}, 1) \leq k$ as a sufficient condition for a valid counterexample lowering.
Finally, we fully construct the $k$-ranked counterexample as follows:
\begin{itemize}
    \item Projection: For all $j \in \{1 \cdots p\}$, the \ind{}-axes in $F_j$ must be a part of the projection $\Gamma_j$.
    Additionally, the projections should respect the canonical bijections of the $(k\!+\!1)$-ranked counterexample.
    We compute the final projections using the \textsc{ComputeProjections} routine described in \algref{expand-proj}.
    
    \item Construct all $k$-ranked components as described above using $\Gamma$.
    Performing a projection ensures that the tensor shapes, access maps, and attributes are valid and the precondition is satisfied in the $k$-ranked counterexample (\theoremref{precond-proj}).

    \item Tensor values: For all $i \in \{1 \cdots n_\mathsf{X}\}$, we require $w_i$, the value of $\mathsf{X}^k$ at the $i^{\text{th}}$ access, to be same as $v_i$, the value of $\mathsf{X}^{k+1}$ at the $i^{\text{th}}$ access.
    The values at other accesses in $\mathsf{X}^k$ are unspecified.
\end{itemize}

Thus, $k = \mathsf{max}(m + \binom{n_\mathsf{X}}{2}, 1)$ is a sufficient rank for the \rclass{} $c$.

\end{proof}

\subsection{Arbitrary Number of Tensors} \seclabel{arbitrary-tensor}

\begin{lemma}
\lemmalabel{arbitrary-tensor}
    Let $R = \mathsf{LHS} \Rightarrow_C \mathsf{RHS}$ be a rewrite rule such that
    \begin{itemize}
        \item[(1)] $R$ does not contain any of the following operators: $\reduce, \tdot, \conv$, and $\convbase$,
        \item[(2)] $M$ is an \rcrank{} map for $R$ and $c$ is an \rclass{} appearing in the rule,
        \item[(3)] $R$ contains $q$ tensors $\mathsf{X}_1,\cdots,\mathsf{X}_q$ that have an \agg{}-axis belonging to the \rclass{} $c$,
        \item[(3)] The scalar function $\scalarf$ comprises: (1) $n_{\mathsf{X}_i}$ accesses to $\mathsf{X}_i$ and (2) $m$ conditions.
    \end{itemize}
    Then the following holds:
    \[
        \mathsf{max}(m + \sum_{i=1}^q \binom{n_{\mathsf{X}_i}}{2}, 1) \leq k \Rightarrow \mathsf{valid}(R, M[c \mapsto k]) \rightarrow \mathsf{valid}(R, M[c \mapsto k\!+\!1])
    \]
\end{lemma}

\begin{proof}

Only tensors that contain an \agg{}-axis belonging to the \rclass{} $c$ can introduce constraints that on the sufficient rank for $c$.
Accesses to any other tensors retain their values during the counterexample lowering, and hence do not contribute to the bound.
For each tensor $\mathsf{X}_i$ that contains an \agg{}-axis with the \rclass{} $c$, every pair of its accesses can introduce at most one constraint to exactly one projection.
Therefore, the total number of constraints through accesses of $\mathsf{X}_i$ is bounded by $\binom{n_{\mathsf{X}_i}}{2}$.

The contributions from different tensors are summed independently because accesses across tensors do not lead to any constraints (we assume that distinct tensors do not alias; more \secref{discussion-app}).
Consequently, the total number of constraints is simply bounded by sum of per-tensor contributions, together with the $m$ constraints from the conditions in the worst case.
The counterexample construction is similar to \lemmaref{arbitrary-rclass}.
We introduce a routine \textsc{TensorsWithRClass} which takes an \rclass{} $c$ and returns all the input tensors which have an \agg{}-axis belonging to $c$.
For each tensor, we use the \textsc{NumTensorAccess} routine to get the number of distinct accesses to that tensor.

Thus, $k = \mathsf{max}(m + \sum_{i=1}^q \binom{n_{\mathsf{X}_i}}{2}, 1)$ is a sufficient rank for the \rclass{} $c$ to verify the rewrite rule $R$, where $\{\mathsf{X_1},\cdots,\mathsf{X}_q\} = \textsc{TensorsWithRClass}(c)$ and $n_{\mathsf{X}_i} = \textsc{NumTensorAccess}(\mathsf{X}_i)$.
\end{proof}

\section{Preconditions} \seclabel{preconditions}

In this section, we prove that if the precondition holds in the $(k\!+\!1)$-ranked counterexample, then the precondition holds in the $k$-ranked counterexample, i.e., $C^{k+1} \Rightarrow C^k$.

\subsection{A Language for Preconditions}

We first define a simple language for preconditions and a useful substitution lemma.

\begin{definition}[Precondition language]
\deflabel{precond-lang}
Let $\mathcal{A}$ be a set of atomic propositions (atoms).
We define a language $\Phi_{\wedge,\vee}$ by the grammar:
\[
    \phi \coloneq a \in \mathcal{A} \mid \phi \wedge \phi \mid \phi \vee \phi
\]
Similarly, the language $\Phi_{\wedge}$ is defined by the grammar:
\[
    \phi \coloneq a \in \mathcal{A} \mid \phi \wedge \phi
\]
Note that $\Phi_\wedge \subseteq \Phi_{\wedge,\vee}$.

We write $\subst{\phi}{a}{b}$ for the formula obtained from $\phi$ by simultaneously replacing every occurrence of atom $a$ by atom $b$ (also called a \emph{substitution}).
\end{definition}

\begin{lemma}[Substitution Monotonicity]
\lemmalabel{substitution}
Let $\phi \in \Phi_{\wedge,\vee}$ and let $a,b \in \mathcal{A}$ be atoms.
If the implication $a \Rightarrow b$ holds, then
\[
    \phi \Rightarrow \subst{\phi}{a}{b}
\]
The same holds for $\Phi_{\wedge}$.
\end{lemma}
\begin{proof}
By structural induction on $\phi$.

\noindent \textbf{Basis.} If $\phi = a$, then $\subst{\phi}{a}{b} = b$, and the claim follows from the hypothesis $a \Rightarrow b$.  
If $\phi = c$ for $c \neq a$, then $\subst{\phi}{a}{b} = c$ and $c \Rightarrow c$ holds trivially.

\noindent \textbf{Inductive step.}  
\begin{itemize}
  \item If $\phi = \phi_1 \wedge \phi_2$, then by the induction hypothesis, $\phi_i \Rightarrow \subst{\phi_i}{a}{b}$ for $i = 1, 2$.
  Since $\subst{\phi}{a}{b} = \subst{\phi_1}{a}{b} \wedge \subst{\phi_2}{a}{b}$ and $\phi_1 \wedge \phi_2 \Rightarrow \subst{\phi_1}{a}{b} \wedge \subst{\phi_2}{a}{b}$, we obtain the desired implication.
  \item If $\phi = \phi_1 \vee \phi_2$, similarly the induction hypothesis yields $\phi_1 \vee \phi_2 \Rightarrow \subst{\phi_1}{a}{b} \vee \subst{\phi_2}{a}{b}$.
\end{itemize}
This completes the induction proof.
\end{proof}

% Now, we use \lemmaref{substitution} to preconditions in \project{}.

\subsection{Conditions in \project{}}

The preconditions in \project{} are drawn from the language $\Phi_{\wedge,\vee}$ (\defref{precond-lang}), whose atoms $\mathcal{A}$ consist precisely of all $\fold$-conditions.
The expression validity conditions (such as those returned by $\valexp$) are drawn from $\Phi_{\wedge}$.

The $\fold$ construct takes a predicate and applies it pointwise to a list of maps with the same domains.
It returns \true{} if all predicate values are \true{}, and \false{} otherwise.
It is defined as:
\[
    \fold(g, [m_1, m_2, \cdots]) = \bigwedge_{i \in \mathsf{dom}(m_1)} g(m_1(i), m_2(i), \cdots)
\]

\begin{theorem}
\theoremlabel{precond-proj}
Let $C^{k+1}, C^{k} \in \Phi_{\wedge,\vee}$ be the preconditions obtained from the same symbolic precondition $C$ by instantiating every fold over a \rclass{} $c$ at rank $k\!+\!1$ and at rank $k$, respectively
During the counterexample construction, each $(k\!+\!1)$-fold is projected to the corresponding $k$-fold.
Then,
\[
    C^{k+1} \Rightarrow C^{k}
\]
Similarly, for any tensor expression $e$, $\valexp(e^{k+1}) \Rightarrow \valexp(e^k)$.
\end{theorem}
\begin{proof}
It suffices to show the implication for each atomic fold and then lift it to the whole precondition using \lemmaref{substitution}.

Let $f^{k+1} = \mathsf{fold}(g,\mathcal{N}^{k+1})$ be one of the atomic $\fold$s in $C^{k+1}$, defined over an \agg{}-axis in \rclass{} $c$.
Let $\mathcal{N}_l^{k+1} = [m_1^{k+1}, m_2^{k+1}, \cdots]$.
It has the form:
\[
    \mathsf{fold}(g,\mathcal{N}^{k+1}) = \bigwedge_{i \in \mathsf{dom}(m_1)}^{k+1} g(m_1^{k+1}(i), m_2^{k+1}(i), \cdots)
\]
Let $\Gamma$ be the chosen projection of size $k$ and let $\mathcal{N}^{k}=\cproject{\mathcal{N}^{k+1}}{x}{\Gamma}$ denote the projected maps.
Then:
\[
    \mathcal{N}_l^{k} = \mathcal{N}_l^{k+1}|^x_{\Gamma} = [m_1^{k+1}|^x_{\Gamma}, m_2^{k+1}|^x_{\Gamma}, \cdots] = [m_1^k, m_2^k, \cdots]
\]
The corresponding fold $f^k = \mathsf{fold}(g,\mathcal{N}^{k})$ in $C^k$ has the form:
\[
    \mathsf{fold}(g,\mathcal{N}^{k}) = \bigwedge_{j \in \Gamma} g(m_1^{k}(j), m_2^{k}(j), \cdots)
\]

Since $f^{k+1}$ is a conjunction over a superset of indices, if the larger conjunction if \true{}, then sub-conjunction is \true{}, i.e.:
\[
    f^{k+1} \Rightarrow f^k
\]

Thus every $(k\!+\!1)$-fold atom implies its projected $k$-fold atom.
Applying \lemmaref{substitution} repeatedly to substitute each $(k\!+\!1)$-fold atom by its projected counterpart in the whole precondition implies $C^{k+1} \Rightarrow C^{k}$, as required.
The same holds for validity conditions returned by $\valexp$.
\end{proof}

\section{Discussion} \seclabel{discussion-app}

\paragraph{Rewrite Rules with Reductions}
So far, we have assumed that rewrite rules do not contain any of the following operators:  $\reduce, \tdot, \conv$, and $\convbase$. 
These operators share a common characteristic: they perform \emph{reductions} or  \emph{accumulations} over one or more axes. 
Because the axes being reduced may have unbounded sizes, we cannot simply expand  the reduction and operate over all elements explicitly. Instead, we treat the reduction sum symbolically and provide special handling for such elements during verification. 
Moreover, the user supplies a bijection between the \textsf{LHS} and \textsf{RHS} reduction elements, which allows us to reason about them abstractly.
Once this is done, the rule can be symbolically evaluated as usual, and all of our previous results regarding sufficient rank and counterexample lowering continue to hold.

\paragraph{Counting Non-Trivial Conditions}
We introduced a simplifying assumption (for notational ease) in \secref{scalarf} that every condition in a rule is uniformly represented as a conjunction of one $\fold$ per \agg{}-axis, regardless of whether that \agg{}-axis contributes meaningful conditions.
In practice, some conditions may be \emph{trivial}: their associated fold is degenerate and therefore impose no constraints.
Such conditions should not contribute to the sufficient rank computation.

To account for this, we introduce a routine \textsc{NumConds} which takes an \rclass{} $c$ and returns the number of conditions having an aggregated-axis belonging to the \rclass{} $c$ and whose $\fold$ is non-trivial.
Only these conditions are \emph{relevant} and contribute to the sufficient rank of $c$.

\paragraph{Putting it all Together}
After relaxing all the assumptions from \secref{assumptions}, we come up with the following lemma:
\begin{lemma}\lemmalabel{inferbound-app}
Let $R$ be any rewrite rule written in \dsl{}.
Let $m$ be an \rcrank{} map for $R$ and $c$ be an \rclass{} appearing in $R$.
If $k = \textsc{InferBound}(R, c)$, then
\begin{equation*}
    \forall i \geq k, ~ \mathsf{valid}(m[c \mapsto i]) \rightarrow \mathsf{valid}(m[c \mapsto i\!+\!1])
\end{equation*}
\end{lemma}
In other words, for all $i \geq k$ and for \emph{any} ranks of the other \rclasses{}, if the rule is valid when $c$ has rank $i$, then it implies that the rule is valid when $c$ has rank $i\!+\!1$.
The \textsc{InferBound} routine is described in \algref{infer-bound-app}.

\SetKwComment{Comment}{/*}{ */}
\SetKwInOut{Input}{Inputs}
\SetKwInOut{Output}{Output}
\begin{algorithm}[t]
    \caption{Sufficient Rank for an \rclass{}} \alglabel{infer-bound-app}
    \Input{$\,$ Rewrite rule $R$ \& \rclass{} $c$}
    \Output{$\,$ $bound$, i.e., a sufficient rank for $c$}
    \SetKwProg{Fn}{Function}{ :}{}
    \Fn{\textsc{InferBound} ($R$, $c$)}{
        $bound \gets 0$\;
        $tensors \gets \textsc{TensorsWithRClass}(R, c)$\;
        \For{$t \in tensors$}{
            $n \gets \textsc{NumTensorAccess}(R, t)$\;
            \If{$n > 1$}{
                $bound \gets bound + \binom{n}{2}$\;
            }
        }
        $bound \gets bound + \textsc{NumConds}(R, c)$\;
        \textbf{return} $\textsc{Max}(bound, 1)$ 
    }
    \textbf{End Function}
\end{algorithm}

\paragraph{\textsf{RHS} Expression Validity}
Our definition of rewrite rule validity (\defref{semantic-valid}) states that the \textsf{LHS} and \textsf{RHS} expressions should be equal for all variable valuations, given the precondition $C$ holds and the \textsf{LHS} expression is valid, i.e., $\valexp(\mathsf{LHS})$ holds.
In reality, we require a stronger postcondition that the \textsf{RHS} expression should be valid, i.e., $\valexp(\mathsf{RHS})$ holds.
On taking contrapositives and considering $(k\!+\!1)$- and $k$-ranked instantiations, we get:
\begin{align*}
    & \exists ~ v^{k+1} \in \vars^{k+1}, C^{k+1} \wedge \valexp(\mathsf{LHS}^{k+1}) \\
    &\phantom{\cdots\cdots\cdots\cdots} \wedge \left(\sem{\mathsf{LHS}^{k+1}} \neq \sem{\mathsf{RHS}^{k+1}} \vee \neg\valexp(\mathsf{RHS}^{k+1})\right) \\
    &\phantom{\scalarf(\mathsf{X}^{k+1}[x^{k+1} \mapsto \fmap}{\big\downarrow} \\
    &\qquad\exists ~ v^k \in \vars^k, C^k \wedge \valexp(\mathsf{LHS}^{k}) \\
    &\phantom{\cdots\cdots\cdots\cdots} \wedge \left(\sem{\mathsf{LHS}^{k}} \neq \sem{\mathsf{RHS}^{k}} \vee \neg\valexp(\mathsf{RHS}^{k})\right)
\end{align*}

Thus there are two distinct kinds of counterexamples to lower from rank $k\!+\!1$ to $k$:
\begin{itemize}
    \item Mismatch counterexamples:
    If $\sem{\mathsf{LHS}} \neq \sem{\mathsf{RHS}}$ in rank $k\!+\!1$, then we lower to a witness of $\sem{\mathrm{LHS}} \neq \sem{\mathrm{RHS}}$ in rank $k$.
    This is handled in \secref{proof}: the sufficient rank given by \algref{infer-bound-app} guarantees a counterexample lowering for this case.

    \item \textsf{RHS} invalidity counterexamples:
    If $\neg\valexp(\mathsf{RHS})$ holds in rank $k\!+\!1$, then we lower to a witness of $\neg\valexp(\mathsf{RHS})$ in rank $k$.
    The condition $\valexp(\mathsf{RHS})$ lies in the fragment $\Phi_\wedge$, i.e., it is a conjunction of $\fold$-atoms.
    % Thus, $\neg\valexp(\mathsf{RHS})$ is a disjunction of negated $\fold$-atoms.
    As in the condition-analysis (\lemmaref{arbitrary-cond}), any $(k\!+\!1)$-ranked falsifying $\fold$-conjunction contains at least one $\fold$, say $f^{k+1}$, whose truth value if \false{}.
    To ensure that $\valexp(\mathsf{RHS}^{k})$ is \false{}, we ensure that $f^k$ is \false{}.
    This can be done by choosing an appropriate \ind{}-axis, giving a sufficient rank of 1.
\end{itemize}

For worst–case bounds we take the maximum of the two cases.
Since each branch is already at least 1, the sufficient rank given by \algref{infer-bound-app} remains unchanged.

\paragraph{Input Tensor Aliasing in Tensor Graph Rewrites}

Our analysis in \lemmaref{arbitrary-tensor} assumes that distinct input tensors do not alias.
For example, in the tensor expression $e = \binary(+,\mathsf{X}, \mathsf{Y})$, we assume that \textsf{X} and \textsf{Y} are different tensors.
In such a case, during counterexample construction, accesses to different tensors are treated as independent, and no cross-tensor constraints arise.
Therefore, the sufficient rank decomposes cleanly into the sum of per-tensor contributions.

However, a compiler matching a symbolic expression $e$ inside a tensor computation graph, may encounter a match such as $\binary(+, t, t)$, where both operands refer to the same concrete tensor.
This breaks the independence assumption: accesses that originate syntactically from different tensors may in fact refer to the same underlying source, thereby introducing additional cross-tensor interactions.
To account for such cases, (1) one can verify a different rewrite capturing the aliasing relation, or (2) adjust the sufficient rank to account for additional cross-tensor interactions and constraints introduced from them.
More specifically, for tensors $\mathsf{X}_1, \cdots, \mathsf{X}_q$ with $n_{\mathsf{X}_1}, \cdots, n_{\mathsf{X}_q}$ number of accesses, respectively, the new constraints are bounded by $\binom{\sum_{i}n_{\mathsf{X}_i}}{2}$, instead of $\sum_i \binom{n_{\mathsf{X}_i}}{2}$.

\end{document}